\newif\ifsingle
\singletrue 
\ifsingle
\documentclass[journal,onecolumn]{IEEEtran}
\else
\documentclass[conference,letterpaper]{IEEEtran}
\fi
\IEEEoverridecommandlockouts
\usepackage{cite}
\usepackage{amsmath,amssymb,amsfonts}
\usepackage{amsthm}
\usepackage{hyperref}
\hypersetup{colorlinks=true, unicode=true, linkcolor=[rgb]{0.10,0.05,0.67}, citecolor=[rgb]{0.10,0.05,0.67}, filecolor=[rgb]{0.10,0.05,0.67}, urlcolor=[rgb]{0.10,0.05,0.67}}
\usepackage{cleveref}
\usepackage{graphicx}
\usepackage{bbold}
\usepackage{textcomp}
\usepackage{xcolor}
\usepackage{algorithm}
\usepackage{algpseudocode}
\usepackage{subcaption}
\newtheorem{theorem}{Theorem}
\newtheorem{lemma}{Lemma}
\newtheorem{claim}{Claim}
\newtheorem{prop}{Proposition}

\newtheorem{remark}{Remark}

\usepackage{acronym}
\usepackage{comment}
\newcommand{\off}[1]{}

\usepackage{mathtools}
\DeclarePairedDelimiter{\abs}{\lvert}{\rvert}

\DeclarePairedDelimiter\floor{\lfloor}{\rfloor}
\newcommand{\paren}[1]{\left( #1 \right)}
\newcommand{\cbrace}[1]{\left\{#1\right\}}
\newcommand{\sbrace}[1]{\left[#1\right]}

\newcommand{\calE}{\mathcal{E}}
\newcommand{\calF}{\mathcal{F}}
\newcommand{\calG}{\mathcal{G}}
\newcommand{\calI}{\mathcal{I}}

\newcommand{\calP}{\mathcal{P}}
\newcommand{\calT}{\mathcal{T}}

\newcommand{\allnchoosei}{\sbrace{\nchoosei}}
\newcommand{\allnchoosek}{\sbrace{\nchoosek}}
\newcommand{\allnkchoosei}{\sbrace{\nkchoosei}}
\newcommand{\der}[2]{\frac{\partial}{\partial #1}\paren{#2}}
\newcommand{\derrho}[1]{\der{\rho}{#1}}
\newcommand{\fact}[1]{\paren{#1}!}
\newcommand{\kchoosei}{\binom{K}{i}}
\newcommand{\nchoosei}{\binom{N}{i}}
\newcommand{\nchoosek}{\binom{N}{K}}
\newcommand{\nchooseki}{\binom{N}{K-i}}

\newcommand{\nkchoosei}{\binom{N-K}{i}}
\newcommand{\nkichoosei}{\binom{N-\paren{K-i}}{i}}
\newcommand{\pe}{P_{e}}
\newcommand{\pr}[1]{P\paren{#1}}
\newcommand{\x}{\mathbf{X}}
\newcommand{\yt}{Y^{T}}
\newcommand{\repsfact}{\frac{1}{\binom{K}{i}}}

\newcommand{\ftilde}{\tilde{f}}

\newcommand{\calItilde}{\tilde{\calI}}
\newcommand{\calItildejdiff}{\calItilde^{j}_{dif}}

\newcommand{\decode}{\decodersym\paren{\yt}}
\newcommand{\decodersym}{g}
\newcommand{\ei}{E_{i}}
\newcommand{\eij}{\ei^{j}}
\newcommand{\errevent}{\zeta_{\wprime}}
\newcommand{\esj}{E_{s,j}\paren{\rho, \pw{\cdot}}}
\newcommand{\fn}{FN}
\newcommand{\fnfull}[2]{\fn\paren{#1,#2}}
\newcommand{\fp}{FP}
\newcommand{\fpfull}[2]{\fp\paren{#1,#2}}
\newcommand{\fpset}[2]{\calF\calP_{#1}\paren{#2}}
\newcommand{\fpseti}[1]{\fpset{i}{#1}}

\newcommand{\fpsetis}{\fpseti{S_{1}}}
\newcommand{\fpsetijcounter}{S_{2}\in\fpsetis}
\newcommand{\fraconerho}{\frac{1}{1+\rho}}
\newcommand{\ftilderhos}{\ftilde\paren{\rho,\Tset}}
\newcommand{\gesym}{\calG\calE}
\newcommand{\gen}{\gesym_{N}}
\newcommand{\genk}{\gesym_{N,K}}
\newcommand{\ith}{U_{i}}

\newcommand{\locglobsym}{\Omega}
\newcommand{\locglob}[2]{\locglobsym_{#1}\paren{#2}}
\newcommand{\numstatesmc}{2\paren{N+1}}
\newcommand{\pw}[1]{P_{W}\paren{#1}}

\newcommand{\sw}{S_{\w}}
\newcommand{\swprime}{S_{\w^{\prime}}}
\newcommand{\swstar}{S_{\w^{\star}}}
\newcommand{\totalones}{\abs{\gen}}
\newcommand{\tp}{TP}
\newcommand{\tpfull}[2]{\tp\paren{#1,#2}}
\newcommand{\tpset}[1]{\calT\calP_{#1}}
\newcommand{\tpseti}{\tpset{i}}

\newcommand{\tpsetijcounter}{S_{1}\in\tpseti}
\newcommand{\w}{\omega}
\newcommand{\wprime}{\w^{\prime}}
\newcommand{\wstar}{\w^{*}}

\newcommand{\xfnwprime}{\x_{\fn}}
\newcommand{\xfpwprime}{\x_{\fp}}
\newcommand{\xfpwprimevec}{X_{\fp}}
\newcommand{\xsw}{\x_{\sw}}
\newcommand{\xswprime}{\x_{\swprime}}
\newcommand{\xswstar}{\x_{\swstar}}
\newcommand{\xtpwprime}{\x_{\tp}}
\newcommand{\xtpwprimevec}{X_{\tp}}
\newcommand{\Tset}{S_1}
\newcommand{\Fset}{S_2}

\acrodef{cs}[CS]{Compressed Sensing}
\acrodef{dd}[DD]{Definitely Defectives}
\acrodef{dnd}[DND]{Definitely Not Defective}
\acrodef{fn}[FN]{false negatives}
\acrodef{fp}[FP]{false positives}
\acrodef{ge}[GE]{Gilbert-Elliott}
\acrodef{gt}[GT]{Group Testing}
\acrodef{hmm}[HMM]{Hidden Markov Model}
\acrodef{lhs}[LHS]{left hand side}
\acrodef{lva}[LVA]{List Viterbi Algorithm}
\acrodef{mac}[MAC]{multiple access channel}
\acrodef{map}[MAP]{maximum a-posteriori}
\acrodef{ml}[ML]{maximum likelihood}
\acrodef{msgt}[2SDGT]{Two-Stage Decoding GT}
\acrodef{rhs}[RHS]{right hand side}
\acrodef{tn}[TN]{true negatives}
\acrodef{tp}[TP]{true positives}
\acrodef{va}[VA]{Viterbi Algorithm}

\usepackage{amsmath}
\allowdisplaybreaks

\def\BibTeX{{\rm B\kern-.05em{\sc i\kern-.025em b}\kern-.08em
    T\kern-.1667em\lower.7ex\hbox{E}\kern-.125emX}}

\usepackage{ifthen}
\newboolean{full_version}
\setboolean{full_version}{true}
\begin{document}

\title{\textcolor{black}{Two-Stage Decoding Algorithm and Bounds\\ for Group Testing with Prior Statistics}\vspace{-0.2cm}}
\author{
   \IEEEauthorblockN{Ayelet C. Portnoy\IEEEauthorrefmark{1},
                     Amit Solomon\IEEEauthorrefmark{2},
                     and Alejandro Cohen\IEEEauthorrefmark{1}}\\
   \IEEEauthorblockA{\IEEEauthorrefmark{1}%
                      Faculty of ECE, Technion, Haifa, Israel, Emails: ayeletco@campus.technion.ac.il, and alecohen@technion.ac.il}
    \IEEEauthorblockA{\IEEEauthorrefmark{2}%
                     Princeton University, Princeton, NJ, USA, Email: as3993@princeton.edu}
\vspace{-0.57cm}}

\maketitle
\begin{abstract}
In this paper, we propose an efficient \textcolor{black}{two-stage decoding} algorithm for non-adaptive Group Testing (GT) with general correlated prior statistics. The proposed solution can be applied to any correlated statistical prior represented in trellis, e.g., finite state machines and Markov processes. We introduce a variation of List Viterbi Algorithm (LVA) to enable accurate recovery using much fewer tests than objectives, which efficiently gains from the correlated prior statistics structure. We also provide a sufficiency bound to the number of pooled tests required by any Maximum A Posteriori (MAP) decoder with an arbitrary correlation\textcolor{black}{, i.e., dependence between infected items}. Our numerical results demonstrate that the proposed \textcolor{black}{two-stage decoding GT (2SDGT)} algorithm can obtain the optimal MAP performance with feasible complexity in practical regimes, such as with COVID-19 and sparse signal recovery applications, and reduce in the scenarios tested the number of pooled tests by at least $25\%$ compared to existing classical low complexity GT algorithms. Moreover, we analytically characterize the complexity of the proposed \textcolor{black}{2SDGT} algorithm that guarantees its efficiency.

\end{abstract}


\section{Introduction}

Classical \ac{gt} aims to detect a small number of ``defective" items within a large population by mixing samples into as few pooled tests as possible. The idea of \ac{gt} was first introduced during World War II when it was necessary to discover soldiers infected with Syphilis. Dorfman \cite{dorfman1943detection} showed that the required number of tests could be reduced if multiple blood samples were grouped into pools.
When the samples that participate in the next pool are selected iteratively based on the previous pool test results, the \ac{gt} algorithm is called adaptive. In contrast, in non-adaptive \ac{gt}, the whole process is designed in advance.
Since it was first suggested, the \ac{gt} problem has been investigated and generalized to many areas and applications, among them disease detection \cite{goenka2021contact, srinivasavaradhan2022dynamic,cohen2021multi,solomon2025one}, cyber security applications \cite{thai2008detection}, pattern matching algorithms \cite{indyk1997deterministic} \off{, image compression \cite{hong2002group}} and communication \cite{wolf1985born, cohen2019serial}.

All of these applications imply a strong connection between \ac{gt} and \textcolor{black}{Compressed Sensing} as two methods for sparse signal recovery that share common applications \cite{gilbert2008group,tan2014strong,aksoylar2016sparse,naderi2022sparsity}. The main difference between the two is that \textcolor{black}{Compressed Sensing} aims to recover a real-valued signal \cite{eldar2012theory}, while \ac{gt} recovers a binary signal \cite{aldridge2019group} or discrete-values \cite{cohen2020efficient,cohen2021serial,cohen2021multi}. Thus, one can consider \ac{gt} as a Boolean \textcolor{black}{Compressed Sensing} \cite{atia2012boolean, sejdinovic2010note}.

Traditional \ac{gt} and its performance (i.e., the tradeoff between the number of \textcolor{black}{pooled tests} and recovery algorithm complexity), focuses on the probabilistic model in which the items are identically distributed and independent \cite{aldridge2014group,aldridge2019group}. Recent research explores cases where prior information about the correlation \textcolor{black}{(i.e., dependence)} of objects is available \cite{nikolopoulos2021groupBITS,goenka2021contact, srinivasavaradhan2022dynamic, lau2022model,arasli2023group}. The motivation for this approach arises from the fact that correlated prior statistics have the potential to achieve higher recovery rates and reduce the number of required tests.
In disease detection, leveraging information about the proximity between individuals, represented by contact tracing information or graphs, can lead to significant savings in \textcolor{black}{pooled tests} \cite{nikolopoulos2021groupBITS, goenka2021contact, srinivasavaradhan2022dynamic, lau2022model}. However, previous \ac{gt} works presented solutions designed for specific models and applications and \textcolor{black}{may not extend easily} to other models and applications. In numerous signal processing applications, correlation between different frequencies, time signals, or among different sensors can also be utilized to achieve more precise estimations \cite{zhang2011sparse,eldar2012theory,7868430}. \textcolor{black}{Hidden Markov Model} is a common model for many physical signals, such as speech signals \cite{gales2008application}, human motion \cite{su2023latency}, and spectrum occupancy in communication systems \cite{ghosh2009markov}. Infections can be also modeled as \textcolor{black}{Hidden Markov Model}s \cite{ceres2020characterizing}. For example, \cite{cao2023group} presents a \ac{gt} solution for a specific \textcolor{black}{Hidden Markov Model} derived from contact tracing. To the best of our knowledge, no existing solution addresses the \ac{gt} problem with general Markovian priors and applicable to a wide range of diverse applications.

In this work, we \textcolor{black}{addresses the \ac{gt} problem with general prior statistics.  We adopt a more general formulation in which the number of defective items is not assumed to be fixed a priori, but is instead induced by the underlying correlated prior distribution. In particular, the number of defectives is treated as a random variable, and the system design is based on a typical cardinality interval of the number of defective items. We introduce} a practical non-adaptive \ac{msgt} algorithm for correlated items with prior statistics. The proposed \textcolor{black}{two-stage decoding} algorithm employs a new variation of the parallel \ac{lva} \cite{andrew13viterbi,lou1995implementing,seshadri1994list} we designed for \ac{gt} to enable accurate low complexity recovery using fewer tests. \textcolor{black}{The proposed algorithm can be applied to general Markovian correlated priors, which may be represented using trellis structures \cite{moon2020error}, e.g., finite-state machines and Markov processes.}  \textcolor{black}{While classical group testing assumes exchangeable items, many practical settings admit structured dependencies arising from spatial, temporal, or graph-based relationships. Rather than prescribing a specific physical model, this work adopts a general trellis-based prior as a unifying abstraction for correlated item statistics. This formulation subsumes the i.i.d.\ model as a special case and enables a single decoding framework to operate across diverse applications where correlation information is available or can be estimated.
}\textcolor{black}{ The use of the \ac{lva} is motivated by the observation that, under the Markovian prior model, the population sequence admits a natural trellis representation, and that approximate MAP decoding over this trellis can be performed efficiently by enumerating a small list of the most likely trajectories rather than a single path.} Using \ac{lva}, \ac{msgt} leverages those statistics to estimate the defective set efficiently, even in a regime below the \ac{ml} upper bound. Furthermore, we show how the algorithm's parameters can be tuned to achieve a maximum probability of success without exceeding the limitation of the available computational capacity. We derive a lower bound that considers the exact priors of the problem and provides analytical results that characterize the \ac{msgt} computational complexity efficiency.
We provide an analytical sufficiency bound for the number of pooled tests needed by any \ac{map} decoder. This bound holds for any \textcolor{black}{dependence} between infected items. It also applies to multiple access communication systems, where $K$ out of $N$ users transmit information simultaneously. This is considered in \cite{wu2015partition,wu2014achievable,cohen2020efficient} without prior statistical information. In contrast, existing analyses with prior statistics assume all $N$ users transmit at the same time \cite{slepian1973coding,zhong2006joint,zhong2007joint,campo2011random,campo2012achieving,rezazadeh2019joint,rezazadeh2019error}.
Our numerical results demonstrate that in practical regimes for COVID-19 \cite{lucia2020ultrasensitive,ben2020large} and sparse signal recovery in signal processing \cite{gilbert2008group,tan2014strong,aksoylar2016sparse,eldar2012theory}, the low computational complexity \ac{msgt} algorithm proposed herein can reduce in the scenarios tested the number of \textcolor{black}{pooled tests} by at least $25\%$.

The rest of this paper is organized as follows. Section~\ref{sec_problem_form} formally describes the \ac{gt} model with correlated prior statistics. Section~\ref{sec_main_res} presents the \ac{msgt} algorithm and the analytical results, and Section~\ref{sec:theoretical_analysis} presents a sufficiency bound for \ac{map} decoder. Section~\ref{sec_numerical_eval} details the simulation evaluation. Finally, Section~\ref{sec_discussion} provides concluding remarks and future directions.

\section{Problem Formulation}\label{sec_problem_form}
Given a set of individuals $\mathcal{N}$, the objective in GT is to discover a small \textcolor{black}{random} subset $\mathcal{K}$ of unknown defective items using the minimum number of measurements $T$. \textcolor{black}{Let $N=|\mathcal{N}|$ denote the total number of items}\off{, where $K=\mathcal{O}(N^\theta)$ for some $\theta\in[0,1)$}. Let $\mathcal{U}=\{0,1\}$ be the state alphabet, and let $\mathbf{U}=(U_1,\ldots,U_N)\in\mathcal{U}^N$ represent the \textcolor{black}{random population vector}, where $U_i=1$ indicates that the $i$-th item is defective.\off{We assume a sparse regime with $\theta\le 1/3$ \cite{atia2012boolean,aldridge2014group,scarlett2017little,aldridge2019group}, and each item is in one of two states: defective or non-defective.} \textcolor{black}{Under this model, the random defective set is $\mathcal{K} \triangleq \{ i\in\mathcal{N} : U_i=1 \}$  and the number of defective items is the random variable $K \triangleq |\mathcal{K}|=  \sum_{i=1}^N U_i$. We denote by $\mu \triangleq \mathbb{E}[K] = \sum_{i=1}^N \Pr(U_i=1)$ the expected number of defective items.}

\textcolor{black}{
We model the prior over the sequence as a $\tau$-order Markov process.
The joint law of the first $\tau$ variables is specified by an initial distribution $\pi(u_1,\ldots,u_\tau)\triangleq P(U_1=u_1,\ldots,U_\tau=u_\tau)$, for $(u_1,\ldots,u_\tau)\in\mathcal{U}^\tau$.
\textcolor{black}{This initial distribution generalizes the marginal parameters $\{\pi_i\}$ stated earlier: for $\tau=1$ it reduces to $\pi(u_1)=P(U_1=u_1)$ (equivalently, $\pi_1(u_1)$), whereas for $\tau>1$ the model requires specifying the joint initialization of the first $\tau$ states.}
For $i\ge \tau+1$, the dependence on the previous $\tau$ states is specified by the transition kernels $\{\Phi_i\}_{i=\tau+1}^N$, defined by
\[
\Phi_i\!\left[u_i \mid u_{i-\tau},\ldots,u_{i-1}\right]
\triangleq
P\!\left(U_i=u_i \,\big|\, U_{i-1}=u_{i-1},\ldots,U_{i-\tau}=u_{i-\tau}\right).
\]
}
\textcolor{black}{
Equivalently, the transition of the $i$-th item can be represented in matrix form by
\[
\boldsymbol{\Phi}_i \in [0,1]^{2^\tau\times 2},\quad \text{with transition probability} \quad
\boldsymbol{\Phi}_i[\ell,k]
= P\!\big(U_i=s_k \,\big|\, (U_{i-\tau},\ldots,U_{i-1})=s_\ell\big),
\]
where $s_\ell\in\{0,1\}^\tau$ is the $\tau$-bit binary representation of $\ell\in\{0,\ldots,2^\tau-1\}$ and $s_k\in\{0,1\}$ (see Fig.~\ref{fig:phi_example}).
The distribution of the entire chain is then
\[
P\left(U_1=u_1,\ldots,U_N=u_N\right)
= \pi(u_1,\ldots,u_\tau)\prod_{i=\tau+1}^{N}\Phi_i\!\left[u_i \mid u_{i-\tau},\ldots,u_{i-1}\right].
\]
When there is no memory ($\tau=0$), the process reduces to independent (but not necessarily identical) marginals, and the prior information is fully captured by the marginal probabilities
\[
P(U_1=u_1,\ldots,U_N=u_N)
=
\prod_{i=1}^N P(U_i=u_i).
\]
}

\begin{figure}
    \centering
    \includegraphics[width=0.5\linewidth]{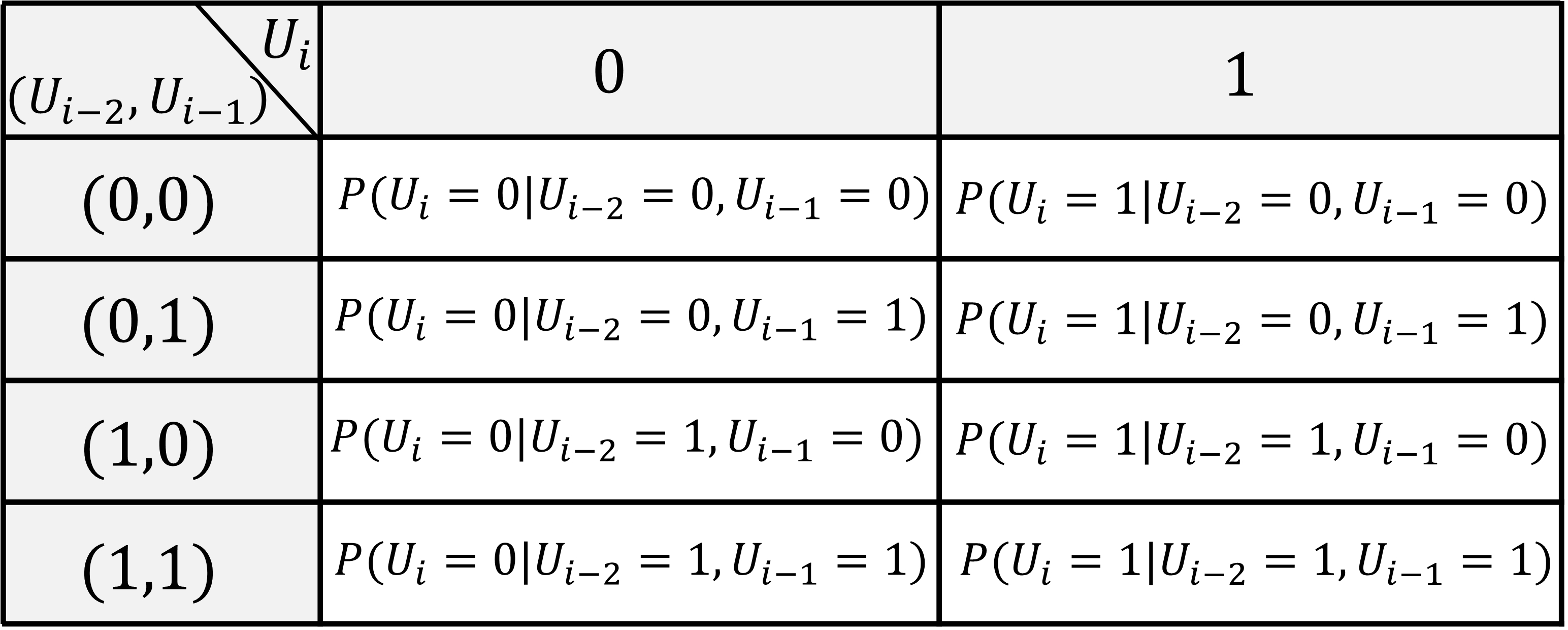}
    \caption{Correlation \textcolor{black}{(i.e., dependence of objects)} matrix $\mathbf{\Phi}_i$ for $\tau=2$.}
    \label{fig:phi_example}
\end{figure}

\textcolor{black}{
Instead of conditioning on a fixed realization of $K$ \cite{aldridge2019group}, we introduce \textcolor{black}{design sparsity bounds $K_{\min}$ and $K_{\max}$} that \textcolor{black}{are used} for test-matrix design and algorithmic thresholds, and \textcolor{black}{are} not assumed to equal the realized value of $K$.
\textcolor{black}{A convenient choice is $K_{\min} \triangleq \lfloor(1-\varepsilon)\mu\rfloor$ and $K_{\max} \triangleq \lceil (1+\varepsilon)\mu\rceil$, for some $\varepsilon\in(0,1)$.}
In subsequent sections, the recovery performance is analyzed by conditioning on the typical event $\textcolor{black}{\{K_{\min}\le K\le K_{\max}\}}$, while accounting separately for the tail probability \textcolor{black}{$\Pr(K<K_{\min})+\Pr(K>K_{\max})$} via concentration bounds on $K$. However, it is important to emphasize that the random variable $K$ is not restricted a priori and may take any value in $\{0,\ldots,N\}$. We focus on sparse regimes with $\mu=\mathcal{O}(N^\theta)$ for some $\theta\in[0,1)$, and in particular $\theta\le 1/3$ \cite{atia2012boolean,aldridge2014group,scarlett2017little,aldridge2019group}.
}

For the non-adaptive GT, the testing matrix $\mathbf{X}\in\left\{ 0,1\right\} ^{T\times N}$ is defined such that each row corresponds to a single pool test, and each column corresponds to a single item. That is, the $i$-th pool test is represented as a binary row vector:
$X_{i}=\left[X_{i}(1),...,X_{i}(N)\right],\,i\in\left\{ 1,...,T\right\}$ whose elements are defined: $X_{i}(j)=1$ if the item with an index $j\in\left\{ 1,...,N\right\}$ is included in the $i$-th pool test, and otherwise $X_{i}(j)=0$. Then, the outcome of the $i$-th pool test is given by
\[
Y_{i}=\bigvee_{j\in\mathcal{K}}X_{i}(j)=\bigvee_{j=1}^{N}X_{i}(j)U_{j},
\]
where $\bigvee$ is the Boolean OR function.

Given $\mathbf{X}$ and the outcome vector $Y$, the recovery success criterion in \ac{gt} can be measured using various metrics \cite{aldridge2019group}. The main metrics we will use herein are exact recovery and partial recovery.
\textcolor{black}{We note that the trellis (Markov) prior adopted herein should be interpreted as a modeling abstraction rather than as an assumption of a fixed physical ordering of items. The index $i$ does not necessarily correspond to spatial or temporal adjacency; instead, it represents an ordering induced by latent structure such as clustering, spatial proximity, temporal evolution, graph embeddings, or other side information available to the designer. When the items are exchangeable and no correlation information is available, the proposed model naturally reduces to the independent prior as a special case ($\tau=0$).
When correlations are present, an appropriate permutation of item indices allows the dependency structure to be represented using a low-order Markov model. Under such correlated priors, the test matrix need not be permutation-invariant, as exploiting prior information inherently breaks symmetry across items. This behavior is not a limitation of the method but rather a consequence of leveraging structured side information to improve recovery performance. The proposed framework does not assume that such priors are always observable in practice; instead, it provides a general recovery architecture capable of exploiting correlated priors whenever they are available or can be inferred from data.}

\textcolor{black}{
We consider the prior probability of a defective set $\mathcal{K}$ through the population vector $\mathbf{U}=\mathbf{1}_{\mathcal{K}}$: $P(\mathcal{K}) \ \triangleq\ P(\mathbf{U}=\mathbf{1}_{\mathcal{K}})$,
where $P(\mathbf{U})$ is induced by the $\tau$-order Markov model described above in this section. Here $P(\mathcal{K})$ denotes the probability of the defective set induced by the prior distribution on the population vector $\mathbf{U}$, without conditioning on a fixed sparsity level. Note that this probability is defined over all subsets $\mathcal K\subseteq\mathcal N$ and does not restrict the cardinality $|\mathcal K|$.
}

In terms of exact recovery, the goal is to detect the precise subset of defective items $\mathcal{K}$. Therefore, given the estimated defective set \textcolor{black}{$\hat{\mathcal{K}}=\hat{\mathcal{K}}\!\left(N,K_{\min}, K_{\max},\mathbf{X},Y,\mathbf{\pi}, \{\mathbf{\Phi}_i \}_{i=\tau+1}^N\right)$}, we define the average exact error probability by\footnote{For simplicity of notation, $P_s$ and $P_e=1-P_s$ refer to success and error probabilities in the exact recovery analysis.}

\textcolor{black}{
\[
P_e^{\mathrm{exact}} \triangleq
\sum_{\mathcal K \subseteq \mathcal N} P(\mathcal K)\;
\Pr\!\left(\hat{\mathcal K} \neq \mathcal K \,\big|\, \mathcal K\right)
=
\Pr\!\left(\hat{\mathcal K} \neq \mathcal K\right).
\]
Moreover, for any design sparsity bounds $(K_{\min},K_{\max})$,
\begin{equation}\label{eq:p_exact}
P_e^{\mathrm{exact}}
\le
\Pr\!\left(\hat{\mathcal K} \neq \mathcal K \mid \textcolor{black}{K_{\min}\le K \le K_{\max}}\right)
\textcolor{black}{+
\Pr(K<K_{\min}) + \Pr(K>K_{\max}).}
\end{equation}
}
\textcolor{black}{Here $(K_{\min},K_{\max})$ are design parameters rather than the realized sparsity $K$.
}

In partial recovery, we allow a false positive (i.e., $|\hat{\mathcal{K}}\setminus \mathcal{K}|$) and false negative (i.e., $|\mathcal{K}\setminus \hat{\mathcal{K}}|$) detection rate.
\textcolor{black}{
Since \( \hat{\mathcal{K}} \) is a random variable, the overlap term \( {|\hat{\mathcal{K}} \cap \mathcal{K}|}/{|\mathcal{K}|} \) is also a random variable.
To define a deterministic performance metric, we take the expectation over the randomness of \( \hat{\mathcal{K}} \), conditioned on the true defective set \( \mathcal{K} \). }
Thus, the average partial success rate is defined as
\textcolor{black}{
\[
P_s^{\mathrm{partial}} \triangleq
\sum_{\mathcal K \subseteq \mathcal N} P(\mathcal K)\;
\mathbb{E} \left[\left.{|\hat{\mathcal K} \cap \mathcal K|}/{|\mathcal{K}|}\,\right|\mathcal K\right].
\]
}
To conclude, \textcolor{black}{for a sparse subset of infected items with random size $K$ (with $\mathbb{E}[K]=\mu$) out of $N$,} the goal in non-adaptive \ac{gt} with correlated prior data items is to design a $T \times N$ testing matrix and an efficient and practical recovery algorithm that can exploit correlated priors, such that by observing $Y^T$ we can identify the subset of infected items with high probability and with feasible computational complexity. Thus, given the knowledge of the correlated prior data items and the available computational resources, the test designer could design the testing matrix and a recovery algorithm to maximize the success probability.

\section{Main Results}\label{sec_main_res}

In this section, we introduce the efficient two-stage\footnote{\textcolor{black}{Although we refer to the decoding procedure as consisting of two stages, the testing process itself is fully non-adaptive: the considered testing matrix is fixed in advance (see Section~\ref{subsec:algo_description}), and no additional tests are designed based on intermediate decoding results. The term ``stage'' is used here only to distinguish between successive computational steps of the proposed decoder herein.}} recovery algorithm for any statistical prior represented in a trellis diagram \cite{moon2020error}, detailed in Algorithm~\ref{alg:cap}. In the first stage, standard low-complexity algorithms \cite{aldridge2014group} reduce the search space independently of prior correlations. This reduction is guaranteed by new analytical results we derive. In the second stage, the algorithm employs a novel adaptation of the \ac{lva} \cite{seshadri1994list}, designed for \ac{gt} to enable accurate low-complexity recovery using fewer tests by exploiting the correlated prior information. Additionally, we derive a bound to ensure the low complexity of the entire algorithm. \textcolor{black}{
Under the correlated prior (See Section~\ref{sec_problem_form}), the number of defectives
$K=\sum_{i=1}^N U_i$ is a random variable. The algorithm is \emph{not} assumed to know the realized $K$.
Whenever a step requires a cardinality (filtering the \ac{lva} output and enumerating candidate supports), we use \textcolor{black}{$(K_{\min},K_{\max})$} as the design \textcolor{black}{bounds}. Performance under random $K$ is handled by conditioning on the typical event \textcolor{black}{$\{K_{\min}\le K \le K_{\max}\}$} and adding the tail probability \textcolor{black}{$\Pr(K<K_{\min})+\Pr(K>K_{\max})$}, as defined in \eqref{eq:p_exact}  and elaborated in Section~\ref{subsec:main_results_discussion}.
} \textcolor{black}{ A key design parameter of the proposed 2SDGT algorithm presented in this section is the constant $\gamma \geq 1$, which controls the maximal candidate support size admitted from the LVA stage to the final MAP decoding stage. Specifically, only trajectories whose estimated number of defective items lies within the interval $[K_{\min}, \lfloor \gamma K_{\max} \rfloor]$ are retained for candidate generation. The parameter $\gamma$ therefore directly governs the tradeoff between computational complexity and recovery performance: larger values of $\gamma$ increase robustness to model mismatch and estimation errors at the cost of increased candidate enumeration, while smaller values enforce stricter complexity constraints.} Section~\ref{subsec:algo_description} describes the proposed algorithm. Section~\ref{subsec:analytical_results} provides analytical results, followed by a discussion in Section ~\ref{subsec:main_results_discussion}. \ifthenelse{\boolean{full_version}}
{
For a detailed explanation of all the algorithms used as integral components of Algorithm~\ref{alg:cap}, see Appendix~\ref{appendix:recovery_algo}.
}{
We refer the reader to \cite[Appendix A]{portnoy2024multi} for a detailed explanation of all the algorithms used as integral components of Algorithm~\ref{alg:cap}.
}

\subsection{Pool-Testing Algorithm} \label{subsec:algo_description}
\subsubsection{Testing Matrix and Pooling}
The proposed \textcolor{black}{two-stage} recovery algorithm is intended to work with any non-adaptive testing matrix, e.g., as given in \cite{johnson2018performance}. To simplify the technical aspects and focus on the key methods, the testing matrix is generated randomly under \textcolor{black}{design sparsity bounds $(K_{\min},K_{\max})$ (as defined in Section~\ref{sec_problem_form})} with Bernoulli distribution of \textcolor{black}{$p=\ln(2)/K_{\max}$} \cite{aldridge2017almost}, using classical GT methods. The pooling and its outcome are given by the process elaborated in Section~\ref{sec_problem_form} and illustrated in Fig.~\ref{fig:MSGT_exmaple:a}.
\textcolor{black}{ We emphasize that the testing matrix is intentionally chosen to be a standard Bernoulli design, so that the impact of exploiting correlated prior information is isolated to the decoding stage\footnote{\textcolor{black}{An interesting direction for future work is the design of non-adaptive testing matrices that explicitly exploit non-uniform or correlated priors, for example via item-dependent testing probabilities.}}.}

\begin{figure}[ht]
    \centering
    \ifsingle
    \includegraphics[width=0.35\linewidth]{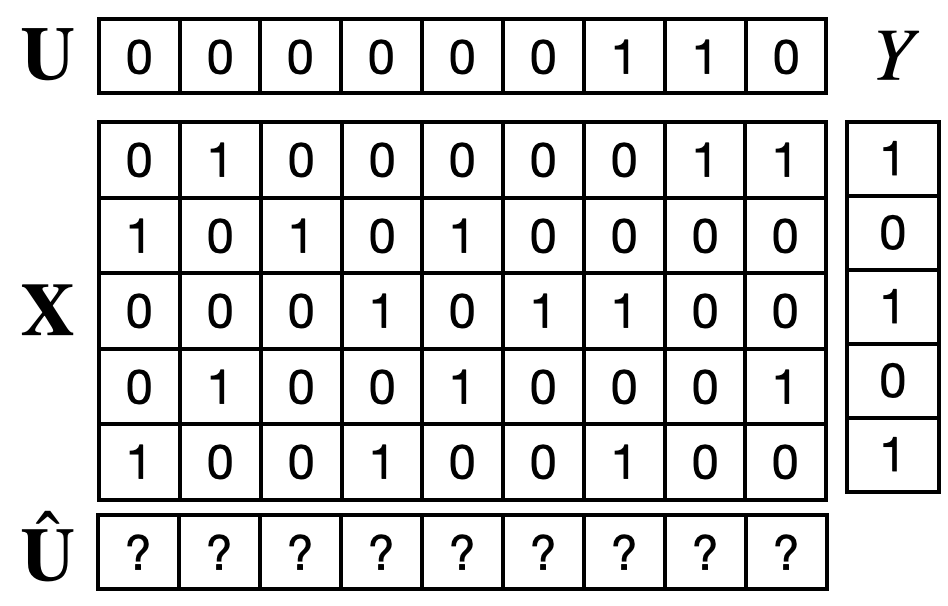}
    \else
    \includegraphics[width=0.49\linewidth]{MSGT9_pooling.png}
    \fi
    \caption{\small For an unknown population $\mathbf{U} \in \left\{0,1\right\}^9 $ \textcolor{black}{a realized defective set of size $K=2$ and design interval $[K_{\min},K_{\max}]=[2,2]$}, a random testing matrix is sampled and the test result $Y$ is calculated.}
    \label{fig:MSGT_exmaple:a}
\end{figure}

\begin{algorithm}\small
\caption{\textcolor{black}{Two-Stage} Recovery Algorithm}\label{alg:cap}
\hspace*{\algorithmicindent} \textbf{Input: $\mathbf{X},Y, \textcolor{black}{\mathbf{\pi}, \{\mathbf{\Phi}_i \}_{i=\tau+1}^N } , \textcolor{black}{K_{\min},K_{\max}}, L, \tau, \gamma$} \\
\hspace*{\algorithmicindent} \textbf{Output: $\hat{\mathcal{K}}$}
\begin{algorithmic}[1]
    \Statex\textbf{\underline{Stage 1}: Reduction of space search}
    \vspace{0.1cm}
    \State $\mathcal{P}^{(DND)} \gets \mathbf{DND}(\mathbf{X}, Y)$  \label{alg:line:coma}
    \State $\mathcal{P}^{(DD)} \gets \mathbf{DD}\left(\mathbf{X}, Y, \mathcal{P}^{(DND)} \right)$ \label{alg:line:dd}
    \State $\mathcal{P}^{(S_{1,2})} \gets \mathcal{N} \setminus \left( \mathcal{P}^{(DND)} \cup \mathcal{P}^{(DD)} \right) $
    \Statex\textbf{\underline{Stage 2}: Recovery exploiting prior info}
    \vspace{0.1cm}
    \State $\textcolor{black}{\mathbf{\pi}, \{\mathbf{\Phi}_i \}_{i=\tau+1}^N }  \gets  \mathbf{{updatePriors}}\left(\textcolor{black}{\mathbf{\pi}, \{\mathbf{\Phi}_i \}_{i=\tau+1}^N } , \mathcal{P}^{(S_{1,2})}, \mathcal{P}^{(DD)}\right)$ \label{alg:line:updateprior}
    \State $\mathbf{Z} \gets \mathbf{LVA}\left(L, \tau, \textcolor{black}{\mathbf{\pi}, \{\mathbf{\Phi}_i \}_{i=\tau+1}^N } \right)$ \label{alg:line:lva}
    \State $\mathbf{C} \gets \{\}$

    \For{$l \gets 1$ to $L$} \label{alg:line:fil1}
        \If{$\textcolor{black}{K_{\min} \leq} \sum_i \left( \mathbf{Z}_{l,i} \right) \leq \textcolor{black}{\lfloor} \gamma \textcolor{black}{K_{\max}}\textcolor{black}{\rfloor} $}\label{alg:line:fil2}
            \State $\mathbf{V}^{(l)} \gets \{i \mid \mathbf{Z}_{l,i} = 1\}$ \label{alg:line:fil00}
            \State $\mathbf{C} \gets \mathbf{C} \cup \mathbf{getAllCombinations}\left( \mathbf{V}^{(l)}, \textcolor{black}{K_{\min}}, \textcolor{black}{K_{\max}} \right)$\label{alg:line:fil7}
        \EndIf\label{alg:line:fil9}
    \EndFor\label{alg:line:fil10}

    \State $\hat{\mathcal{K}} \gets \mathbf{MAP}\left( \mathbf{X}, Y, \mathbf{C}, \textcolor{black}{\mathbf{\pi}, \{\mathbf{\Phi}_i \}_{i=\tau+1}^N }  \right)$ \label{alg:line:map}

\end{algorithmic}
\end{algorithm}

\subsubsection{Recovery Process} The suggested recovery algorithm operates in two main stages. In the first stage (Stage 1), to reduce the space of search (i.e., the possible defective items), the algorithm efficiently identifies non-defective and definitely defective items without considering the prior correlated information. In the first step of this stage, we use the \ac{dnd} algorithm \cite{kautz1964nonrandom,chen2008survey, chan2011non,aldridge2014group} (line \ref{alg:line:coma}, Fig.~\ref{fig:MSGT_exmaple:stage1}.(a), and Algorithm~\ref {alg:cap_coma}). \ac{dnd} compares the columns of the testing matrix, $\mathbf{X}$, with the outcome vector, $Y^T$. If $Y(i) = 0$ for some $i \in \{1, \ldots, T\}$, the algorithm eliminates all items participating in the $i$-th test from being defective, and outputs them as the set $\mathcal{P}^{(DND)} \subset \mathcal{N}$.

In the second step of Stage 1, we use the \ac{dd} algorithm \cite{aldridge2014group} (line~\ref{alg:line:dd}, Fig.~\ref{fig:MSGT_exmaple:stage1}.(b), and Algorithm~\ref {alg:cap_dd}), which goes over the testing matrix and the test result, looking for positive \textcolor{black}{pooled tests} that include only one possibly defective item. \ac{dd} denotes those items as definitely defective items and outputs them as the set $\mathcal{P}^{(DD)}$.

Let $\mathcal{P}^{(S_{1,1})} = \mathcal{N}\setminus \mathcal{P}^{(DND)}$ and $\mathcal{P}^{(S_{1,2})} = \mathcal{N} \setminus \left( \mathcal{P}^{(DND)} \cup \mathcal{P}^{(DD)} \right)$ denote the set of items that their status is still unclear after the first step and the second step, respectively\textcolor{black}{, where we use the notation $S_{i,j}$ with the first index $i$ indicating the stage and the second index $j$ indicating the step within that stage}.
$\mathcal{P}^{(S_{1,2})}$ holds a new space search, and $\mathcal{P}^{(DD)}$ holds the already known defectives. This knowledge is acquired without utilizing any prior data, which we reserve for the second stage.


\begin{figure}[ht]
    \vspace{-0.1cm}
    \centering
    \ifsingle
    \includegraphics[width=0.35\linewidth]{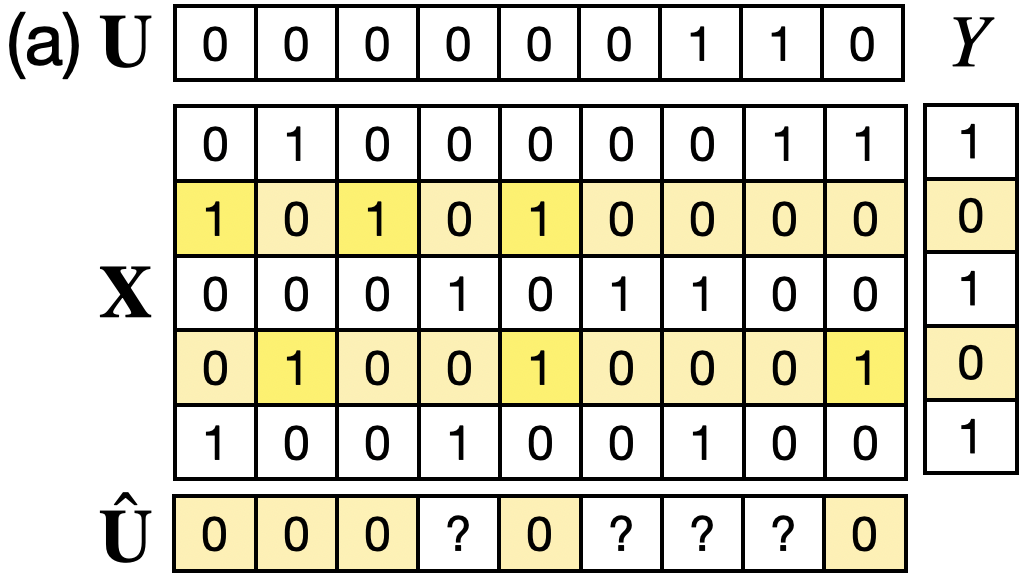}
    \includegraphics[width=.35\linewidth]{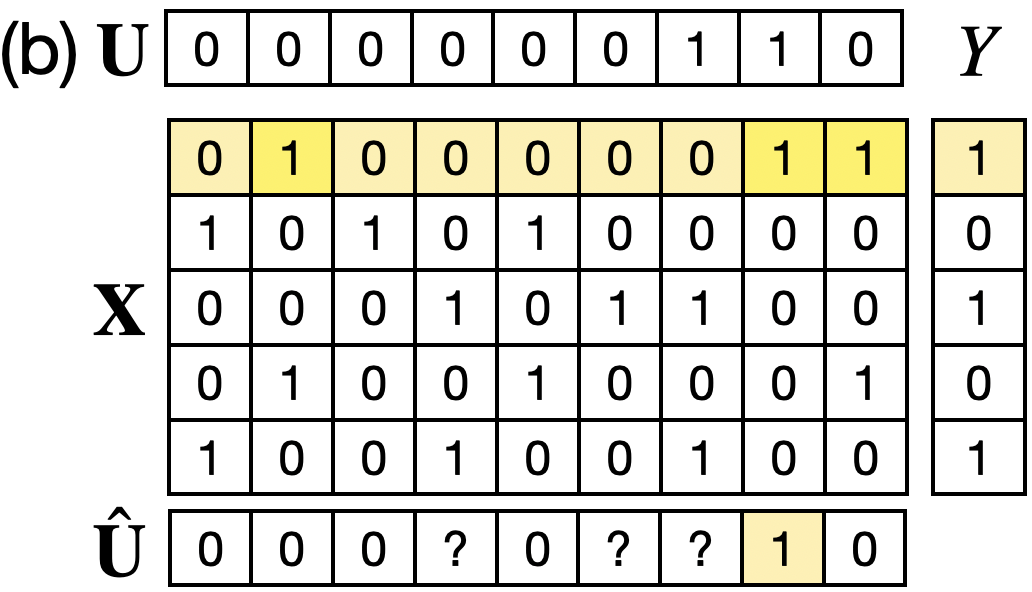}
    \else
    \includegraphics[width=0.49\linewidth]{MSGT9_DND.png}
    \includegraphics[width=.49\linewidth]{MSGT9_DD.png}
    \fi
    \caption{\small First stage of \textcolor{black}{\ac{msgt}}. (a) The first step of Stage 1, the DND algorithm, reveals 5 DND items in $\mathbf{U}$, forming $\mathcal{P}^{(DND)}$. Since items participating in negative tests must be non-defective, we mark all the participants in the two negative test results as non-defective. (b) The second step of Stage 1, the DD algorithm, outputs $\mathcal{P}^{(DD)}$ that includes a single DD item, based on the first test result, as it is the only possibly defective item participating in this test. The two other positive test results do not contribute to our knowledge here because there is more than one possibly defective item participating in them.}
    \label{fig:MSGT_exmaple:stage1}
    \ifsingle\else
    \vspace{-0.4cm}
    \fi
\end{figure}


In the first step of Stage 2, we translate the data we obtained in \ac{dnd} and \ac{dd}, $\mathcal{P}^{(S_{1,2})}$ and $\mathcal{P}^{(DD)}$, into the state space in terms of transition probabilities, $\textcolor{black}{\mathbf{\pi}}$, and initial probabilities, $\left\{\mathbf{\pi}_i\right\}_{i=1}^N$, so we can employ all the gathered information in the next steps (line~\ref{alg:line:updateprior}, Fig.~\ref{fig:MSGT_example:stage2}.(a), and Algorithm~\ref{alg:cap_update_priors}). \textcolor{black}{In the state space, the population sequence, $\mathbf{U}$, considered herein, is replaced with time steps considered traditionally in VA for communications problems \cite{seshadri1994list}, in which we have in GT two possible states per item, the first indicates ``non-defective'' and the second indicates ``defective''.}

\begin{remark}
\textcolor{black}{
Under the Markovian prior model described in Section~II, the population vector
$\mathbf{U}=(U_1,\ldots,U_N)$ admits a natural representation as a path in a trellis,
where each state corresponds to the defectivity status of an item (or a $\tau$-length
state for $\tau$-order memory), and the transition probabilities are given by the
prior parameters $(\pi,\{\Phi_i\})$.
In this representation, exact MAP decoding corresponds to
finding the most likely path through the trellis given the observed test outcomes.
However, exact MAP decoding over this trellis requires evaluating an exponential
number of paths and is therefore computationally infeasible for large $N$.
The LVA provides an efficient approximation by returning a
\emph{list} of the $L$ most likely trellis trajectories, rather than committing to a
single path.
This is particularly well suited to the proposed framework, since after the Stage~1
reduction (DND and DD), multiple high-probability trajectories may remain consistent
with the observations.
By producing a small set of candidate trajectories, LVA enables systematic
exploitation of the correlated prior structure while maintaining feasible
computational complexity.
}
\end{remark}

In the second step, the suggested \ac{lva} for \ac{gt} goes over the sequence of items and outputs $\mathbf{Z} \in \{0,1\}^{L\times N}$, a list of the $L$ most likely trajectories in the state space (line~\ref{alg:line:lva}, Fig.~\ref{fig:MSGT_example:stage2}.(b), and Algorithm~\ref{alg:cap_lva}) according to \ac{map} decision based on the given prior information. Each trajectory is a sequence of $N$ states, representing items classified as either defective or non-defective. Thus, $\mathbf{Z}$ provides $L$ estimations of $\mathbf{U}$. In practice, the $L$ estimations may include any number of defective items and require further processing. \textcolor{black}{ The list size parameter $L$ controls the tradeoff between recovery performance and
computational complexity. Larger values of $L$ increase the probability that the true defective pattern is contained in the list of candidate trajectories, at the cost of increased computation in subsequent candidate enumeration and MAP selection. This tunable tradeoff is a key advantage of using LVA in the proposed decoding architecture.
}

In the third step, we extract candidates for the defective set out of the $L$ estimated sequences $\mathbf{Z}$ (lines~\ref{alg:line:fil1} to \ref{alg:line:fil10}). For some $l\in\{1,\ldots, L\}$, let $\mathbf{V}^{(l)}$ denote the set of items estimated as defective in $\mathbf{Z}_l$, the $l$-th row of $\mathbf{Z}$ (line~\ref{alg:line:fil00}). We ignore sequences that contain \textcolor{black}{fewer than $K_{\min}$ defective items} or more than \textcolor{black}{$\textcolor{black}{\lfloor} \gamma K_{\max}\textcolor{black}{\rfloor} $} defective items, for some $\gamma \geq1 $, and consider only $\mathbf{Z}_l$ in which $\textcolor{black}{K_{\min}\leq} \left| \mathbf{V}^{(l)}\right|<\textcolor{black}{\lfloor} \gamma\textcolor{black}{K_{\max}\textcolor{black}{\rfloor}}$ as valid sequences. For each one of the valid sequences, \textcolor{black}{we refer to all the combinations of size $k$ in $\mathbf{V}^{(l)}$, for every $k\in\{K_{\min},\ldots,\min(K_{\max},|\mathbf{V}^{(l)}|)\}$, as candidates for $\hat{\mathcal{K}}$} and add them to the candidates list $\mathbf{C}$ (line~\ref{alg:line:fil7}). \textcolor{black}{In particular, for a given set of indices $\mathbf{V}^{(l)}\subseteq\mathcal{N}$ and integers
$1\le K_{\min}\le K_{\max}$, we define
$\mathbf{getAllCombinations}(\mathbf{V}^{(l)},K_{\min},K_{\max})
\;\triangleq\;
\bigcup_{k=K_{\min}}^{\min(K_{\max},|\mathbf{V}^{(l)}|)}
\left\{
\mathcal{S}\subseteq \mathbf{V}^{(l)} \;:\; |\mathcal{S}|=k
\right\}$.}
\textcolor{black}{Thus, $\mathbf{getAllCombinations}$ $(\mathbf{V}^{(l)},K_{\min},K_{\max})$ returns the collection of all subsets of $\mathbf{V}^{(l)}$ whose cardinality lies in the range $K_{\min}$ and $K_{\max}$. The operator $\mathbf{getAllCombinations}(\cdot)$ as detailed in Algorithm~\ref{alg:getAllCombinations} is used here to generate candidate defective sets without assuming knowledge of the exact sparsity level.}

At this point, we have in $\mathbf{C}$ a list of sets of size \textcolor{black}{$k\in\{K_{\min},\ldots,K_{\max}\}$}, that are candidates to be our final estimation $\hat{\mathcal{K}}$, and we can calculate the probability of each one of them using $\textcolor{black}{\mathbf{\pi}, \{\mathbf{\Phi}_i \}_{i=\tau+1}^N}$ (see Algorithm~\ref{alg:cap_update_priors}). Then, in the fourth step, the estimated defective set, $\hat{\mathcal{K}}$, is finally chosen using \ac{map} estimator out of the $\mathbf{C}$ (line~\ref{alg:line:map}), i.e., $\hat{\mathcal{K}}=\arg\max_{c\in \mathbf{C}}{ P\left( Y \lvert \mathbf{X},c\right) P\left( c\right) }$ (see Algorithm~\ref{alg:cap_map}).

\textcolor{black}{If there are no valid sequences in $\mathbf{Z}$, we consider trajectories with fewer than $K_{\min}$ detections for partial recovery. We select the trajectory with the most detections and randomly complete it to form a set of size $K_{\min}$ for our final estimation $\hat{\mathcal{K}}$.}

\vspace{-0.25cm}
\begin{figure}[ht]
    \centering
    \ifsingle
    \includegraphics[width=0.60\linewidth]{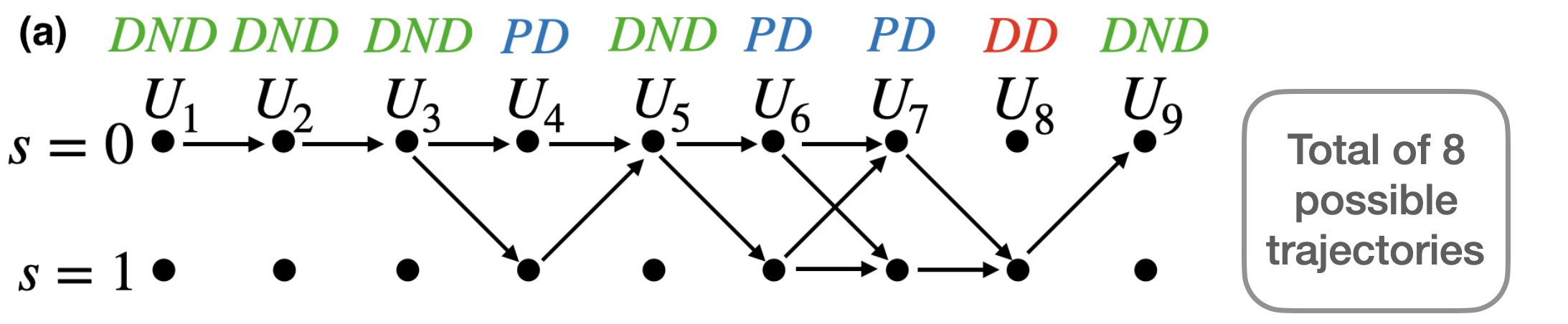}
    \vspace{0.1cm}
    \includegraphics[width=0.65
    \linewidth]{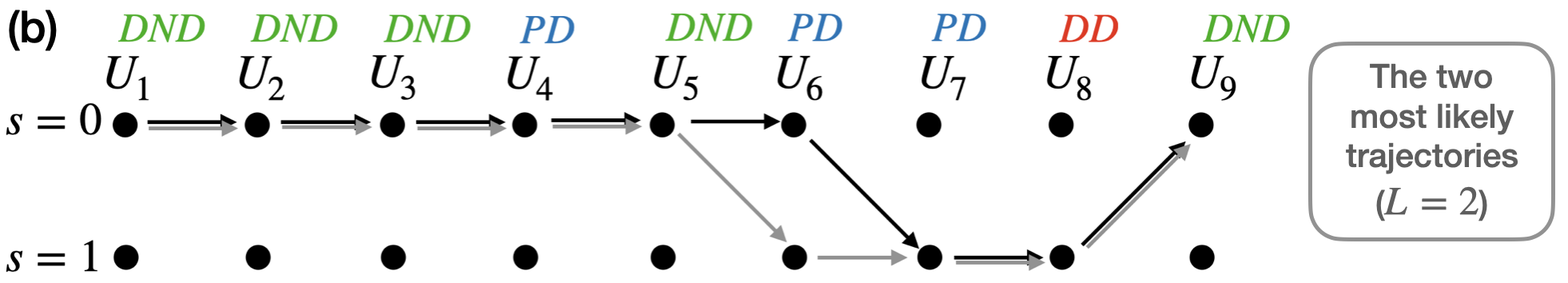}
    \vspace{0.1cm}
    \includegraphics[width=.35\linewidth]{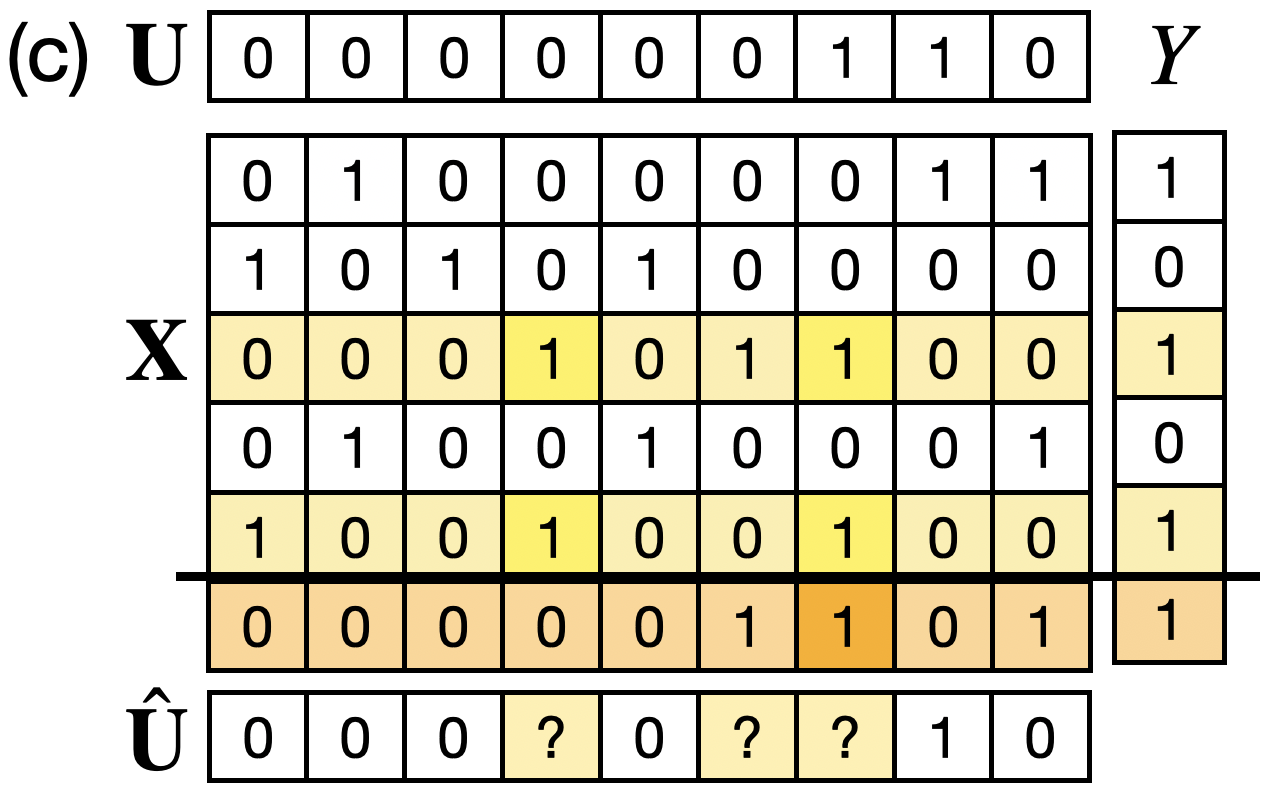}
    \else
    \includegraphics[width=0.85\linewidth]{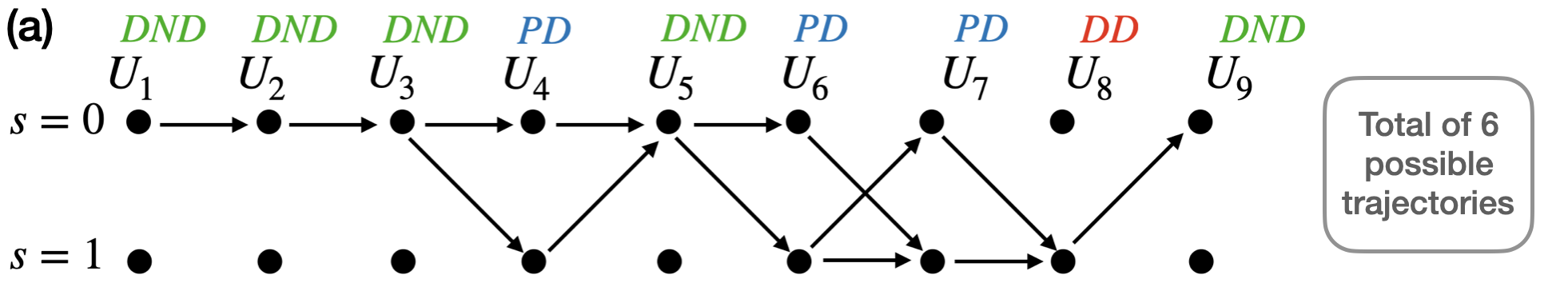}
    \vspace{0.1cm}
    \includegraphics[width=0.85
    \linewidth]{MSGT9_LVA.png}
    \vspace{0.1cm}
    \includegraphics[width=.47\linewidth]{MSGT9_vsML.png}
    \fi
    \caption{\small Stage 2 of \textcolor{black}{\ac{msgt}}. (a) All the possible transitions in the state space that we consider in the LVA step, following the insights obtained in Stage 1. These transitions aggregate to a total of 6 trajectories. (b) The two most likely trajectories returned by LVA (assuming $L=2$). Given $K=2$, the black trajectory corresponds to a valid population vector $\mathbf{U}$ with 2 defective items, while the gray trajectory indicates an invalid population with 3 defective items instead. Consequently, in the subsequent step, \ac{msgt} will extract two optional defective sets: $\left\{U_6,U_8\right\}$ and $\left\{U_7,U_8\right\}$, and will finally choose the most likely one using \ac{map} estimator. (c) Comparison of Stage 2 to ML. With $T=5$, we use the first 5 rows of the testing matrix, ignoring the last test result. This leaves 3 possibly defective items, forming two potentially defective sets of size $K=2$. Using ML, one set is chosen randomly, leading to an error probability of 0.5. With $T=6$, based on the third and sixth test results, there is only one set of size $K=2$ that matches the outcome $Y$, resulting in successful decoding with the ML decoder. As shown above, \ac{msgt}'s Stage 2 can successfully decode $\mathbf{U}$ with just $T=5$, as using the LVA step it narrows down to only 2 possible trajectories, and then the final estimation is selected based on the given prior information and the insights gained in Stage 1.}
    \label{fig:MSGT_example:stage2}
\end{figure}

\ifsingle\else
\begin{figure*}				
	\centering	
	\begin{subfigure}[b]{.32\linewidth}
		\includegraphics[trim={1cm 0 0 0 }, width=\linewidth, height=1.2in]{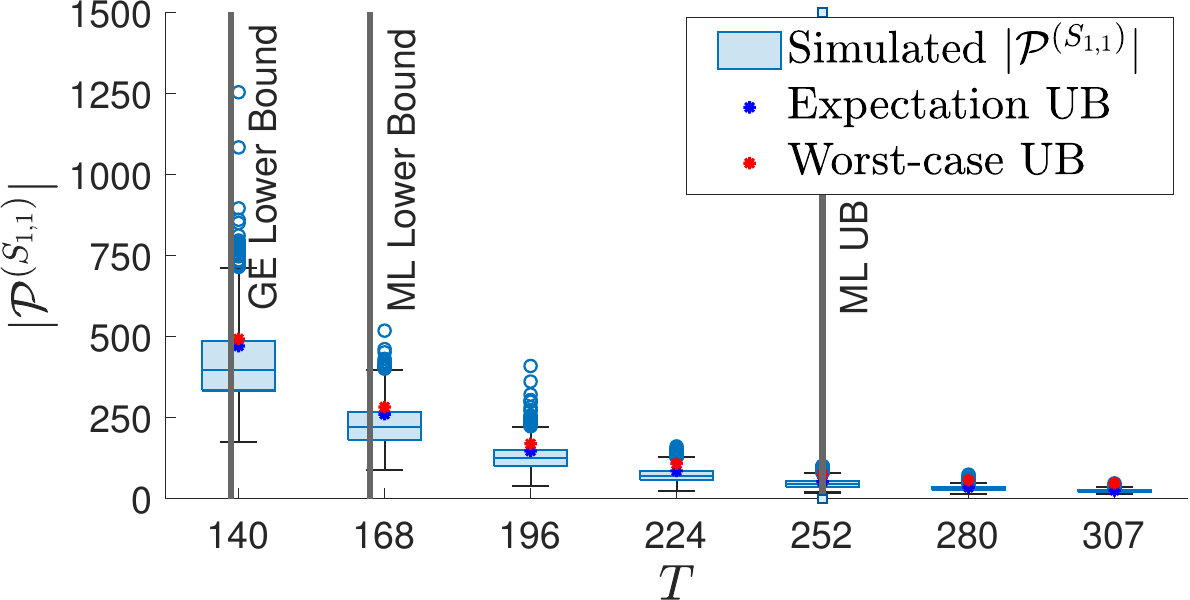}	
        \captionsetup{width=0.97\linewidth}
	\caption{\small Upper bound for possibly defective items after DND for $N=10000$, and $K=15$.}\label{fig:pd_after_coma_bounds}
	\vspace{-0.15cm}
	\end{subfigure}
	\begin{subfigure}[b]{.32\linewidth}
		\includegraphics[width=\linewidth, height=1.2in]{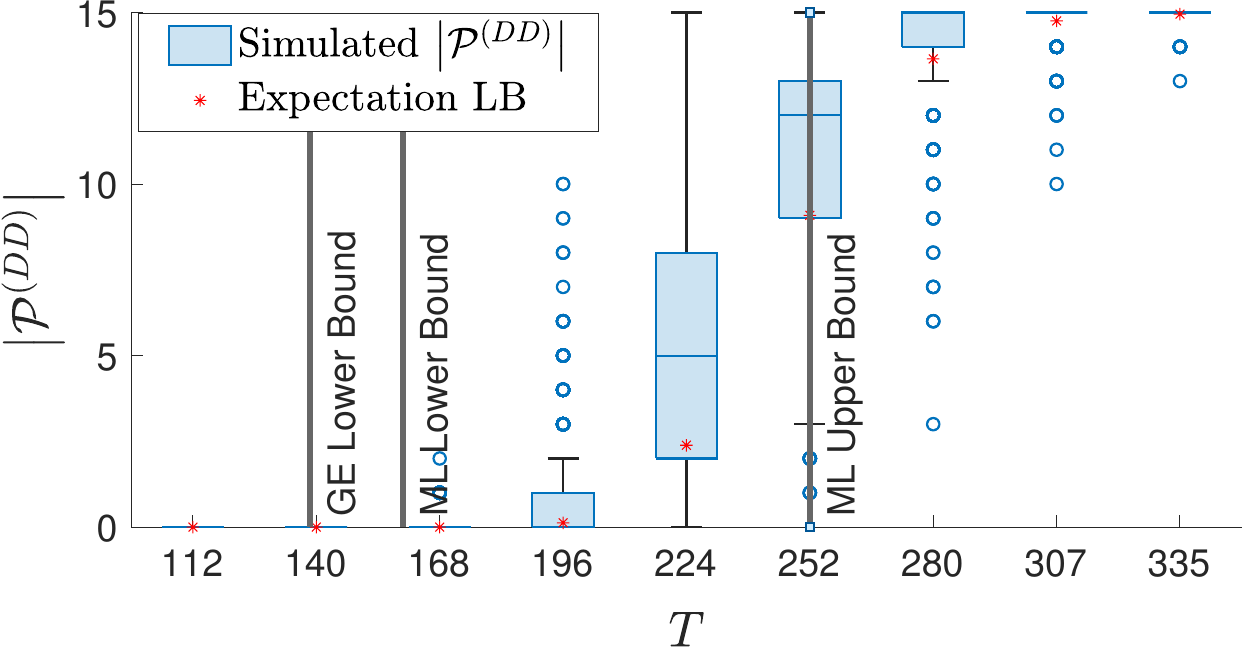}
        \captionsetup{width=0.99\linewidth}
		\caption{\small Lower bound for definitely defective items after DD for $N=10000$, and $K=15$.}\label{fig:dd_after_dd_bounds}
        \vspace{-0.15cm}
	\end{subfigure}
	\begin{subfigure}[b]{.33\linewidth}
		\includegraphics[width=\linewidth, height=1.2in]{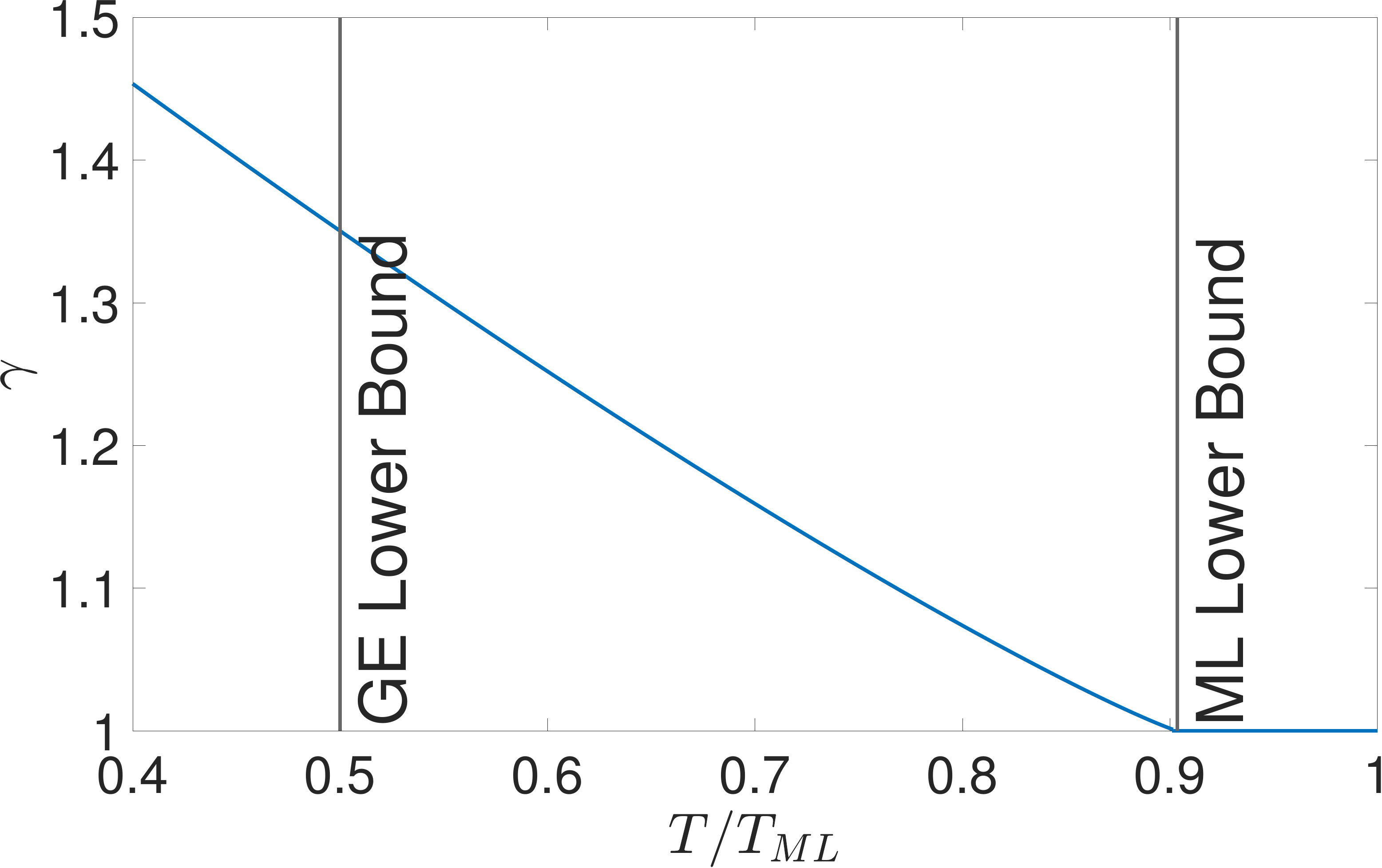}
        \captionsetup{width=0.97\linewidth}
		\caption{\small Minimum $\gamma$ parameter to satisfies~\eqref{eq:gamma_cond} for $N=500$, and $K=3$.}
		\label{fig:gamma_vs_T}
    \vspace{-0.15cm}
	\end{subfigure}
	\caption{\small Numerical evaluation for theoretical results and bounds. The results in (a), (b), and (c) are over 1000 iterations. For ML Upper Bound (UB), $T_{ML} = (1+\epsilon)K\log_{2}N$, for any $\epsilon>0$ \cite{atia2012boolean}. In particular, $\epsilon=0.25$ in the results presented herein.}
	\vspace{-0.5cm}
	\label{fig:ConfNew}
\end{figure*}
\fi

\subsection{Analytical Results}\label{subsec:analytical_results}

In this section, we provide analytical results related to the proposed \ac{msgt}. To support our theoretical results, we demonstrate the derived bounds and compare them to simulation results in Fig.~\ref{fig:ConfNew}.

\textcolor{black}{
In the correlated-prior model, the number of defectives $K=\sum_{i=1}^N U_i$ is random but is assumed to satisfy the known bound $K_{\min}\le K\le K_{\max}$ (see Section~\ref{sec_problem_form}). Accordingly, throughout this section we analyze Stage~1 under a Bernoulli design with parameter $p=\ln(2)/K_{\max}$ and $T$ tests, where $K_{\max}$ is the design upper bound. Theorem~\ref{thm:avg_coma_bounds} bounds the expected number of items whose status remains unknown after the \ac{dnd} step (equivalently, it bounds the expected number of non-defectives that survive as ``possibly defective''), as a function of the \emph{realized} $K=k$, and we denote the corresponding average contribution by $P_{e,a}^{(DND)}$. Complementarily, Theorem~\ref{thm:upper_bound_PD_coma_bernoulli} upper bounds the probability that this number exceeds its expectation by a margin, and we denote this upper-tail deviation term by $P_{e,u}^{(DND)}$. These bounds control the size of the remaining search space passed to Stage~2, and in particular ensure that the Stage~2 list-size budget (parameterized by $\gamma$ and the design upper bound $K_{\max}$) is met.
}

\begin{theorem}[\cite{solomon2025one}]\label{thm:avg_coma_bounds}
    \textcolor{black}{Consider a group test with a Bernoulli testing matrix with $p=\ln{2}/K_{\max}$ and $T$ tests, and suppose the (random) number of defectives $K$ satisfies $K_{\min}\le K\le K_{\max}$\off{, where $K_{\min}$ and  $K_{\max}$ are the design lower and upper bounds, respectively}. \textcolor{black}{Fix any realization $K=k$ and}} \textcolor{black}{let $P_{e,a}^{(DND)} \triangleq N^{-\alpha \left( 1-\ln{2}/K_{\max}\right) /2}$ for $\alpha \triangleq T/(K_{\max} \log_2 N)$.}
         The expected number of possibly defective items is bounded by
         \begin{equation*}\label{eq_expected_pd}
            \mathbb{E}\left[ \left| \mathcal{P}^{(S_{1,1})}\right| \,\big|\, K=\textcolor{black}{k}\right] \leq \textcolor{black}{k+\left( N-k\right)} P_{e,a}^{(DND)}.
         \end{equation*}
\end{theorem}

\textcolor{black}{The proof of Theorem~\ref{thm:avg_coma_bounds} follows directly from the definition of the \ac{dnd} decoding rule and the Bernoulli test design \cite{cohen2021multi,solomon2025one}. A non-defective item is eliminated unless it appears only in tests whose outcomes are positive; equivalently, it survives the \ac{dnd} step if and only if all tests in which it participates are occluded by at least one defective item. Since each item participates in tests independently with probability $p$, the survival event can be characterized by conditioning on the number of tests involving the item and on whether those tests include at least one defective. In particular, under a Bernoulli$(p)$ design with $k$ defectives, $(1-p)^{k}$ denotes the probability that a test is negative, and $p(1-p)^{k-1}$ denotes the probability that a test contains defective $i$ and no other defectives. This leads to the stated expression for the probability of a non-defective item surviving the \ac{dnd} step, which is then used in the subsequent expectation and concentration analysis.}

\begin{proof}
According to \cite[Lemma~2]{bharadwaja2022approximate}, conditioning on $\{K=\textcolor{black}{k}\}$ yields
\begin{equation}\label{eq_pd}
\mathbb{E}\!\left[\left|\mathcal{P}^{(S_{1,1})}\right| \,\big|\, K=\textcolor{black}{k}\right]
= \textcolor{black}{k+\left(N-k\right)}P_{e}^{(\mathrm{DND})},
\quad \text{where} \quad
P_{e}^{(\mathrm{DND})}=\bigl(1-p(1-p)^{\textcolor{black}{k}}\bigr)^T .
\end{equation}
\textcolor{black}{By substituting the design choice $p=\ln 2/K_{\max}$, we bound $P_e^{(\mathrm{DND})}$ as follows}
\begin{align*}
    P_e^{(\mathrm{DND})}
    &= \left(1 - p(1-p)^{\textcolor{black}{k}} \right)^{T} \\
    &= \left(1 - p(1-p)^{\textcolor{black}{k-1}}(1-p) \right)^{T} \\
    &\le \left(1 - p\,e^{-\ln 2}\,(1-p) \right)^{T},
\end{align*}
\textcolor{black}{where the last step follows since $e^{-\ln 2} \leq (1-\ln2/m)^{m-1}$ for any integer $m \geq 0$, such that, together with the bound $\textcolor{black}{k}\le K_{\max}$ and
with $p=\ln 2/K_{\max}$, we have
\[
(1-p)^{\textcolor{black}{k-1}} \;\ge\; (1-p)^{K_{\max}-1}
\;\ge\; e^{-\ln 2}
\;=\;\frac{1}{2},
\]
as the first inequality holds since $0<1-p<1$ and $\textcolor{black}{k-1}\le K_{\max}-1$, and the second from the exponential bound.
}

\textcolor{black}{
Now, applying $(1-z)^T \le e^{-zT}$ with
\[
z \triangleq p\cdot \frac{1}{2}\cdot(1-p)
= \frac{\ln 2}{K_{\max}}\cdot \frac{1}{2}\cdot\left(1-\frac{\ln 2}{K_{\max}}\right),
\quad \text{and} \quad
T=\alpha K_{\max}\log_2 N,
\]
we obtain
\[
P_e^{(\mathrm{DND})}
\le \exp\!\left(- z\cdot \alpha K_{\max}\log_2 N\right)
= N^{-\frac{\alpha}{2}\left(1-\frac{\ln 2}{K_{\max}}\right)}.
\]
}

Finally, substituting \textcolor{black}{$P_{e,a}^{(\mathrm{DND})} \triangleq N^{ - \frac{\alpha}{2} \left(1 - \frac{\ln 2}{K_{\max}} \right)}$}
into \eqref{eq_pd} yields
\[
\mathbb{E}\!\left[\left|\mathcal{P}^{(S_{1,1})}\right| \,\big|\, K=\textcolor{black}{k}\right]
\le \textcolor{black}{k+(N-k)}\,P_{e,a}^{(\mathrm{DND})}.
\]
\end{proof}

\textcolor{black}{
\begin{theorem}\label{thm:upper_bound_PD_coma_bernoulli}
\textcolor{black}{Consider a group test with $T$ tests, a Bernoulli testing matrix with the design parameter $p=\ln(2)/K_{\max}$, and the \ac{dnd} decoder. Fix any realization $K=k\in[K_{\min},K_{\max}]$ and} let $G$ denote the number of non-defective items marked as possibly defective by \ac{dnd}.
Then, for any threshold $g>0$,
\begin{equation} \label{eq:main_peu_bound}
P_{e,u}^{(\mathrm{DND})}(g)
\;\triangleq\;
P\!\left( G - \mathbb{E}[G] > g \,\big|\, K=\textcolor{black}{k}\right)
\;\le\;
\frac{\mathrm{Var}(G\mid K=\textcolor{black}{k})}{\mathrm{Var}(G\mid K=\textcolor{black}{k})+g^{2}},
\end{equation}
where
\begin{equation}\label{eq:VarG}
\mathrm{Var}(G\mid K=\textcolor{black}{k})
= (N-\textcolor{black}{k})(\bar r_1 - \bar r_2)
\;+\;
(N-\textcolor{black}{k})^2 (\bar r_2 - \bar r_1^{\,2}),
\end{equation}
with
$\bar r_1 = \big(1 - p q_0\big)^{T}$ and
$\bar r_2 = \big( 1 - (2p-p^{2})q_0 \big)^{T}$ for $q_0 \triangleq (1-p)^{\textcolor{black}{k}}$.
\end{theorem}
}

\begin{proof}
\textcolor{black}{
Let $G$ denote the number of non-defective items that were hidden and not detected in \ac{dnd} (i.e., the number of false positive items).
We introduce the following upper-tail deviation probability for the number of false positives
\[
P_{e,u}^{(\mathrm{DND})} \;\triangleq\; P\!\left( G-\mathbb{E}[G]>g \,\big|\, K=\textcolor{black}{k}\right),\quad \text{for} \quad g>0.
\]
\textcolor{black}{Under the defective model given in Section~\ref{sec_problem_form} with $K_{\min}\le K\le K_{\max}$, Stage~2 enforces only an \emph{upper} budget $\textcolor{black}{\lfloor} \gamma K_{\max}\textcolor{black}{\rfloor} $ on the support size (Algorithm~\ref{alg:cap}). Accordingly, we choose $g$ to control the event that the number of possibly defective items after \ac{dnd} exceeds this budget, e.g.,
\[
g \;=\; \textcolor{black}{\lfloor} \gamma K_{\max} \textcolor{black}{\rfloor} \;-\; \mathbb{E}\!\left[\left|\mathcal{P}^{(S_{1,1})}\right|\mid K=\textcolor{black}{k}\right]
\;=\; \textcolor{black}{\lfloor} \gamma K_{\max} \textcolor{black}{\rfloor} \;-\; \Big(\textcolor{black}{k+(N-k)}P_{e,a}^{(\mathrm{DND})}\Big),
\]
whenever this quantity is positive.}
}
\textcolor{black}{
Applying the one-sided Chebyshev's inequality \cite{feller1991introduction} to bound the concentration term gives
\begin{equation}\label{eq:cantelli_for_g}
P\!\big(G-\mathbb{E}[G]>g \,\big|\, K=\textcolor{black}{k}\big) \;\le\; \frac{\mathrm{Var}(G\mid K=\textcolor{black}{k})}{\mathrm{Var}(G\mid K=\textcolor{black}{k})+g^2}.
\end{equation}
Let $M_0$ denote the number of negative tests, $q_0 \triangleq (1-p)^{\textcolor{black}{k}}$, and $n'\triangleq N-\textcolor{black}{k}$. Each test is negative independently with probability $q_0$, hence $M_0 \sim Bin(T,q_0)$.
Conditioned on $M_0=m_0$, a fixed non-defective item appears in none of the $m_0$ negative tests with probability $(1-p)^{m_0}$. Consequently, $G\,\big|\,M_0=m_0,\ K=\textcolor{black}{k} \;\sim\; Bin\big(n',\ (1-p)^{m_0}\big)$ \cite{aldridge2014group}.
}
\textcolor{black}{
For $M\sim Bin(T,q)$ and any $a\in\mathbb{R}$,
\begin{equation}
\mathbb{E}\!\left[a^{M}\right] \;=\; \big((1-q)+q\,a\big)^{T}.
\label{eq:bin-moment-gamma}
\end{equation}
Using the law of total variance and $r\triangleq \left(1-p\right)^{M_0}$,
\begin{align}
\mathrm{Var}(G\mid K=\textcolor{black}{k}) &= \mathbb{E}\!\left[\mathrm{Var}(G\mid M_0, K=\textcolor{black}{k})\right] + \mathrm{Var}\!\left(\mathbb{E}[G\mid M_0, K=\textcolor{black}{k}]\right) \nonumber\\
&= \mathbb{E}\!\left[n' r(1-r)\right] + \mathrm{Var}(n' r) \nonumber\\
&= n'\big(\mathbb{E}[r] - \mathbb{E}[r^2]\big) + n'^2\big(\mathbb{E}[r^2] - (\mathbb{E}[r])^2\big).
\label{eq:var-decomp-gamma}
\end{align}
Using \eqref{eq:bin-moment-gamma}, we obtain
\[
\bar r_1 = \mathbb{E}[r] = (1-pq_0)^T,
\quad \text{and} \quad
\bar r_2 = \mathbb{E}[r^2]
= \big( 1 - (2p-p^{2})q_0 \big)^{T}.
\]
Substituting these expressions into the variance decomposition
\eqref{eq:var-decomp-gamma} yields the variance formula stated in
\eqref{eq:VarG}. Finally, combining \eqref{eq:cantelli_for_g} with
\eqref{eq:VarG} establishes the upper-tail bound in
\eqref{eq:main_peu_bound}, completing the proof.
}
\end{proof}

The following theorem provides the expected defective items detected by \ac{dd} in the second step of Stage~1.

\begin{theorem} \label{thm:avg_dd_bound}
\textcolor{black}{Consider noiseless nonadaptive group testing with $N$ items, a random number of defectives $K$ (conditioned on its realization), and $T$ tests, under a Bernoulli$(p)$ design with design parameter $p=\ln(2)/K_{\max}$. Fix any realization $k\in[K_{\min},K_{\max}]$ and consider performance conditioned on $\{K=k\}$.}
\textcolor{black}{
The number of defective items detected by the \ac{dd} decoder is bounded by
\begin{equation}
\mathbb{E}\big[\,|\mathcal{P}^{(\mathrm{DD})}|\,\big|\, K=\textcolor{black}{k}\big]
\;\ge\;
\textcolor{black}{k}\!\left\{
a\,T(1-q_0)\,(1-p)^{N-\textcolor{black}{k}}
-\frac{a^2}{2}\,T^2
\right\}^{+},
\label{eq:E_DD_bound_raw}
\end{equation}
where $\{x\}^{+}\triangleq \max\{x,0\}$, \textcolor{black}{all expectations are taken conditionally on the realized $K$, and}
$q_0 \triangleq (1-p)^{\textcolor{black}{k}}$, $q_1 \triangleq p(1-p)^{\textcolor{black}{k-1}}$, and
$a \triangleq q_1/(1-q_0)$.
}
\end{theorem}

\begin{proof}[Proof:]
The proof of Theorem~\ref{thm:avg_dd_bound} in given in Appendix~\ref{sec:proof:avg_dd_bound}.
\end{proof}

\begin{remark}
\textcolor{black}{
Under the Bernoulli design used in this work, the testing matrix is parameterized by a design sparsity range $[K_{\min},K_{\max}]$, and the realized number of defectives $K$ is random and unknown to the decoder. Accordingly, the lower bound in~\eqref{eq:E_DD_bound_raw}, which is stated
conditionally on a fixed realization of $K$, may evaluate to zero for some values of $K\in[K_{\min},K_{\max}]$ after applying the positive-part operator. As a result, the bound is not uniformly informative over the entire admissible
sparsity range, but remains valid for all realizations of $K$.} Since we always have the trivial bound
$\mathbb{E}[\,|\mathcal{P}^{(\mathrm{DD})}|\,]\ge 0$, this behavior is consistent with correctness and does not indicate a failure of the analysis. \textcolor{black}{Moreover, the proposed \ac{msgt} algorithm does not rely on the magnitude of $|\mathcal{P}^{(\mathrm{DD})}|$ for correctness or feasibility. In particular, Stage~2 applies a deterministic filtering rule that retains only candidate supports whose cardinality lies in the range $K_{\min}\le |\mathcal{V}^{(l)}|\le \lfloor \gamma K_{\max}\rfloor$, independently of the quantitative DD bound.} For this reason, the complexity analysis of \ac{msgt} in Theorem~\ref{thm:complexity} relies only on the elementary inequalities $0 \le |\mathcal{P}^{(\mathrm{DD})}| \le |\mathcal{P}^{(S_{1,1})}|
\le N$, rather than on the quantitative lower bound~\eqref{eq:E_DD_bound_raw}.
\end{remark}

\textcolor{black}{The proposition below states a \emph{prior-free counting bound} relating the Stage~2 candidate-set budget and the best achievable \emph{uniformly averaged} probability of exact support recovery at a fixed sparsity level $k$, where the averaging is uniform over all supports of cardinality $k$ (see Remark~\ref{rem:UA}). It is used to motivate how the Stage~2 parameter $\gamma$ should scale with $T$ and $(N,K_{\max})$. Although the \ac{msgt} and MAP decoders explicitly depend on the prior $P(\mathcal{K})$, this proposition is a prior-free counting bound that applies to any decoding rule. In particular, the bound remains valid when MAP is used with an arbitrary correlated prior, since it constrains only the number of candidate supports consistent with a given outcome.}

\textcolor{black}{
\begin{prop}\label{prop:gamma_cond}
Fix any $k\in\{K_{\min},\ldots,K_{\max}\}$. Let $\eta_k \triangleq \left({1}/{\binom{N}{k}}\right) \sum_{\mathcal{K}:\,|\mathcal{K}|=k} P\!\left(\hat{\mathcal{K}}=\mathcal{K}\right)$ denote the uniformly averaged (over all $k$-subsets) probability of exact support
recovery at sparsity level $k$. Consider any decoding procedure that, given $Y$,
outputs a candidate item set $\mathcal{L}(Y)\subseteq[N]$ satisfying
$|\mathcal{L}(Y)|\le \textcolor{black}{\lfloor} \gamma K_{\max}\textcolor{black}{\rfloor} $. Then, achieving $\eta_k \ge \bar\eta_k$
requires
\begin{equation}\label{eq:gamma_cond}
T \;\ge\;
\log_2 \bar\eta_k
+ \log_2 \binom{N}{k}
- \log_2 \binom{\textcolor{black}{\lfloor} \gamma K_{\max}\textcolor{black}{\rfloor}}{k}.
\end{equation}
\end{prop}
}

\begin{proof}
\textcolor{black}{
Fix an arbitrary $k\in\{K_{\min},\ldots,K_{\max}\}$.
We follow a standard counting argument in the spirit of list-decoding converses
(cf.~\cite{scarlett2017little}), adapted here to decoders that output a
bounded-size item set. For each outcome vector $Y\in\{0,1\}^T$, define
\[
N(Y)\triangleq
\sum_{\mathcal{K}:\,|\mathcal{K}|=k}
\mathbb{1}\{\mathcal{K}\subseteq \mathcal{L}(Y)\},
\]
namely, the number of $k$-supports fully contained in the candidate set
$\mathcal{L}(Y)$.}

\textcolor{black}{Since \off{$|\mathcal{L}(Y)|\le \gamma K_{\max}$ and} $|\mathcal{L}(Y)|$ is integer, we
have $|\mathcal{L}(Y)|\le \lfloor \gamma K_{\max}\rfloor$, and hence
\[
N(Y)=\binom{|\mathcal{L}(Y)|}{k}
\le \binom{\lfloor \gamma K_{\max}\rfloor}{k}.
\]}

\textcolor{black}{Moreover, for any true support $\mathcal{K}$ with $|\mathcal{K}|=k$, the event
$\{\hat{\mathcal{K}}=\mathcal{K}\}$ necessarily implies that all defective items
must survive Stage~2, and hence $\mathcal{K}\subseteq \mathcal{L}(Y)$.
Averaging uniformly over all such supports and summing over all outcomes yields
\[
\eta_k \le
\frac{1}{\binom{N}{k}}
\sum_Y N(Y),
\]
where the inequality follows since
$\mathbb{1}\{\hat{\mathcal K}=\mathcal K\}\le
\mathbb{1}\{\mathcal K\subseteq \mathcal L(Y)\}$ and
$\Pr(Y\mid\mathcal K)\le 1$.}

\textcolor{black}{Using the bound on $N(Y)$ and the fact that there are $2^T$ possible outcomes,
we obtain
\begin{equation}\label{eq:b_eta_k}
\eta_k
\le
\frac{2^T}{\binom{N}{k}} \binom{\lfloor \gamma K_{\max}\rfloor}{k}.
\end{equation}
Rearranging \eqref{eq:b_eta_k} completes the proof.}
\end{proof}

\textcolor{black}{
\begin{remark}\label{rem:UA}
The success probability $\eta_k$ in Proposition~\ref{prop:gamma_cond} is defined as the probability of exact support recovery, uniformly averaged over all supports of cardinality $k$. This metric is \emph{not} the prior-dependent success probability of the MAP/LVA decoder under the correlated prior considered in this work. It is introduced solely to enable a prior-free counting argument.
Importantly, this proposition does not characterize the performance of MAP decoding under the correlated prior. Rather, it provides a necessary combinatorial limitation that applies to \emph{any} decoding rule whose Stage~2 output is restricted to a candidate item set of size at most $\textcolor{black}{\lfloor} \gamma K_{\max}\textcolor{black}{\rfloor} $. Since this
limitation is independent of the prior distribution over supports, it remains valid when MAP is employed with an arbitrary correlated prior. The bound is therefore used solely to motivate the scaling of $\gamma$ and $T$, and not as a direct performance guarantee.
\end{remark}
}

One of the key features of the proposed \ac{msgt} algorithm is its low and feasible complexity in practical regimes compared to \ac{ml} or MAP-based GT decoders. Both \ac{ml} and \ac{map} involve exhaustive searches, resulting in a complexity of $\mathcal{O}\left(\binom{N}{K} K N \log_2 N\right)$ operations \cite{atia2012boolean}. The theorem below and the subsequent Remark characterize the computational complexity of \ac{msgt}.

\begin{theorem}\label{thm:complexity}
    \textcolor{black}{Consider a group test for a population of $N$ items, with an unknown number of defectives $K$ satisfying $K_{\min}\le K\le K_{\max}$, and a Bernoulli testing matrix designed using $K_{\max}$.}
    The computational complexity of the \ac{msgt} algorithm is bounded by
    $\textcolor{black}{\mathcal{O}\bigl(L \gamma^{K_{\max}} K_{\max} N \log_2 N\bigr)}$ operations.
\end{theorem}

\begin{proof}
We begin by analyzing the complexity of each step of the proposed \ac{msgt} solution given in Algorithm~\ref{alg:cap}, and finally sum everything up to determine the total complexity.

The complexity of \ac{dnd} is $\textcolor{black}{\mathcal{O}\bigl(K_{\max} N\log_2{N} \bigr)}$, as analyzed in \cite[Remark~6]{cohen2020secure}.
Then, for each positive entry of the test result vector $Y^T$, the \ac{dd} algorithm counts the number of possibly defective items that participate in the corresponding pool test. This can be implemented with at most
$\textcolor{black}{K_{\max}} \bigl|\mathcal{P}^{(S_{1,1})}\bigr| \log_2{N}$ operations, which we bound by the \ac{dnd} complexity, i.e., by $\textcolor{black}{\mathcal{O}\bigl(K_{\max}N\log_2 N\bigr)}$.
Hence, the overall complexity of Stage~1 is $\textcolor{black}{\mathcal{O}\bigl(K_{\max}N\log_2 N\bigr)}$.

Parallel \ac{lva} requires $L$ times more computations than the \ac{va} \cite{seshadri1994list}. The \ac{va} calculates all the possible transition probabilities for each step in the sequence. In GT, this sequence is the items’ ordering, and with the proposed algorithm it is enough to consider only the $\bigl|\mathcal{P}^{(S_{1,2})}\bigr|$ items as the sequence length.

The possible states are “non-defective’’ and “defective’’, so there are four transitions in each step of the GT trellis. More generally, the algorithm can be implemented to leverage additional memory and decide the state of each item based on the preceding $\tau$ items. Consequently, \ac{lva} takes at most
\[
    2^{2 \tau} L \bigl|\mathcal{P}^{(S_{1,2})}\bigr|
\]
elementary operations.
\textcolor{black}{
For the purpose of an upper bound on the complexity, it suffices to note that for every realization
\[
    0 \,\le\, \bigl|\mathcal{P}^{(DD)}\bigr|
    \,\le\, \bigl|\mathcal{P}^{(S_{1,1})}\bigr|
    \,\le\, N,
\]
and hence
\[
    \bigl|\mathcal{P}^{(S_{1,2})}\bigr|
    \;=\;
    \bigl|\mathcal{P}^{(S_{1,1})}\bigr| - \bigl|\mathcal{P}^{(DD)}\bigr|
    \;\le\;
    \bigl|\mathcal{P}^{(S_{1,1})}\bigr|
    \;\le\; N.
\]
Therefore, we have
\[
    2^{2\tau} L \bigl|\mathcal{P}^{(S_{1,2})}\bigr|
    \;\le\;
    2^{2\tau} L N
    \;=\;
    \mathcal{O}(L N),
\]
so the complexity of the \ac{lva} step in \ac{msgt} is $\mathcal{O}(L N)$.
}

In the next step, we filter the \ac{lva} results.
We sum each sequence $\mathbf{Z}_\ell$ with complexity $\mathcal{O}(N)$ per sequence. \textcolor{black}{If the sum lies in the range $[K_{\min},\textcolor{black}{\lfloor} \gamma K_{\max}\textcolor{black}{\rfloor} ]$, we extract all combinations of sizes $k-\bigl|\mathcal{P}^{(DD)}\bigr|$ for all $k\in[K_{\min},K_{\max}]$.}
Hence, this step requires at most
\[
    \textcolor{black}{
    \mathcal{O}\!\left(
        L\left(
        N + \sum_{k=K_{\min}}^{K_{\max}}
        \binom{\textcolor{black}{\lfloor} \gamma K_{\max}\textcolor{black}{\rfloor} }{k-\left|\mathcal{P}^{(DD)}\right|}
        \right)
    \right).
    }
\]
This term will be dominated by the subsequent \ac{map} step.

Finally, in the \ac{map} step of Stage~2, the algorithm considers at most $L$ sequences. In each sequence, there are at most $\textcolor{black}{\lfloor} \textcolor{black}{\gamma K_{\max}}\textcolor{black}{\rfloor}$ possibly defective items, and we examine all combinations of sizes $k-\bigl|\mathcal{P}^{(DD)}\bigr|$ for $k\in[K_{\min},K_{\max}]$. Thus, the number of candidate supports per sequence is at most
\[
    \textcolor{black}{
    \sum_{k=K_{\min}}^{K_{\max}}
    \binom{\textcolor{black}{\lfloor} \gamma K_{\max}\textcolor{black}{\rfloor}}{k-\left|\mathcal{P}^{(DD)}\right|}
    \;\le\;
    \sum_{k=K_{\min}}^{K_{\max}}
    \binom{\textcolor{black}{\lfloor} \gamma K_{\max}\textcolor{black}{\rfloor}}{k}.
    }
\]
For each candidate support, we apply the group test metric, whose computation can be implemented with $\textcolor{black}{\mathcal{O}\bigl(K_{\max}N\log_2 N\bigr)}$ operations. Therefore, the complexity of the \ac{map} stage is bounded by
\[
    \textcolor{black}{
    L \sum_{k=K_{\min}}^{K_{\max}} \binom{\textcolor{black}{\lfloor} \gamma K_{\max}\textcolor{black}{\rfloor} }{k}
    K_{\max} N \log_2 N.
    }
\]
Using the standard bounds
\[
    \Big(\tfrac{\textcolor{black}{\lfloor} \gamma K_{\max}\textcolor{black}{\rfloor}}{K_{\max}}\Big)^{K_{\max}}
    \;\le\;
    \binom{\textcolor{black}{\lfloor} \gamma K_{\max}\textcolor{black}{\rfloor} }{K_{\max}}
    \;\le\;
    \Big(\tfrac{e\textcolor{black}{\lfloor} \gamma K_{\max}\textcolor{black}{\rfloor}}{K_{\max}}\Big)^{K_{\max}},
\]
we obtain
\[
    L \sum_{k=K_{\min}}^{K_{\max}} \binom{\textcolor{black}{\lfloor} \gamma K_{\max}\textcolor{black}{\rfloor} }{k}
    K_{\max} N \log_2 N
    \;=\;
    \mathcal{O}\bigl( L \gamma^{K_{\max}} K_{\max} N \log_2 N \bigr).
\]

Taking all stages together, the overall complexity of \ac{msgt} is
\[
    \mathcal{O}\Big(
        K_{\max} N \log_2 N
        + L N
        + L\!\sum_{k=K_{\min}}^{K_{\max}}
        \binom{\textcolor{black}{\lfloor} \gamma K_{\max}\textcolor{black}{\rfloor} }{k-\left|\mathcal{P}^{(DD)}\right|}
        + L \gamma^{K_{\max}} K_{\max} N \log_2 N
    \Big).
\]
As $N$ grows, this sum is dominated by the \ac{map} term, and we conclude that the complexity of the \ac{msgt} algorithm is bounded by
\[
    \textcolor{black}{\mathcal{O}\bigl(L \gamma^{K_{\max}} K_{\max} N \log_2 N\bigr)},
\]
which completes the proof.
\end{proof}

\begin{remark}\label{re:complexity}
        \textcolor{black}{
        Under the interval-sparsity formulation $K_{\min}\le K\le K_{\max}$, skipping the \ac{lva} step results in an exhaustive MAP search over all defective sets with cardinality in $[K_{\min},K_{\max}]$, whose complexity scales as
        \[
        \mathcal{O}\!\left(
        \sum_{k=K_{\min}}^{K_{\max}}
        \binom{N}{k}\, K_{\max} N \log_2 N
        \right),
        \]
        which remains combinatorial in the effective sparsity level.
        }
\end{remark}

\textcolor{black}{
Note that from Theorem~\ref{thm:complexity} and Remark~\ref{re:complexity}, it follows that the proposed \ac{msgt} algorithm performs
\[
\mathcal{O}\!\left(
({1}/{L})
\left({N}/{\textcolor{black}{\lfloor} \gamma K_{\max}\textcolor{black}{\rfloor} }\right)^{K_{\max}}
\right),
\]
times fewer computational operations compared to exhaustive MAP search over $[K_{\min},K_{\max}]$.} \textcolor{black}{The MAP baseline reported herein corresponds to MAP decoding performed over the reduced candidate set obtained after the Stage~1 pruning step, rather than exhaustive MAP decoding over all $\binom{N}{K}$ supports. This choice reflects the fact that exhaustive MAP decoding is computationally infeasible for the problem dimensions considered and would not provide a meaningful practical benchmark. The reported complexity, therefore, accounts only for the enumeration and likelihood evaluation over the retained candidate supports, enabling a fair comparison with the proposed 2SDGT algorithm.}

\textcolor{black}{The Bernoulli test-design parameter $p=\ln(2)/K_{\max}$ and the subsequent complexity and performance guarantees in this section are derived under a \emph{design sparsity interval} $(K_{\min},K_{\max})$, rather than assuming knowledge of the realized sparsity $K$. This distinction is crucial: although the realized number of defectives enters the analysis of Stage~1 through concentration bounds, Stage~2 and the associated list-decoding and counting guarantees depend \emph{only} on the design interval. Specifically, the LVA filtering rule in Algorithm~\ref{alg:cap} retains trajectories whose support size satisfies $ K_{\min} \le |V^{(\ell)}| \le \lfloor \gamma K_{\max} \rfloor$, independently of the true realization of $K$. As a result, a moderate mismatch between the realized sparsity and the design upper bound affects performance only through the tail probability $\Pr(K<K_{\min})+\Pr(K>K_{\max})$, which is explicitly separated in the overall error bound in~\eqref{eq:p_exact}. Importantly, the counting argument in Proposition~\ref{prop:gamma_cond} and the complexity bound in Theorem~\ref{thm:complexity} remain valid regardless of the exact value of $K$, as they depend solely on the candidate-set budget $\lfloor\gamma K_{\max}\rfloor$ and not on the decoder’s knowledge of the realized sparsity. This formulation enables \ac{msgt} to operate reliably under random or imperfectly known sparsity levels without redesigning the decoding pipeline, while preserving provable complexity control.}

\subsection{Discussion} \label{subsec:main_results_discussion}
To the best of our knowledge, \ac{msgt} is the first GT algorithm to effectively leverage Markovian prior statistics. Unlike numerous previous approaches, \ac{msgt} \textcolor{black}{uses the same decoding flow for any prior model, without requiring problem-specific redesign. In contrast, prior non-adaptive GT algorithms that incorporate structural assumptions—such as graph-constrained GT \cite{cheraghchi2012graph}, model-based and graph-based GT \cite{lau2022model}, and contact-based GT \cite{goenka2021contact}—must often adapt their test design or graph structure when the population model changes. \ac{msgt} instead relies only on the given initial distribution and transition matrices, while the algorithmic pipeline remains unchanged}. The algorithm offers the flexibility to be fine-tuned to optimize its performance in accordance with the available computational resources and the number of tests, $T$. The simple reduction of the search space in Stage~1 enables \ac{msgt} to handle challenging regimes with a small number of tests. Stage~2, particularly the \ac{lva} step, contributes to its high success probability. Additionally, using the parallel implementation of \ac{lva}, rather than the iterative one, keeps the complexity low \cite{seshadri1994list}. It is important to note that, as explained in \cite{seshadri1994list}, achieving optimal results is ensured only with a very large $L$, inevitably leading to complexity equivalent to MAP's. However, as we empirically demonstrate in the following section, results equivalent to MAP's can be achieved with reasonable complexity. Moreover, it is shown that \ac{msgt} addresses practical regimes, e.g., in COVID-19 detection \cite{shental2020efficient} ($T=48$ for $(N,K)=(384,4)$), in communications \cite{luo2008neighbor} ($(N,K)=(10^5, 6)$), and in GT quantizer \cite{cohen2019serial} ($(N,K)=(1024,16)$ \cite{jacques2013quantized}).

Another aspect of novelty \textcolor{black}{in this research} is the integration of the Viterbi Algorithm into the GT problem. In the context of Markovian priors, one can think of the \textcolor{black}{population sequence of items} as a sequence of observations stemming from a hidden Markov process within a given Markov model over $N$ steps. In that case, the selection of a Viterbi decoder becomes natural, offering an optimal and efficient decoding solution. However, the most likely sequence of items does not necessarily include \textcolor{black}{$K_{\max}$} defective items. Particularly, in sparse signal scenarios, which is the focus of GT, the most likely sequence typically involves the minimum number of defectives that explain the observations. As a result, \ac{va} may not necessarily detect more defective items than already known and may detect even more than \textcolor{black}{$K_{\max}$}. To address this, we employ \ac{lva}, which produces a list of the $L$ most likely sequences, such that choosing an appropriate value for $L$ guarantees a successful recovery.

\textcolor{black}{
Unlike works that assume a deterministic or perfectly known $K$, the model adopted in this work allows the realized number of defectives $K$ to be random and unknown to the decoder. To accommodate this uncertainty without redesigning the algorithmic pipeline, \ac{msgt} operates with a \emph{design interval} $(K_{\min},K_{\max})$, which is used consistently for test design, Stage~2 filtering, and candidate enumeration.
}

\textcolor{black}{
Specifically, in Stage~2, the decoder retains only trajectories whose estimated support size satisfies $K_{\min} \le |\mathbf{V}^{(l)}| \le \textcolor{black}{\lfloor} \gamma K_{\max}\textcolor{black}{\rfloor} $,
and enumerates candidate subsets of all cardinalities $k\in[K_{\min},K_{\max}]$ within these supports. This modification eliminates the need to know the exact realization of $K$ while preserving the counting-based guarantees and complexity bounds derived in this Section. The overall recovery error probability under random $K$ is controlled by conditioning on the \emph{typical event} $
\mathcal{E}_{\mathrm{typ}} \triangleq \{K_{\min} \le K \le K_{\max}\}$, and adding the probability of its complement. In particular,
\[
P_e \;\le\; P_e \big| \mathcal{E}_{\mathrm{typ}} \;+\; \Pr(K<K_{\min}) \;+\; \Pr(K>K_{\max}).
\]
}

\textcolor{black}{
When $\{U_i\}_{i=1}^N$ are generated by an ergodic finite-state Markov source (or, more generally, a strongly mixing process) \cite{cover2006elements,gallager1968information}, the random variable $K=\sum_{i=1}^N U_i$ satisfies concentration bounds analogous to Chernoff inequalities \cite{lezaud1998chernoff,paulin2015concentration}. \textcolor{black}{That is, the concentration bounds invoked below hold under the additional assumption that the $\tau$-order Markov prior is ergodic (irreducible and aperiodic), or more generally strongly mixing; this assumption is standard in correlated-source models and is used solely to control the tail behavior of $K$.} In particular, letting $\mu=\mathbb{E}[K]$, there exist constants $c_1,c_2>0$ such that for any $0<\varepsilon\le 1$, $
\Pr\!\left(K>(1+\varepsilon)\mu\right)\le c_1 \exp\!\left(-c_2\,\varepsilon^2\,\mu\right)$, and $\Pr\!\left(K<(1-\varepsilon)\mu\right)\le c_1 \exp\!\left(-c_2\,\varepsilon^2\,\mu\right)$. Therefore, choosing $K_{\min}=\lfloor(1-\varepsilon)\mu\rfloor$, and $K_{\max}=\lceil (1+\varepsilon)\mu\rceil$, one can obtain that both tail probabilities decay exponentially in $\mu$. As a result, for sufficiently sparse regimes ($\mu=\mathcal{O}(N^\theta)$ with $\theta<1$), the probability that $K$ falls outside the design interval $(K_{\min},K_{\max})$ is negligible, while \ac{msgt} retains its correctness and feasible computational complexity guarantees. \textcolor{black}{If the above ergodicity or strong mixing assumptions do not hold, alternative concentration bounds or tail estimates for dependent random variables may be employed to control $\Pr(K<K_{\min})$ and $\Pr(K>K_{\max})$, albeit potentially with weaker decay rates.}
}

\ifsingle\else
\begin{figure*}				
	\centering	
    \begin{subfigure}[b]{.31\linewidth}
		\includegraphics[width=\linewidth, height=0.53\linewidth]{overall_complexity_comparison_with_zoom_v1}
		\vspace{-0.3cm}
        \caption{\vspace{-0.15cm}Number of computational operations.}
        \label{fig:complexity:number_of_computations_2plots}
        \vspace{0.5cm}
	\end{subfigure}	
        \begin{subfigure}[b]{.31\linewidth}
		\includegraphics[width=\linewidth, height=0.458\linewidth]{phi_example_tau2.jpg}
        \vspace{-0.09cm}
        \caption{Correlation matrix $\mathbf{\Phi}_i$ for $\tau=2$.}
        \label{fig:phi_example}
        \vspace{0.35cm}
	\end{subfigure}
	\begin{subfigure}[b]{.31\linewidth}
		\includegraphics[width=\linewidth, height=0.53\linewidth]{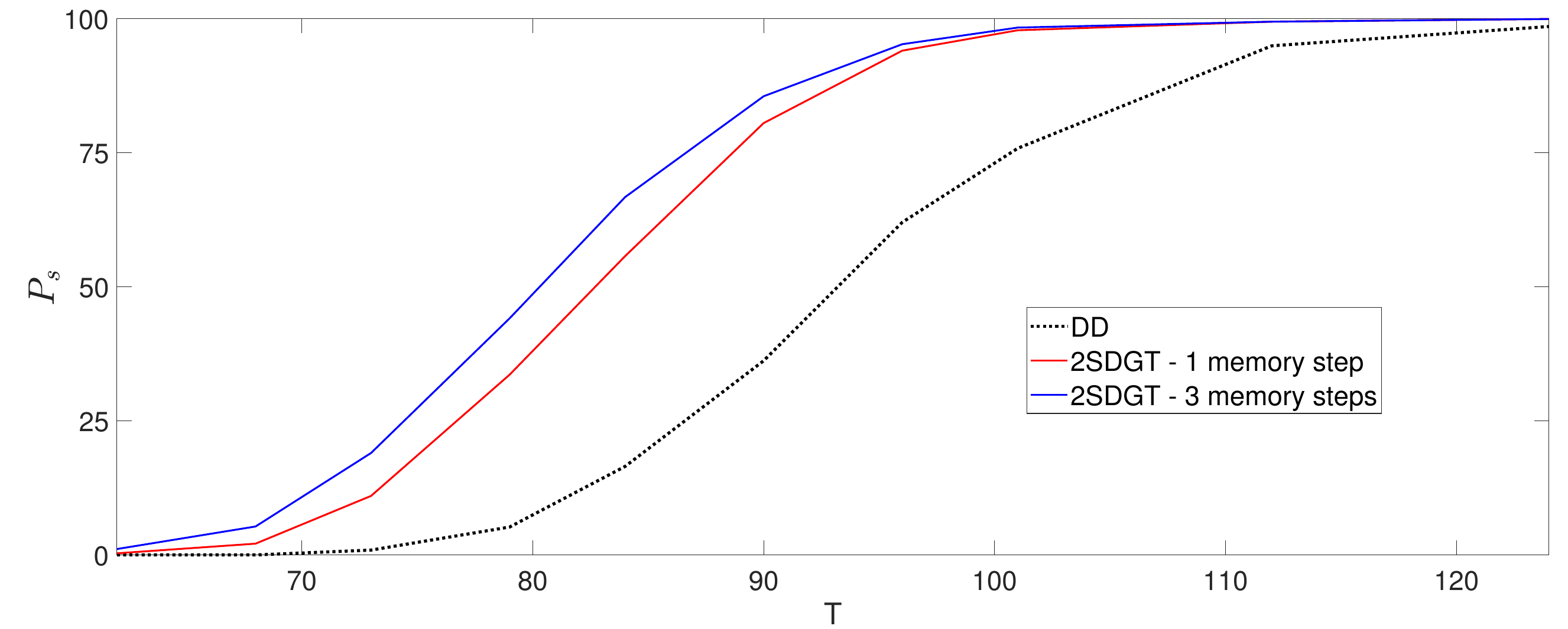}
        \caption{Long memory priors - exact recovery.}
        \label{fig:long_memory_exact}
        \vspace{0.35cm}
	\end{subfigure}
        \vspace{-0.4cm}
	\caption{(a) Number of computational operations in \ac{msgt} and \ac{map} as given in Theorem~\ref{thm:complexity} and Remark~\ref{re:complexity}, respectively. (b) Example of correlation matrix $\mathbf{\Phi}_i$. (c) Probability of success of \ac{msgt} with exact prior statistics of 3-memory-steps Markov process, and with limited prior statistics assuming the Markov process has only one memory step. $N = 1024, K=8$, 1000 iterations.}
    \vspace{-0.5cm}
	\label{fig:computations_and_complex_models}
\end{figure*}
\fi



\section{MAP Analytical Bound for GT with General Correlated Prior Statistics}\label{sec:theoretical_analysis}
\textcolor{black}{In this section, we revisit the classical information-theoretic sufficiency condition for \ac{map} decoding in group testing and adapt it to the setting considered here, which involves general correlated prior statistics and a constrained typical pattern in which only $K_{\min}\leq K \leq K_{\max}$ out of $N$ items are defective.}
This bound also applies to the \ac{lva}-based \ac{msgt} algorithm, which can approach \ac{map} decoder performance for a sufficiently large list \cite{seshadri1994list}. We calculate this bound by extending the known result of \ac{ml} decoding for \ac{gt} problems \cite{atia2012boolean}, which bounds the probability of error of a \ac{ml} decoder using a Gallager-type bound \cite{gallager1968information} for \ac{mac} channels \cite{gallager1985perspective}. Recall that a noiseless \ac{gt} setup can be considered as a noiseless \ac{mac} setting with $N$ ``users", each having a single codeword. While these type of bounds are well explored in the literature for \ac{mac} channels with correlated prior statistics \cite{slepian1973coding,zhong2006joint,zhong2007joint,campo2011random,campo2012achieving,rezazadeh2019joint,rezazadeh2019error}, they do not expand naturally to \ac{gt}, since \ac{gt} can be thought as a special case where only ``$K$" users transmit \textcolor{black}{simultaneously} at any given point. This key difference necessitates dedicated analysis as follows in this section. In particular, to gain from the correlated \textcolor{black}{prior statistical information}, unlike \ac{ml} decoders and bounds e.g., in \cite{atia2012boolean,7541823,aldridge2019group}. \textcolor{black}{That is, unlike prior works that analyze \ac{ml} decoding under a uniform prior over defective sets of fixed cardinality, our analysis considers \ac{map} decoding with arbitrary, possibly correlated priors, requiring explicit weighting of competing defective sets according to their prior probabilities.} We first introduce new notations and definitions in Subsection~\ref{subsec:def}, then provide the probability of error to bound the \ac{map} performance in Subsection~\ref{subsec:Perror}, characterizing in Subsection~\ref{subsec:bound} the maximal number of tests required to guarantee reliability as given in Theorem~\ref{direct_theorem}. In Subsection~\ref{subsec:GE}, we evaluate this bound in \textcolor{black}{Gilbert-Elliott} model \cite{gilbert1960capacity} and illustrate it for a practical regime in Section~\ref{sec_numerical_eval}. 
\textcolor{black}{Throughout those subsection, for essay of notation and to follow the existing \ac{gt} literature, we first analyze the performance of a \ac{map} decoder
\emph{conditioned on a fixed realization} $K=k$. Accordingly, all quantities below (including the index set $\allnchoosek$, the random variable $W$, and its distribution $P_W$)
are interpreted under the conditional prior $P(\cdot \mid K=k)$, and the
resulting bound applies to the conditional error probability $P_e(k)\triangleq \Pr\!\big(\widehat{\mathcal K}\neq \mathcal K \,\big|\, K=k\big)$.  Finally, in Subsection~\ref{subses:map_randomK}, we extend the analysis to support the random-$K$ model (see Section~\ref{sec_problem_form}) with tail probability (See Section~\ref{subsec:main_results_discussion}).}

\subsection{Definitions and notations}\label{subsec:def}
Let $\sbrace{n}=\cbrace{1,\ldots,n}$. Let $H( \cdot)$ and $I(\cdot;\cdot)$ denote Shannon's entropy and mutual information, respectively \cite{CovThom06}.
%
Let $\w\in\allnchoosek$ denote the index of a specific set of $K$ items out of the $N$ items. Throughout the paper, $\wstar$ denotes the index of the set of the $K$ actual infected items. Let $W$ denote the random variable which points to the index of the infected items, and has a probability distribution of $\pw{\cdot}:\allnchoosek\rightarrow\sbrace{0,1}$. In the following, $W$ might be omitted for convenience of read.
The $\w$-th set is denoted by $\sw$\footnote{\textcolor{black}{Note that there is a one-to-one mapping from $\w$ to $\sw$, so in the MAP bound analysis herein, we use both $\pw{\sw}$ and $\pw{\w}$ interchangeably.}}. Note that $\abs{\sw}=K$. The rows of $\x$ that correspond to the items in $\sw$ are denoted by $\xsw$. For a given $\w\neq\wstar$, the \ac{tp}/\ac{fn}/\ac{fp} are denoted by $\tpfull{\w}{\wstar},\fnfull{\w}{\wstar},\fpfull{\w}{\wstar}$. Note that $\abs{\fnfull{\w}{\wprime}}=\abs{\fpfull{\w}{\wprime}}$. In the following, we omit the dependency on $\w,\wstar$ and write $\xtpwprime, \xfpwprime,\xfnwprime$ instead of $\x_{TP\paren{\w,\wstar}}, \x_{FP\paren{\w,\wstar}}, \x_{FN\paren{\w,\wstar}}$ for notational simplicity. We say that an error of size $i$ happens if an output set of items $\w\neq\wstar$ is chosen and $\abs{\fpfull{\w}{\wstar}}=i$. This event is denoted by $\ei$. Assuming the event $\ei$ happened, the index of the \ac{fn} items is denoted by $j$, i.e. $j\in\allnkchoosei$. Let $\locglobsym$ denote a function that converts a $j\in\allnkchoosei$ index to a set of size $i$ of the \ac{fn} items (the dependence on $\wstar$ is omitted for notational simplicity). A decoder is denoted by $\decodersym$ and outputs some index $\w=\decode$, where the dependency on $\x$ is omitted in favor of notational simplicity. These definitions are illustrated in the following example. 

\subsubsection*{Example}\label{sub:example}
Consider a case with $N=7$ items, where $K=3$ are defective. Suppose the first $3$ items are defective, i.e. items $\cbrace{1,2,3}$ are defective. In this case, $\swstar=\cbrace{1,2,3}$. Suppose there exists another set that is more likely: $\swprime=\cbrace{1,4,5}$. In this case, $\tpfull{\wprime}{\wstar}=\cbrace{1},\fpfull{\wprime}{\wstar}=\cbrace{4,5},\locglob{2}{j}=\fnfull{\wprime}{\wstar}=\cbrace{2,3}$, where $j$ is the index of the \ac{fn} set. In this example, an error of size $i=\abs{\fpfull{\wprime}{\wstar}}=2$ occurred, i.e. the event $E_2$ happened.

\subsection{Probability of Error}\label{subsec:Perror}
The optimal \ac{map} decoding rule that minimizes the decoding error probability selects an index $\w$ that satisfies
\begin{equation}\label{eq:map_c}
\pw{\sw}P(\yt|\x_{\sw})\geq \pw{\swprime}\pr{\yt|\x_{\swprime}}; \hspace{0.2cm}\forall \wprime\neq \w.
\end{equation}
Recall that the pooling matrix $\x$ is chosen according to a Bernoulli distribution with parameter $p$ in an i.i.d. fashion. An error occurs if there exists $\wprime\neq\wstar$ such that $\pw{\wprime}\pr{\yt\lvert \xswprime}\geq\pw{\wstar}\pr{\yt\lvert\xswstar}$. This event is denoted by $\errevent$. The average probability of error, in terms of exact recovery as defined in Section~\ref{sec_problem_form}, is defined as
\begin{equation}
    \pe = \sum_{\w} \pw{\sw}\pr{\decode\neq\w\lvert \w}.\label{eq:avg_err}
\end{equation}
Let $\eij$ denote the event where an error of size $i$ at index $j\in\allnchoosei$ happens. The error probability for this event is
\begin{equation}\label{eq:P(Ei)}
\pr{\eij}=\sum_{\w\in \allnchoosek} P\left(w\right) P\left(E_{i}^{j}\lvert w\right).
\end{equation}
By plugging in \eqref{eq:P(Ei)} into \eqref{eq:avg_err}, we get
\begin{align*}
    \pe & = \sum_{i=1}^K \sum_{j = 1}^{\nkchoosei} \pr{\eij} \\
    &= \sum_{i=1}^K \sum_{j = 1}^{\nkchoosei} \sum_{\w=1}^{\nchoosek} \pw{\w} \pr{\eij\lvert W=\w}.
\end{align*}

\subsection{Sufficiency GT MAP Bound}\label{subsec:bound}

Let $(\Fset,\Tset)$ denote a partition of the defective set $S$ into disjoint sets $\Fset$ and $\Tset$, with cardinalities $i$ and $K-i$, respectively. This partition allows for categorizing error occurrences into distinct groups, where in group $i$, $K - i$ defective items are already identified correctly, and the primary error scenario arises from trying to detect the remaining $i$ defective items.
\begin{theorem}\label{direct_theorem}
    Consider a group test with a \ac{map} decoder for the population with correlated prior statistics. For $K=O(1)$, if the number of the tests satisfies
    \begin{equation*}
        T \geq (1+\varepsilon)\max_{i=1,\ldots ,K} \off{\max_{j\in\allnkchoosei}}\frac{K}{i}H\left(P_{\Fset|\Tset}\right)
    \end{equation*}
    then, as $N\rightarrow \infty$ the average error probability approaches zero.
\end{theorem}

\textcolor{black}{
\begin{remark}
While Theorem~\ref{direct_theorem} is stated for a MAP decoder conditioned on a fixed realization $K=k$, it forms the analytical foundation for the random-$K$ model adopted throughout the paper. As discussed in Section~\ref{subsec:main_results_discussion} and presented in Section~\ref{subses:map_randomK}, the overall error probability under random $K$ is controlled by conditioning on a typical sparsity interval and adding exponentially decaying tail probabilities. Finally, since the proposed \ac{msgt} algorithm approaches MAP performance for sufficiently large list sizes, this bound also characterizes the fundamental test-complexity regime in which \ac{msgt} operates effectively.
\end{remark}
}

The sufficiency proof of Theorem~\ref{direct_theorem} is twofold. We first bound in the following lemma the error of probability of the optimal \ac{map} decoder as defined in \eqref{eq:map_c}. We then use this error of probability analysis to bound the required number of tests required by any decoder. 

\begin{lemma}\label{lemma:Ei_prob_ub}
Consider a group test with a \ac{map} decoder for a population with correlated prior statistics. The error probability $P(E_{i}^{j})$ in \eqref{eq:P(Ei)} is bounded by
\begin{equation*}
    P(E_{i}^{j}) \leq 2^{-T\Big(E_{0}(\rho) - \textcolor{black}{ \left(1/\binom{K}{i}\right) E_{s,j}(\rho,P_W)}/{T}\Big)},
\end{equation*}
where the error exponent $E_0(\rho)$ is given by
\begin{small}
\begin{equation*}
\begin{split}
&E_0(\rho) =
-\log \sum_{\yt}\sum_{\xtpwprimevec}\paren{ \sum_{\xfpwprimevec}  \pr{\xfpwprimevec} \pr{\yt, \xtpwprimevec\lvert \xfpwprimevec}^{\frac{1}{1+\rho}}
}^{1+\rho},
\end{split}
\end{equation*}
\end{small}
\noindent the correlated prior statistics function of the items $\esj$ is given by
\begin{small}
\begin{equation*}
\begin{split}
\esj = \log \sum_{\tpsetijcounter}{\sbrace{\sum_{\fpsetijcounter} \pw{\wstar}^{\frac{1}{\paren{1+\rho}}}
}^{1+\rho}},
\end{split}
\end{equation*}
\end{small}
for $0\leq\rho\leq1$, where $\tpseti=\cbrace{S_{1}:\abs{S_{1}}=K-i}$,
$\fpseti{S_1}=\cbrace{S_{2}:\abs{S_{2}}=i,\;S_{2}\cap S_{1}=\emptyset}$, and
$ S_{\w}=S_{1}\cup S_{2}$.
\end{lemma}

\noindent\textit{Proof of Lemma~\ref{lemma:Ei_prob_ub}:} The proof is given in Appendix~\ref{appendix:Ei_prob_ub}. \qed


\off{\begin{lemma}\label{direct_lemma}
    For $K=O(1)$, if the number of the tests satisfies
    \begin{equation*}
        T \geq (1+\varepsilon)\max_{i=1,\ldots ,K} \max_{j\in\allnkchoosei}\frac{H\left(P_{\Fset|\Tset}\right)}{I(\xfpwprimevec;\xtpwprimevec,Y)}
    \end{equation*}
    then, as $N\rightarrow \infty$ the average error probability approaches zero.
\end{lemma}}

\begin{proof}[Proof of Theorm~\ref{direct_theorem}]
Define
\begin{equation*}\label{eq:exp_bound}
\emph{f}(\rho) = E_{o}(\rho)-\frac{E_{s,j}(\rho,P_W)}{T}.
\end{equation*}
%
Since \( 0 \leq \rho \leq 1 \) can be optimized, we aim to show that \( T f(\rho) \to \infty \) as \( N \to \infty \) for some \( \rho \) within this range. If this holds for all \( E_{i}^{j} \), then, due to the resulting exponential decay of \( P(E_{i}^{j}) \), applying a simple union bound will demonstrate that the total error probability remains small, thus completing the proof of Theorem~\ref{direct_theorem}. Since the function \( f(\rho) \) is differentiable and admits a power series expansion, we use a Taylor series expansion in the neighborhood of \( \rho = 0 \), yielding
\begin{equation*}
f(\rho) = f(0) + \rho f'(0) + \frac{\rho^2}{2} f''(\psi)
\end{equation*}
for some \( \psi \in (0, \rho) \).


For completeness, we show as in \cite{atia2012boolean} that
\ifsingle
\begin{small}
\begin{eqnarray}
&& \hspace{-0.85cm} \frac{\partial E_{o}}{\partial\rho} |_{\rho = 0}\nonumber\\
&\hspace{-0.7cm}  = &\hspace{-0.5cm} \sum_Y \sum_{\xtpwprimevec}[\sum_{\xfpwprimevec}\pr{\xfpwprimevec}\pr{Y,\xtpwprimevec|\xfpwprimevec}\log \pr{Y,\xtpwprimevec|\xfpwprimevec} - \sum_{\xfpwprimevec}\pr{\xfpwprimevec}\pr{Y,\xtpwprimevec|\xfpwprimevec} \sum_{\xfpwprimevec}P(\xfpwprimevec)p(Y,\xtpwprimevec|\xfpwprimevec)]\nonumber\\
&\hspace{-0.7cm}  = &\hspace{-0.5cm} \sum_Y \sum_{\xtpwprimevec}\sum_{\xfpwprimevec}P(\xfpwprimevec)p(Y,\xtpwprimevec|\xfpwprimevec) \log\frac{p(Y,\xtpwprimevec|\xfpwprimevec)}{\sum_{\xfpwprimevec}P(\xfpwprimevec)p(Y,\xtpwprimevec|\xfpwprimevec)}\nonumber\\
&\hspace{-0.7cm}  = &\hspace{-0.5cm} I(\xfpwprimevec;\xtpwprimevec,Y).\nonumber
\end{eqnarray}
\end{small}
\else
\begin{small}
\begin{eqnarray}
&& \hspace{-0.85cm} \frac{\partial E_{o}}{\partial\rho} |_{\rho = 0}\nonumber\\
&\hspace{-0.7cm}  = &\hspace{-0.5cm} \sum_Y \sum_{\xtpwprimevec}[\sum_{\xfpwprimevec}\pr{\xfpwprimevec}\pr{Y,\xtpwprimevec|\xfpwprimevec}\log \pr{Y,\xtpwprimevec|\xfpwprimevec}\nonumber\\
&\hspace{-0.7cm}  - &\hspace{-0.5cm} \sum_{\xfpwprimevec}\pr{\xfpwprimevec}\pr{Y,\xtpwprimevec|\xfpwprimevec} \sum_{\xfpwprimevec}P(\xfpwprimevec)p(Y,\xtpwprimevec|\xfpwprimevec)]\nonumber\\
&\hspace{-0.7cm}  = &\hspace{-0.5cm} \sum_Y \sum_{\xtpwprimevec}\sum_{\xfpwprimevec}P(\xfpwprimevec)p(Y,\xtpwprimevec|\xfpwprimevec)\nonumber\\
&& \hspace{-0.85cm} \log\frac{p(Y,\xtpwprimevec|\xfpwprimevec)}{\sum_{\xfpwprimevec}P(\xfpwprimevec)p(Y,\xtpwprimevec|\xfpwprimevec)}
= I(\xfpwprimevec;\xtpwprimevec,Y).\nonumber
\end{eqnarray}
\end{small}
\fi
And also, in Appedix~\ref{appendix:divE_s}, that
\begin{eqnarray*}
&& \hspace{-0.85cm} \frac{\partial E_{s,j}(\rho,P_W)}{\partial\rho} |_{\rho = 0}  = H(P_{\Fset|\Tset}).
\end{eqnarray*}
This is obtained directly from \cite{slepian1973noiseless} as also elaborated in \cite[Chap. 1]{rezazadeh2019error}.
Note that $E_o(0)=0$ and $E_{s,j}(0,P_W)=0$, hence we have
\ifsingle
\begin{eqnarray*}
    T\emph{f}(\rho) =& T\rho\Bigg(I\paren{\xfpwprimevec;\xtpwprimevec,Y} - \frac{H\paren{P_{\Fset|\Tset}}}{T}\Bigg)
    - \log\repsfact + T\frac{\rho^2}{2}E_o''\paren{\psi} - T\frac{\rho^2}{2}E_{s,j}''\paren{\psi},
\end{eqnarray*}
\else
\begin{eqnarray*}
    T\emph{f}(\rho) =& T\rho\Bigg(I\paren{\xfpwprimevec;\xtpwprimevec,Y} - \frac{H\paren{P_{\Fset|\Tset}}}{T}\Bigg) \nonumber\\
    -& \log\repsfact + T\frac{\rho^2}{2}E_o''\paren{\psi} - T\frac{\rho^2}{2}E_{s,j}''\paren{\psi},
\end{eqnarray*}
\fi
and, if
\begin{equation}\label{eq:T_MAP_bound}
 T \geq \paren{1+\varepsilon}\max_{i=1,\ldots ,K} \off{\max_{j\in\allnkchoosei}}\frac{H\paren{P_{\Fset|\Tset}}}{I\paren{\xfpwprimevec;\xtpwprimevec,Y}}
\end{equation}
for some $\varepsilon >0$, we have 
\ifsingle
\begin{small}
\begin{eqnarray*}
    &&\hspace{-0.5cm}T\emph{f}(\rho) \ge -\log\repsfact +T\rho\Bigg(I\paren{\xfpwprimevec;\xtpwprimevec,Y} \paren{\frac{\varepsilon}{1+\varepsilon}} + \frac{\rho}{2}\paren{E_o''(\psi) - E_{s,j}''\paren{\psi}}\Bigg) \nonumber
\end{eqnarray*}
\end{small}
\else
\begin{small}
\begin{eqnarray*}
    &&\hspace{-0.5cm}T\emph{f}(\rho) \ge -\log\repsfact \nonumber\\
    &&\hspace{-0.5cm}+T\rho\Bigg(I\paren{\xfpwprimevec;\xtpwprimevec,Y} \paren{\frac{\varepsilon}{1+\varepsilon}} + \frac{\rho}{2}\paren{E_o''(\psi) - E_{s,j}''\paren{\psi}}\Bigg) \nonumber
\end{eqnarray*}
\end{small}
\fi
Moreovver, note that $E_o''\paren{\psi}$ is negative \cite{gallager1968information}. However, it is independent of the other constants and $T$, hence choosing
\[
0 < \rho < \frac{2I\paren{\xfpwprimevec;\xtpwprimevec,Y}\paren{\frac{\varepsilon}{1+\varepsilon}}}{|E_o''\paren{\psi}-E_{s,j}''\paren{\psi}|},
\]
and as $\repsfact$ is fixed and independent of $T$, we have $T\emph{f}(\rho) \rightarrow \infty$ as $N\rightarrow \infty$. 

In practice, the expression $I(\xfpwprimevec;\xtpwprimevec,Y)$ in \eqref{eq:T_MAP_bound} is important to understand how many tests are required, yet it is not a function of the problem parameters and is bounded in \cite{atia2012boolean} and \cite{cohen2020secure} to get a better handle on $T$.
\begin{claim}[\hspace{-0.01cm}\cite{atia2012boolean}]\label{InformationClaim}
\textcolor{black}{\off{For large $K$, and }Under} a fixed input distribution for the testing matrix $\left(\frac{\ln(2)}{K},1-\frac{\ln(2)}{K}\right)$, the mutual information between $\xfpwprimevec$ and $(\xtpwprimevec,Y)$ is bounded by
\begin{equation*}
\textstyle I(\xfpwprimevec;\xtpwprimevec,Y)\geq \frac{i}{K}.
\end{equation*}
\end{claim}

\textcolor{black}{Claim~\ref{InformationClaim} relies exclusively on the independence of the Bernoulli test design used to construct the testing matrix $X$. No independence is assumed for the defective-item prior. In particular, the result holds for arbitrary (possibly correlated) priors on $\mathbf U$, including the trellis-based model considered in this paper. All prior dependence is accounted for in the MAP decoding step through the enumeration and weighting of competing defective supports, while Claim~\ref{InformationClaim} characterizes only the conditional likelihood of the test outcomes given a fixed support.}
Hence, by applying Claim~\ref{InformationClaim} in Eq.~\eqref{eq:T_MAP_bound}, we achieves the sufficiency bound on $T$ as provided in Theorem~\ref{direct_theorem},
\begin{equation*}
     (1+\varepsilon)\max_{i=1,\ldots ,K} \off{\max_{j\in\allnkchoosei}}\frac{K}{i}H\left(P_{\Fset|\Tset}\right),
\end{equation*}
such that, as $N\rightarrow \infty$ the average error probability approaches zero, and reliability is established.



\end{proof}

\ifsingle
\begin{figure*}				
	\centering	
	\begin{subfigure}[b]{.4\linewidth}
		\includegraphics[trim={1cm 0 0 0 }, width=\linewidth]{coma_average_upper_bound_vs_empirical_N10000_K15_with_lbs_v5}	
        \captionsetup{width=0.97\linewidth}
	\caption{\small Upper bound for possibly defective items after DND for $N=10000$, and $K=15$.}\label{fig:pd_after_coma_bounds}
	\end{subfigure}
	$\qquad$
	\begin{subfigure}[b]{.4\linewidth}
		\includegraphics[width=\linewidth]{boxplot_N10000_K15_v3}
        \captionsetup{width=0.99\linewidth}
		\caption{\small Lower bound for definitely defective items after DD for $N=10000$, and $K=15$.}\label{fig:dd_after_dd_bounds}
	\end{subfigure}
	\begin{subfigure}[b]{.4\linewidth}
        \vspace{0.2cm}
		\includegraphics[width=\linewidth]{gamma_vs_t_N500_K3_v6}
        \captionsetup{width=0.97\linewidth}
		\caption{\small Minimum $\gamma$ parameter to satisfies~\eqref{eq:gamma_cond} for $N=500$, and $K=3$.}
		\label{fig:gamma_vs_T}
    \vspace{-0.15cm}
	\end{subfigure}
	\caption{\small Numerical evaluation for theoretical results and bounds. The results in (a), (b), and (c) are over 1000 iterations. For ML Upper Bound (UB), $T_{ML} = (1+\epsilon)K\log_{2}N$, for any $\epsilon>0$ \cite{atia2012boolean}. In particular, $\epsilon=0.25$ in the results presented herein.}
	\vspace{-0.5cm}
	\label{fig:ConfNew}
\end{figure*}
\else\fi

\subsection{Gilbert Elliott calculations}\label{subsec:GE}
Consider a \textcolor{black}{Gilbert-Elliott} model \cite{gilbert1960capacity}. It is a Markov process with two states. In the context of \ac{gt}, as also recently considered in \cite{ravi2025fundamental}, we call them infected (state 1) and not infected (state 0). The probability of transition from state 0,1 to state 1,0 is denoted by $q,s$, respectively. The probability that the first state is in state 1 is denoted by $\pi\triangleq\pr{U_i=1}$. Unless stated otherwise, we assume the system is in steady state, that is, the probability of the system to be at state 1 at any point is given by $\pi=\frac{q}{q+s}$.

To evaluate Theorem \ref{direct_theorem}, we need to evaluate
\ifsingle
\begin{equation*}
    H\paren{P_{\Fset|\Tset}}={\sum_{\tpsetijcounter}} \paren{\sum_{\fpsetijcounter}\pw{\Tset,\Fset}\log\frac{\pr{\Tset}}{\pw{\Tset,\Fset}}},
\end{equation*}
\else
\begin{multline*}
    H\paren{P_{\Fset|\Tset}}=\\{\sum_{\tpsetijcounter}} \paren{\sum_{\fpsetijcounter}\pw{\Tset,\Fset}\log\frac{\pr{\Tset}}{\pw{\Tset,\Fset}}},
\end{multline*}
\fi
where $\pr{S_1}=\sum_{\fpsetijcounter}\pw{S_1,S_2}$.

Let $\sw=\paren{S_1,S_2}$ denote a specific set of $K$ items.  Recall that $\pw{\sw}$ calculates the probability to get $\sw$ as the infected set, given that there are $K$ infected items.

Let $\gen$ denote a \textcolor{black}{Gilbert-Elliott} process that runs for $N$ steps. Let $\cbrace{U_i}_{i=1}^{N}$ denote the states of the process during the $N$ steps. Let $\totalones\triangleq\sum_{i=1}^{N}U_{i}$. While our analysis is conditioned on the assumption that exactly $K$ items are infected, a general $\gen$ process may result in any number of infected items between 0 and $N$. We therefore analyze and simulate a special class of \textcolor{black}{Gilbert-Elliott} processes; \textcolor{black}{Gilbert-Elliott} processes that reach the infected state exactly $K$ times after $N$ steps. Such a \textcolor{black}{Gilbert-Elliott} process that is guaranteed to produce exactly $K$ infected items is denoted by $\genk$. Then $\pw{\sw}$ of $\genk$ is given as follows
\begin{eqnarray}\label{eq:P_GE}
    \pw{\sw}&=&
    \pr{\genk\text{ results in }\sw}\nonumber \\
    &\overset{(a)}{=}&\pr{\gen \text{ results in } \sw \lvert \; \totalones=K} \nonumber \\
    &\overset{(b)}{=}&\frac{\pr{\gen \text{ results in } \sw, \totalones=K}}{\pr{\totalones=K}} \nonumber \\
    &\overset{(c)}{=}&\frac{\pr{\gen \text{ results in } \sw}}{\pr{\totalones=K}},
\end{eqnarray}
where (a) follows from the definition of $\genk$, (b) from the definition of conditional probability, and (c) follows since $\sw$ contains exactly $K$ infected items.

\subsubsection{\underline{Numerator of Eq.\eqref{eq:P_GE}}}
The numerator evaluates the probability that a standard \textcolor{black}{Gilbert-Elliott} process results in a specific pattern of states, where state one is achieved only at the states in $\sw$.
Recall that \textcolor{black}{Gilbert-Elliott} is a Markov chain, and so $\pr{\sw}=\pr{U_{1}}\cdot\prod_{i=2}^{N}\pr{\ith\lvert U_{i-1}}$. This can be readily calculated using the definitions of $q,s,\pi$, as described at the beginning of this section
\[
\pr{U_1} =
\begin{cases}
    1-\pi & \text{if } 1\not\in\sw \\
    \pi & \text{if } 1\in\sw
\end{cases},
\]
and
\[
\pr{U_i\lvert U_{i-1}} =
\begin{cases}
    1-q & \text{if } i\not\in\sw, i-1\not\in\sw \\
    q & \text{if } i\in\sw, i-1\not\in\sw \\
    s & \text{if } i\not\in\sw, i-1\in\sw \\
    1-s & \text{if } i\in\sw, i-1\in\sw
\end{cases}.
\]

\subsubsection{\underline{Denominator of Eq.\eqref{eq:P_GE}}}

We create a new Markov chain that represents $\sum_{i=1}^n U_i$ as $n$ advances from $1$ to $N$. This Markov chain keeps track of how many times the \textcolor{black}{Gilbert-Elliott} process was in state 1. Hence this Markov chain has $\numstatesmc$ states. State $n,
n'\in\cbrace{0,\ldots,N}$ denotes a state where the \textcolor{black}{Gilbert-Elliott} process was in the infected state (state 1) $n$ times, and the current state is 0,1, respectively. An example of this Markov chain is shown in Fig.~\ref{fig:markov_chain}. For convenience, the states $0',N$ are also listed, even though they cannot be achieved. This Markov chain's states can be arranged in a vector, where the $j$-th state represents $\sum_{i=1}^{N} U_i=\floor*{j/2}$, and even/odd $j$ represents the current state is not/is $0$. For instance, in the example of Fig.~\ref{fig:markov_chain}, the states $6/7$ represent states $2/2'$ in the figure. The initial state vector is given as by $\paren{1-\pi,0,0, \pi, 0, \ldots, 0}$. The transition matrix is given as follows
\begin{equation}P=
\begin{bmatrix}
    1-q     & 0         & 0         & q         & 0         & 0         & \dots  & 0 \\
    0       & 0         & 0         & 0         & 0         & 0         & \dots  & 0 \\
    0       & 0         & 1-q       & 0         & 0         & q         & \dots  & 0 \\
    0       & 0         & s         & 0         & 1-s         & 0       & \dots  & 0 \\
    \vdots  & \vdots    & \vdots    & \vdots    & \vdots    & \vdots    & \ddots & \vdots \\
    0       & 0         & 0         & 0         & 0         & 0         & \dots     & 0
\end{bmatrix}\nonumber
.\end{equation}

\ifsingle
\begin{figure*}				
	\centering	
	\begin{subfigure}[b]{.47\linewidth}
		\includegraphics[width=\linewidth]{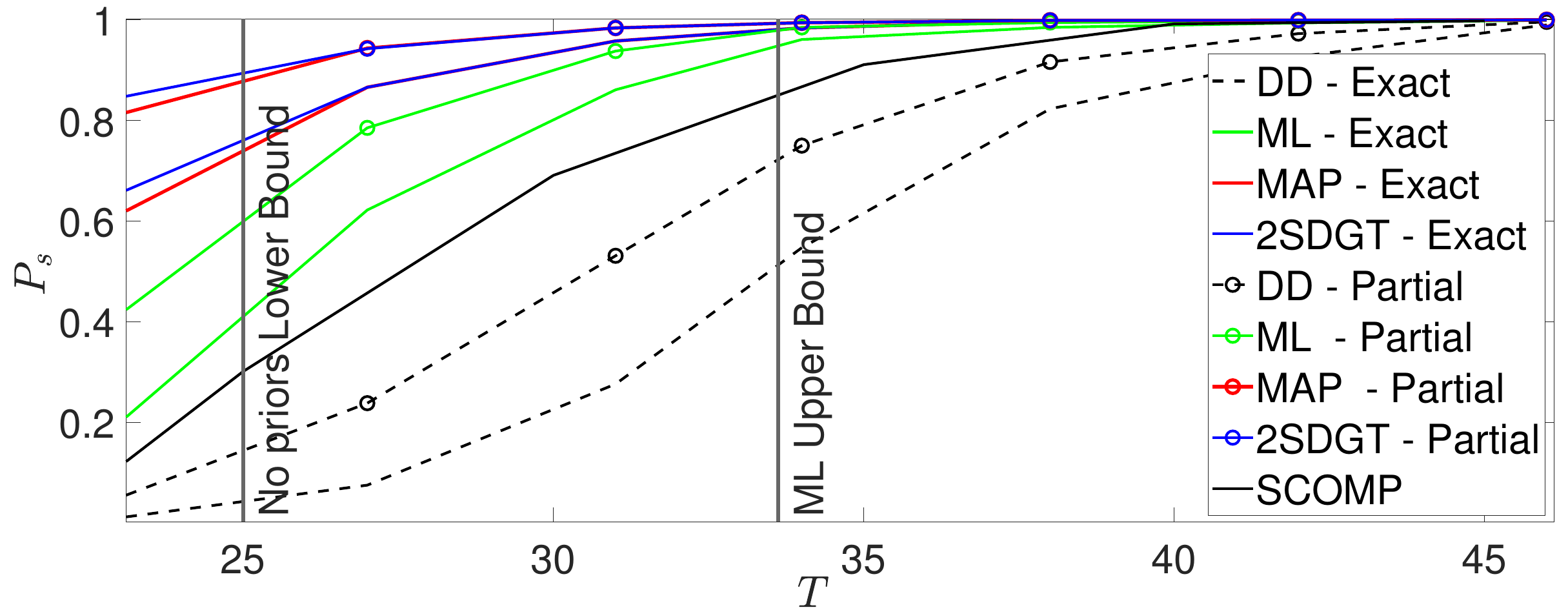}
	\caption{$N = 500$ and $K=3$.}\label{fig:benchmark_K3}
	\end{subfigure}
	$\quad$
	\begin{subfigure}[b]{.47\linewidth}
		\includegraphics[width=\linewidth]{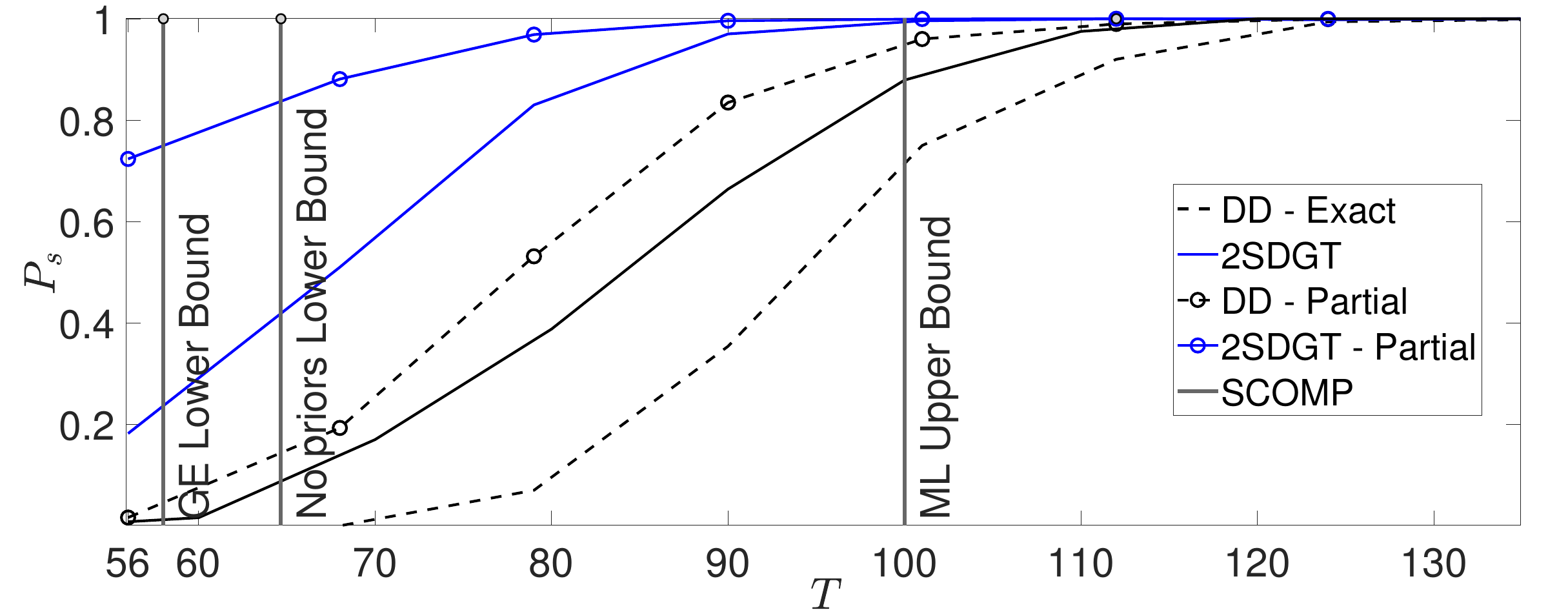}
		\caption{$N = 1024$ and $K=8$.}\label{fig:benchmark_K8}
	\end{subfigure}
	\begin{subfigure}[b]{.47\linewidth}
        \vspace{0.2cm}
		\includegraphics[width=\linewidth]{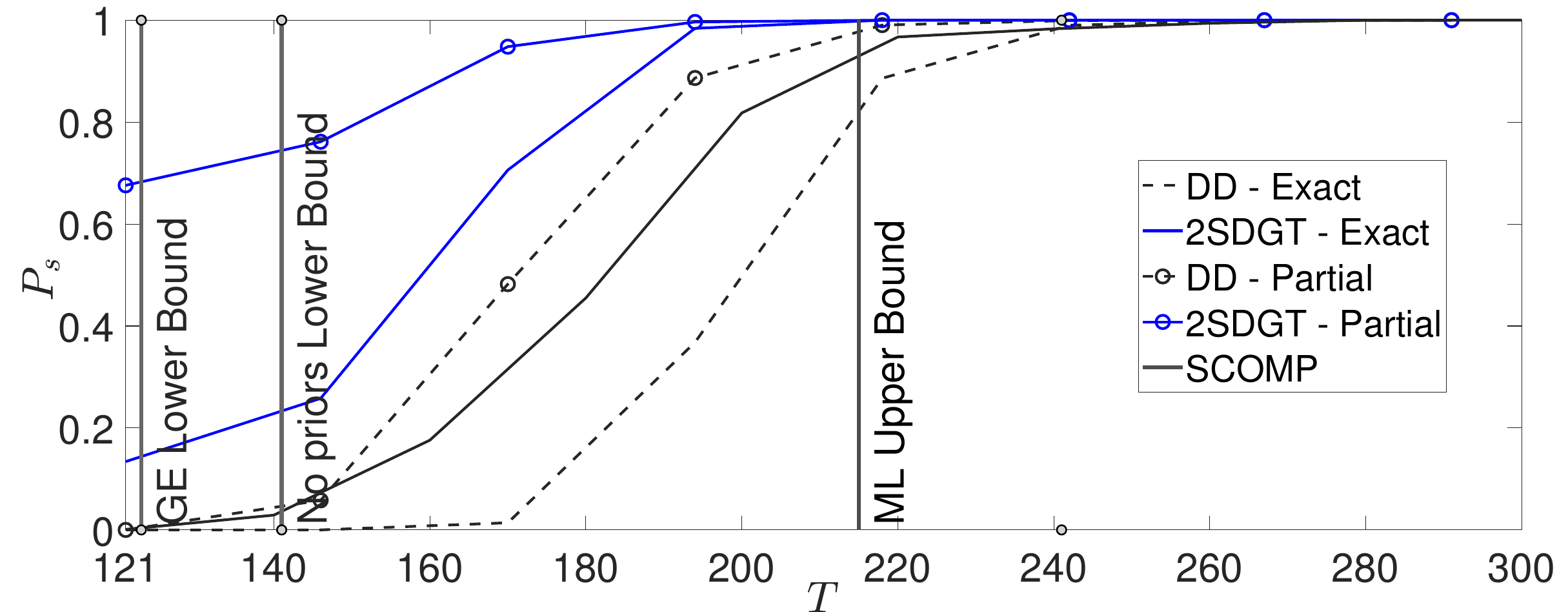}
		\caption{$N = 10000$ and $K=13$.}
		\label{fig:benchmark_K13}
		\vspace{-0.15cm}
	\end{subfigure}	
	\caption{Success probability of \ac{msgt}, MAP, ML, SCOMP, and DD over 1000 iterations. A comparison to ML and MAP is not presented in (b) and (c), as they are not feasible for populations of those sizes due to the computational complexity burden.}
	\vspace{-0.5cm}
	\label{fig:ConfNew1}
\end{figure*}
\else\fi

The probabilities to be in any state after $N$ steps is therefore given by
\begin{equation}\label{eq:ge_markov}
    \paren{1-\pi,0,0,\pi,0,\ldots,0}\cdot P^{N-1}.
\end{equation}
To conclude the calculation, the probability to have exactly $K$ infected items after $N$ steps can be obtained by summing items $2K+1$ and $2K+2$ of the vector in \eqref{eq:ge_markov}.
\begin{figure}
    \centering
    \ifsingle
    \includegraphics[width=0.6\linewidth]{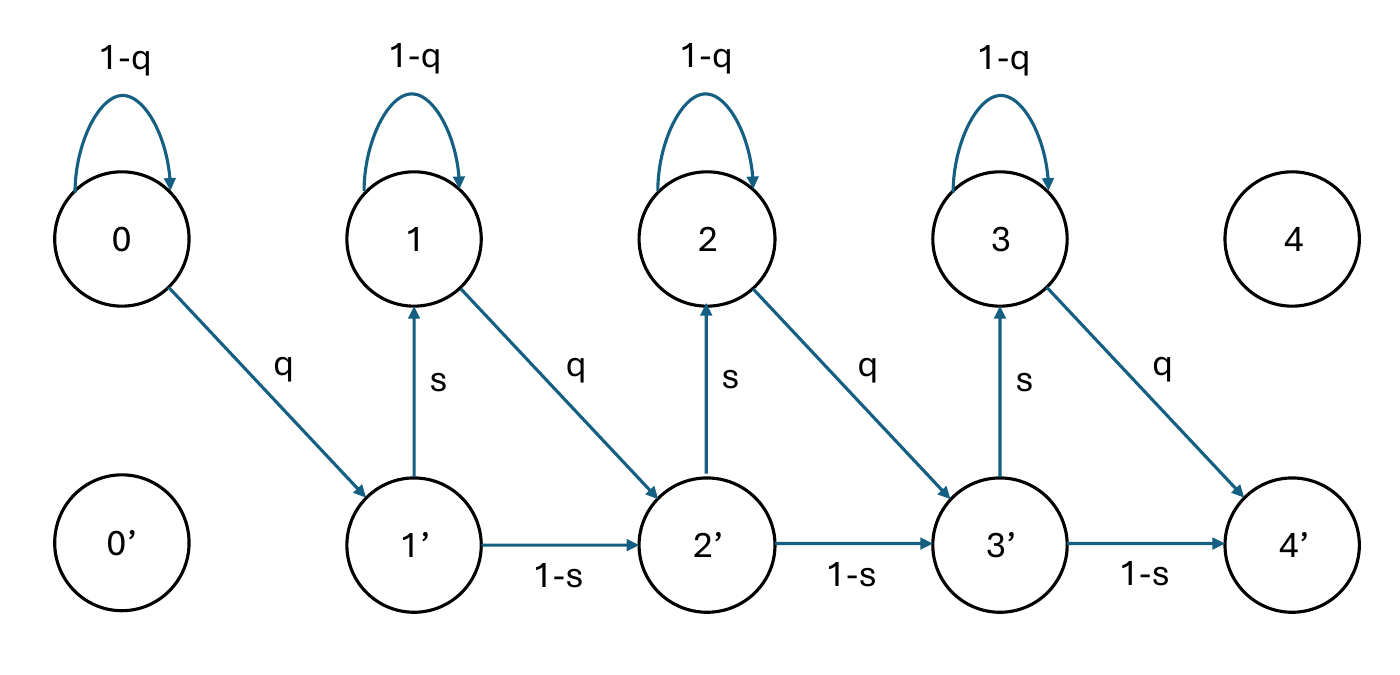}
    \else
    \includegraphics[width=1\linewidth]{markov.pdf}
    \fi
    \caption{Markov chain for the number of infected items of a \textcolor{black}{Gilbert-Elliott} model, $N=4$.}
    \label{fig:markov_chain}
\end{figure}

In the next section, we evaluate the proposed sufficient bound given in Theorem~\ref{direct_theorem} for \ac{map} decoder in the \textcolor{black}{Gilbert-Elliott} setting considered herein.

\subsection{\textcolor{black}{Extension to random-$K$ with tail probability}}\label{subses:map_randomK}
\textcolor{black}{Here we extend the analysis of the sufficient bound for fixed defectives $K=k$ to the model given in Section~\ref{sec_problem_form}, where the defectives $K$ is a random variable. The analyses in the subsection before are conditional on the realization $K=k$ and yield
a conditional MAP error bound
$\Pr(\widehat{\mathcal K}\neq\mathcal K\mid K=k)$.
By the law of total probability, the overall error probability satisfies
\[
\Pr(\widehat{\mathcal K}\neq\mathcal K)
=
\sum_{k=K_{\min}}^{K_{\max}}
\Pr(K=k)\Pr(\widehat{\mathcal K}\neq\mathcal K\mid K=k)
+
\Pr\!\big(K\notin[K_{\min},K_{\max}]\big).
\]
Consequently,
\[
\Pr(\widehat{\mathcal K}\neq\mathcal K)
\le
\max_{K_{\min}\le k\le K_{\max}}
\Pr(\widehat{\mathcal K}\neq\mathcal K\mid K=k)
+
\Pr\!\big(K\notin[K_{\min},K_{\max}]\big).
\]
Hence, if the condition of Theorem~\ref{direct_theorem} holds uniformly for all
$k\in[K_{\min},K_{\max}]$ and the tail probability
$\Pr(K\notin[K_{\min},K_{\max}])$ is negligible, then reliable recovery is guaranteed.}

\section{Numerical Evaluation}\label{sec_numerical_eval}

This section assesses the performance of the proposed \ac{msgt} algorithm by numerical study. First, in Subsection~\ref{subsec:sim_bouns}, we provide a numerical evaluation to support our theoretical results and bounds given in Section~\ref{sec_main_res}. Then, in Subsection~\ref{subsec:eval}, we contrast the performance of \ac{msgt} with those of \ac{dd}, \ac{ml}, and \ac{map} in a practical regime of $N$ and $K$. \textcolor{black}{
Although Section~\ref{sec_problem_form} adopts a random-$K$ model (with a design interval $(K_{\min},K_{\max})$), the simulations in this section are conducted for a fixed realized value of $K$ for comparison with the existing solutions in the GT literature \cite{atia2012boolean,aldridge2019group,aldridge2014group}.
This is also consistent with the theory, since the analytical results can be applied conditionally on the event $\{K= k\}$ for a typical value $k\in[K_{\min},K_{\max}]$, while the probability of atypical realizations $\{K\notin[K_{\min},K_{\max}]\}$ is controlled separately via exponential tail bounds (Chernoff-type concentration) as discussed in Section~\ref{sec_problem_form} and Section~\ref{subsec:main_results_discussion}.} To generate the correlated prior information between adjacent items, we use the special \textcolor{black}{Gilbert-Elliott} process that reaches the infected state exactly $K$ times after $N$ steps, as described and analyzed in Section \ref{subsec:GE}. The \textcolor{black}{Gilbert-Elliott} model is characterized by initial probabilities assigned to these two states, denoted as $\pi_i \in [0, 1]^2$, as well as transition probabilities between them $\mathbf{\Phi}_i \in [0, 1]^{2\times2}$. These characteristics align well with the inputs required by the \ac{msgt} algorithm. In the practical scenarios tested (e.g., in the regime of COVID-19, when the test machine can simultaneously process a fixed small number of measurements \cite{lucia2020ultrasensitive,ben2020large}, or in sparse signal recovery in signal-processing with fixed vector size of input samples \off{\cite{gilbert2008group,tan2014strong,aksoylar2016sparse,eldar2012theory,plan2013one,li2018survey})}\cite{gilbert2008group,tan2014strong,aksoylar2016sparse,eldar2012theory,li2018survey}), we show that the low computational complexity \ac{msgt} algorithm can reduce the number of \textcolor{black}{pooled tests} by at least $25\%$.

\ifsingle
\begin{figure*}				
	\centering	
    \begin{subfigure}[b]{.4\linewidth}
		\includegraphics[width=\linewidth]{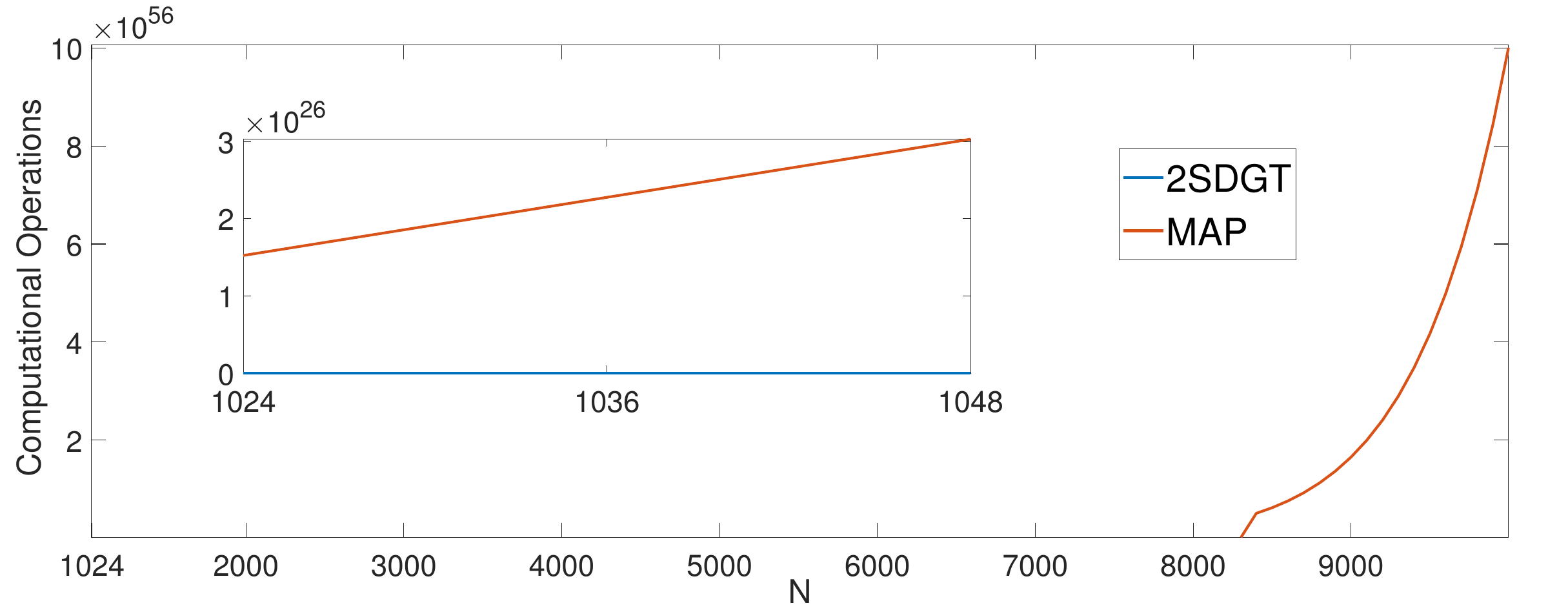}
		\vspace{-0.3cm}
        \caption{\vspace{-0.15cm}Number of computational operations.}
        \label{fig:complexity:number_of_computations_2plots}
        \vspace{0.5cm}
	\end{subfigure}	
    \off{\begin{subfigure}[b]{.31\linewidth}
		\includegraphics[width=\linewidth]{phi_example_tau2.jpg}
        \vspace{-0.09cm}
        \caption{Correlation matrix $\mathbf{\Phi}_i$ for $\tau=2$.}
        \label{fig:phi_example}
        \vspace{0.35cm}
	\end{subfigure}}
	\begin{subfigure}[b]{.4\linewidth}
		\includegraphics[width=\linewidth]{memory_time_step_comparison_1024_20paths_exact1}
        \caption{Long memory priors - exact recovery.}
        \label{fig:long_memory_exact}
        \vspace{0.35cm}
	\end{subfigure}
        \vspace{-0.4cm}
	\caption{(a) Number of computational operations in \ac{msgt} and \ac{map} as given in Theorem~\ref{thm:complexity} and Remark~\ref{re:complexity}, respectively. (b) Probability of success of \ac{msgt} with exact prior statistics of 3-memory-steps Markov process, and with limited prior statistics assuming the Markov process has only one memory step. $N = 1024, K=8$, 1000 iterations.}
    \vspace{-0.2cm}
	\label{fig:computations_and_complex_models}
\end{figure*}
\else\fi

\subsubsection{Theoretical Analysis}\label{subsec:sim_bouns}
In Fig.~\ref{fig:pd_after_coma_bounds} we show the concentration of $\left|\mathcal{P}^{(S_{1,1})} \right|$, as obtained from the simulation, along with the bound on its expectation and \textcolor{black}{the upper deviation bound} that were calculated in Theorems \ref{thm:avg_coma_bounds} and \ref{thm:upper_bound_PD_coma_bernoulli} respectively. \textcolor{black}{
Under the random-$K$ model, these theoretical bounds are interpreted conditionally on the realized value of $K$ (as done here), while the probability that $K$ falls outside the design interval $(K_{\min},K_{\max})$ is negligible by concentration and is accounted for separately in the overall error analysis.
} Note that the worst-case scenario regarding \ac{msgt} is when \ac{lva} filters the correct set of defective items. That may happen if the number of possibly defective items exceeds the threshold $\gamma K$. Since in \ac{msgt}, we only allow this deviation from the average and ignore the case of exceeding this threshold, our upper bound for the worst-case does not cover all potential realizations of $\left| \mathcal{P}^{(S_{1,1})} \right|$. Similarly, Fig.~\ref{fig:dd_after_dd_bounds} demonstrates the concentration of $\left| \mathcal{P}^{(DD)}\right|$, as acquired through simulation, and the lower bound on its expectation as given in Theorem~\ref{thm:avg_dd_bound}.
Fig.~\ref{fig:gamma_vs_T} illustrates the numerically computed lower bound for $\gamma$, derived from the inequality provided in Proposition~\ref{prop:gamma_cond}. For this simulation, we fix $N=500$, $K=3$, and calculate $\gamma$ value for a specific range of $T$ relative to the upper bound of \ac{ml}. \off{As explained in Section~\ref{sec_main_res}, our conjecture asserts that any value of $\gamma$ surpassing this lower bound guarantees that \ac{msgt} performance will be at least on par with that of \ac{ml}. Therefore, whenever computational resources allow, it is advisable to choose the value of $\gamma$ corresponding to the lower bound. This approach was followed in the subsequent simulations, and the practical outcomes presented in Subsection~\ref{subsec:eval} provide empirical support for our conjecture.}
Fig.\ref{fig:complexity:number_of_computations_2plots} compares the number of potential combinations to be examined in the \ac{map} step, with and without the execution of the \ac{lva} step in \ac{msgt}. It can be observed that the \ac{lva} step performs an extensive filtering process, which allows \ac{msgt} to remain feasible even when executing \ac{map} is no longer possible, especially in a regime below \ac{ml} upper bound.

The converse of the \ac{gt} problem with general prior statistics was developed by Li et al. \cite{li2014group} and according to which any \ac{gt} algorithm with a maximum error probability $P_e$ requires a number of tests that satisfies
\[
\textstyle T\geq \left( 1-P_e\right)H(\mathbf{U}).
\]
Using the joint entropy identity we have
\[
\textstyle \left( 1-P_e\right)H(\mathbf{U}) = \sum_{i=1}^NH\left(U_i \lvert U_1,\ldots,U_{i-1}\right).
\]

The \textcolor{black}{Gilbert-Elliott} model considered in our numerical evaluations is a stationary Markov chain with $\tau=1$. Thus,
\[
P\left(U_i \lvert U_0,\ldots,U_{i-1}\right) = P\left( U_i\lvert U_{i-1}\right), \quad \text{for} \quad  i\in\{2,\ldots,N\}.
\]
Substituting those priors, it follows that the converse of our problem is
\begin{equation}\label{eq:ge_lower_bound}
    \textstyle T \geq H(U_1) + \sum_{i=2}^NH\left(U_i\lvert U_{i-1}\right).
\end{equation}
This bound is illustrated in the practical scenarios tested in Fig.~\ref{fig:ConfNew} and Fig.~\ref{fig:ConfNew1}. In Fig~\ref{fig:bounds}, we compare the lower and the upper bounds with \textcolor{black}{Gilbert-Elliott} model. The lower and upper bounds for \ac{gt} with correlated prior statistics are given in \eqref{eq:ge_lower_bound} and in Section~\ref{subsec:GE}, respectively, while for \ac{gt} without prior statistical information, the bounds are as given in \cite{atia2012boolean}.
\begin{figure}
    \centering
    \ifsingle
    \includegraphics[width=0.49\linewidth]{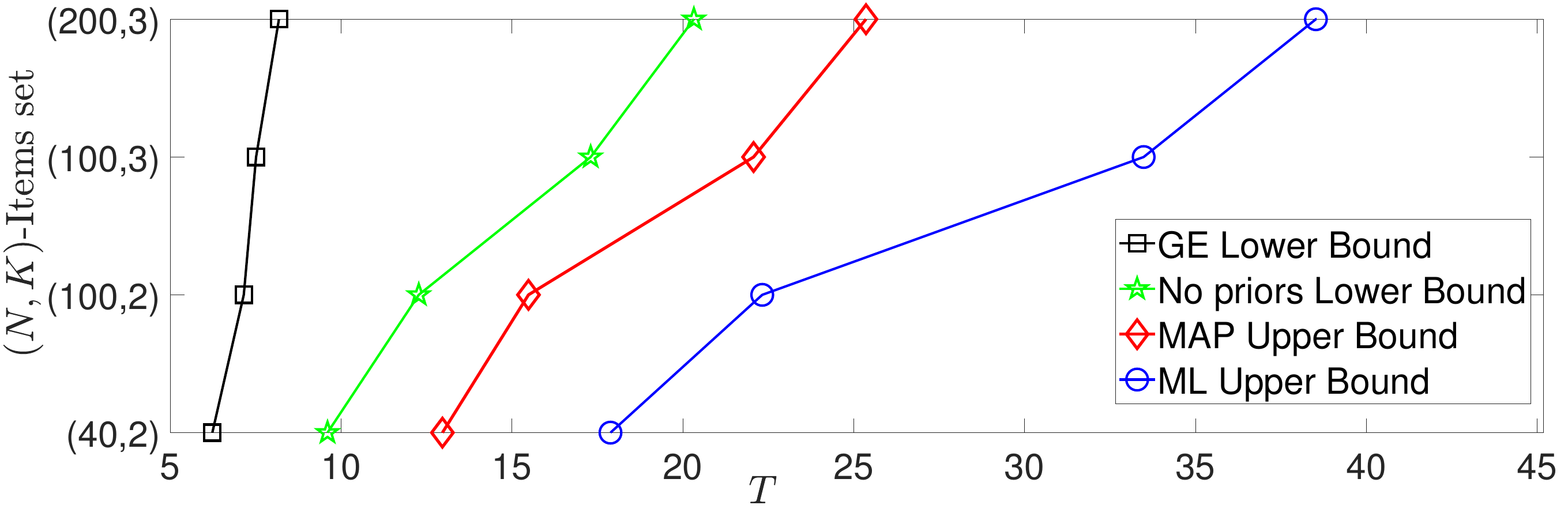}
    \includegraphics[width=0.49\linewidth]{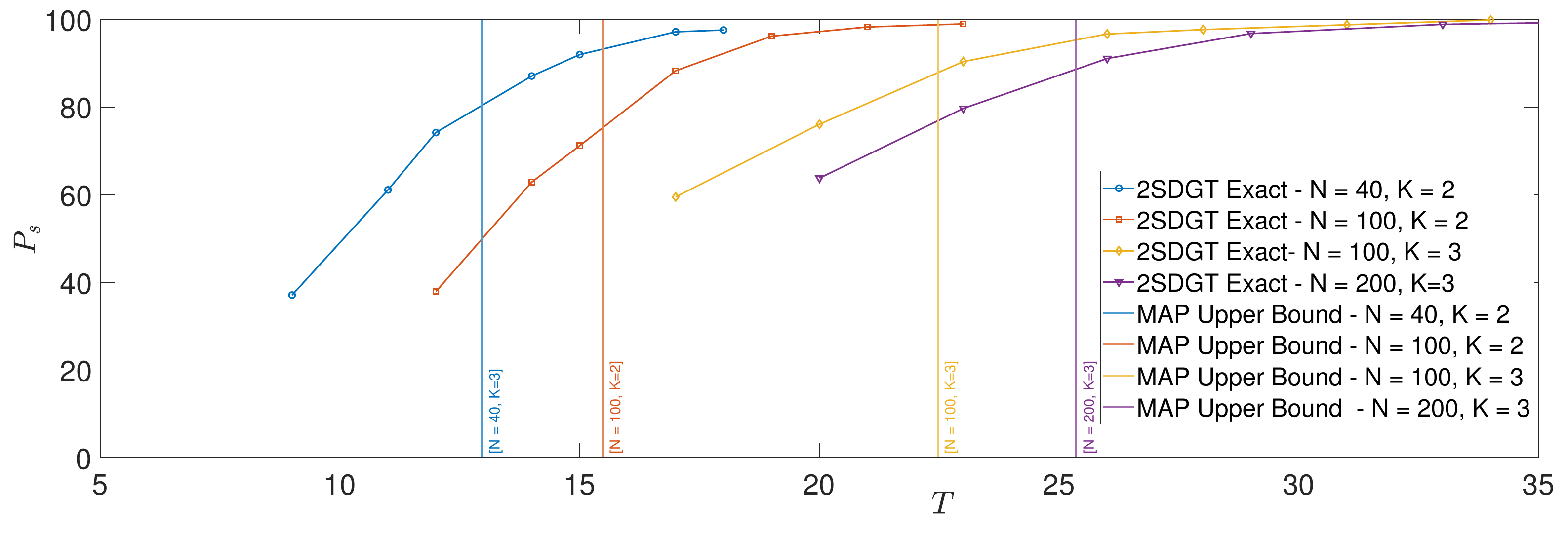}
    \else
    \includegraphics[width=1\linewidth]{bounds_v4}
    \includegraphics[width=1\linewidth]{MAP_v4}
    \fi
    \caption{\small Lower and upper bounds for \ac{gt} with GE model (left). \ac{msgt} and corresponding MAP upper bound for \ac{gt} with GE model (right). For MAP and ML upper bounds $\varepsilon=0.6$.}
    \label{fig:bounds}
    \vspace{-0.4cm}
\end{figure}

\subsubsection{Algorithm Evaluation}\label{subsec:eval}
We demonstrate the performance of \ac{msgt} using simulation. The population is sampled from \textcolor{black}{Gilbert-Elliott} model, and the regime is $\mu=\mathcal{O}\left( N^\theta \right)$ with $\theta\leq 1/3$. The \textcolor{black}{Gilbert-Elliott} parameters serve as our prior statistics, but in practice, we ignore samples where the number of defective items does not match $K$. \textcolor{black}{
Equivalently, the reported success probabilities are conditional on a fixed realized $K$, which matches the way the analytical guarantees are applied on the typical event $K\in[K_{\min},K_{\max}]$ in the random-$K$ formulation.
} In addition, although in the complexity analysis, we considered Bernoulli encoder for simplification, here we use a near-constant column weight encoder that optimizes \ac{dnd}'s performance \cite{johnson2018performance}, with $\frac{\ln{2}}{K}T$ tests sampled randomly for each item.
\textcolor{black}{In all experiments, the parameter $\gamma$ is selected to satisfy a target upper bound $\eta$ on the probability that the number of surviving candidates after Stage~1 exceeds $\gamma K_{\max}$ (see \eqref{eq:gamma_cond}), and $L=500$ was chosen empirically. We fix $\eta = 10^{-2}$, which was found to provide a good tradeoff between computational complexity and recovery performance across all tested regimes. The resulting value of $\gamma$ is computed according to the analytical bounds derived in Section~III-B and is held fixed across Monte Carlo trials.}
Fig.~\ref{fig:ConfNew1} compares \ac{msgt} to \ac{map}, \ac{ml}\textcolor{black}{, \ac{dd}, and SCOMP (an algorithm that starts with DD and sequentially marks extra items as defective, ensuring the result is a satisfying set \cite{aldridge2014group})} algorithms. We run \ac{dnd} and \ac{dd} before \ac{ml} and \ac{map} for feasible runtime and memory consumption. The population includes $N \in \{500,1024,10000\}$ items and $K \in \{3,8,13\}$ defective items, respectively, and the empirical success probability is the average over $1000$ experiments. Note that for $N=1024,10000$ it is no longer possible to compare the performance since \ac{ml} and \ac{map} become infeasible (Fig.~\ref{fig:benchmark_K8},\ref{fig:benchmark_K13}).

Finally, \ac{msgt} was also tested with more complex probabilistic models. We sample the population based on a 3-memory-steps Markov process ($\mathbf{\Phi}_i$ is a $8\times2$ matrix for all $i$). We execute \ac{msgt} using these prior statistics and also execute it with limited prior statistics, assuming that the process has only one memory step ($\mathbf{\Phi}_i$ is $2\times2$). For example of $\mathbf{\Phi}_i$ representing a 2-memory-steps process, see Fig~\ref{fig:phi_example}. The results are shown in Fig.~\ref{fig:long_memory_exact}. It is evident that utilizing the prior of long memory improves the success probability by $10\%$ in this scenario.
We do note again that for the practical regime tested as in \cite{shental2020efficient,luo2008neighbor,cohen2019serial,jacques2013quantized}, i.e., $N=1024$ and $K=8$, it is no longer possible to compare the performance since, unlike the efficient proposed \ac{msgt} algorithm, \ac{ml} and \ac{map} decoders for \ac{gt} become infeasible.\\


\vspace{-0.12cm}
\section{Conclusions and Future Work}\label{sec_discussion}
In this work, we propose a \ac{map} decoder for \ac{gt} problems with \textcolor{black}{general correlated priors that admit a trellis representation}. We analyze the performance of the \ac{map} decoders for this case and provide, to the best of our knowledge, the only analysis that goes below $T_{ML}$, as the existing solutions in the literature do not offer an efficient and general solution for this regime.
Our proposed \ac{msgt} can leverage the prior correlated information, within the \ac{lva} step, to significantly reduce the number of potential defective combinations. This approach yields an efficient computational solution that, as demonstrated in our simulation result for practical scenarios, can approach the minimum number of tests as \ac{map} algorithm.

Future work includes the study of the proposed algorithm and analytical bounds in communication systems \cite{slepian1973coding,zhong2006joint,zhong2007joint,campo2011random,campo2012achieving,rezazadeh2019joint,rezazadeh2019error} and signal processing applications \cite{zhang2011sparse,eldar2012theory,7868430}, in which correlation between different frequencies, time signals, or among different sensors, etc., can also be utilized to achieve more precise estimations.


\section{Acknowledgment}
The authors would like to thank Arezou Rezazadeh for fruitful discussions and suggestions. 

\bibliographystyle{IEEEtran}
\bibliography{ref_short}

\ifthenelse{\boolean{full_version}}{

\appendices

\section{Recovery Algorithms}\label{appendix:recovery_algo}
For the completeness of our proposed Algorithm~\ref{alg:cap}, we present here all the algorithms that we use as integral components within \ac{msgt}, and provide detailed explanations for each.

\subsection{Definitely Not Defective (DND)}
 \ac{dnd} algorithm \cite{kautz1964nonrandom, chen2008survey, chan2011non, aldridge2014group}, which is shown in Algorithm~\ref{alg:cap_coma} takes the testing matrix $\mathbf{X}$ and the results vector $Y$ as inputs. It systematically compares each column of the  $\mathbf{X}$ with $Y^T$. If a specific column contains a value of $1$ while its corresponding test result is $0$,  we say that this column cannot explain the test result. Since each column represents one item in the population, if $\mathbf{X}(i)$ cannot explain $Y^T$, then \ac{dnd} marks this item as definitely non-defective. The output is the set of all such items, denoted by $\mathcal{P}^{(DND)}$.

\begin{algorithm}
\small
\caption{Definite Not Defectives (DND) \cite{kautz1964nonrandom,chen2008survey, chan2011non,aldridge2014group}}\label{alg:cap_coma}
\hspace*{\algorithmicindent} \textbf{Input: $\mathbf{X}, Y$} \\
\hspace*{\algorithmicindent} \textbf{Output: $\mathcal{P}^{(DND)}$}
\begin{algorithmic}[1]
        \State $\mathcal{P}^{(DND)} \gets \emptyset$
        \For{$i$ s.t. $\mathbf{Y}_i=0$}
            $\mathcal{P}^{(DND)} \gets \mathcal{P}^{(DND)}  \cup \{j \lvert X_{i,j}=1\}$
        \EndFor
        \State \Return $\mathcal{P}^{(DND)}$
\end{algorithmic}
\end{algorithm}
\subsection{Definite Defectives (DD)}
\ac{dd} algorithm \cite{aldridge2014group}, as shown in Algorithm~\ref{alg:cap_dd} takes as inputs $\mathcal{P}^{(S_{1,1})}$, the testing matrix $\mathbf{X}$ and the results vector $Y$. \ac{dd}'s objective is to identify test results that can be explained by a single possibly defective item. In practice, \ac{dd} examines the positive entries of $Y$. For each positive test, if there is only one possibly defective item participating in it, \ac{dd} marks that item as definitely defective. We skip the last step of \ac{dd} in \cite{aldridge2014group}, and in our implementation, \ac{dd} returns only the set of definitely defective items, denoted by $\mathcal{P}^{(DD)}$.

\begin{algorithm}
\small
\caption{Definite Defective (DD) \cite{aldridge2014group}}\label{alg:cap_dd}
\hspace*{\algorithmicindent} \textbf{Input: $\mathbf{X}, Y, \mathcal{P}^{(S_{1,1})}$} \\
\hspace*{\algorithmicindent} \textbf{Output: $\mathcal{P}^{(DD)}$}
\begin{algorithmic}[1]
        \State $\mathcal{P}^{(DD)} \gets \emptyset$
        \For{$i$ s.t. $\mathbf{Y}_i=1$}
            \For{$j \in \mathcal{P}^{(S_{1,1})}$}
                \If{$\mathbf{X}_{i,j}=1$ and $\mathbf{X}_{i,j'}=0,$  $\forall j' \neq j\in \mathcal{P}^{(S_{1,1})}$}
                        \State $\mathcal{P}^{(DD)} \gets  \mathcal{P}^{(DD)} \cup  \left\{j\right\}$
                \EndIf
            \EndFor
        \EndFor
        \State \Return $\mathcal{P}^{(DD)}$
\end{algorithmic}
\end{algorithm}

\subsection{Priors Update}
\textcolor{black}{
The \texttt{updatePriors} procedure incorporates the information obtained in Stage~1 into the prior model.
When Stage~1 determines that an item is certainly defective or non-defective, we treat this as fixed information and update the prior accordingly. If the item lies in the first $\tau$ positions, we modify the initial distribution so that it assigns probability only to configurations consistent with the known value. If the item appears at a later position, we restrict the transition probabilities to this item to the detected state only.
}

\begin{algorithm}
\small
\caption{updatePriors}\label{alg:cap_update_priors}
\hspace*{\algorithmicindent} \textbf{Input: $\mathbf{\pi} ,\{\mathbf{\Phi}_i\}_{i=\tau+1}^N, \mathcal{P}^{(S_{1,2})}, \mathcal{P}^{(DD)}$} \\
\hspace*{\algorithmicindent} \textbf{Output: $\mathbf{\pi} ,\{\mathbf{\Phi}_i\}_{i=\tau+1}^N$}
\begin{algorithmic}[1]

\For{each $i \in \mathcal{P}^{(S_{1,2})} \cup \mathcal{P}^{(DD)}$}
    \State Let $u^\star \gets 0$ if $i \in \mathcal{P}^{(S_{1,2})}$, otherwise $u^\star \gets 1$.
    \If{$i \leq \tau$}
        \textcolor{black}{
        \[
        \pi'(u_1,\ldots,u_\tau) \gets
        \frac{\pi(u_1,\ldots,u_\tau)\,\mathbf{1}\{u_i=u^\star\}}
        {\sum_{(v_1,\ldots,v_\tau)}\pi(v_1,\ldots,v_\tau)\,\mathbf{1}\{v_i=u^\star\}}.
        \]
        }
    \Else
        \textcolor{black}{
        \[
        \Phi'_i(u_i \mid u_{i-1},\ldots,u_{i-\tau}) \gets
        \begin{cases}
        1, & u_i=u^\star,\\
        0, & u_i\neq u^\star,
        \end{cases}
        \]
        }
    \EndIf
\EndFor
\end{algorithmic}
\end{algorithm}

\subsection{List Viterbi Algorithm (LVA)}
Our variations of the \ac{lva} algorithm \cite{seshadri1994list}, outlined in Algorithm~\ref{alg:cap_lva}, is designed to return the $L$ most likely sequences, representing the estimation of the whole population, for a given $L \geq 1$. In the algorithm we suggest herein, the key differences are: (1) the population sequence in \ac{gt} replaces the time sequence in classical \ac{lva} as given in \cite{seshadri1994list},  and (2) we use the aggregated sequence likelihood instead of the general cost function presented in the original paper. In particular, in the suggested algorithm, as we traverse the trellis diagram, we are iteratively maximizing the likelihood of the sequence representing the status of the population, tested with correlated and non-uniform prior statistics, denoted by $\Psi$ (see line~\ref{alg:lva:line:rec1}).

\begin{algorithm}
\small
\caption{List Viterbi Algorithm \cite{seshadri1994list}}\label{alg:cap_lva}
\hspace*{\algorithmicindent} \textbf{Input: $L, \tau, \textcolor{black}{\mathbf{\pi}, \{\mathbf{\Phi}_i \}_{i=\tau +1}^N}$} \\
\hspace*{\algorithmicindent} \textbf{Output: ${\mathbf{Z}}$}
\begin{algorithmic}[1]

\Statex \textbf{\underline{Initialization:}}
\For{$s \gets 1$ to $2^{\tau}$}
    \State $\textcolor{black}{\Psi(\tau,s,1) \gets \pi(s)}$     \Comment{initial probability for each  state} \label{alg:lva:line:ini1}
    \State $\textcolor{black}{\xi(\tau,s,1)} \gets s$  \Comment{previous state of each state}\label{alg:lva:line:ini2}
    \State  $\textcolor{black}{\chi(\tau,s,1) \gets 1}$
    \For{$l \gets 2$ to $L$}
        \State $\Psi(\tau,s,l) \gets 0$
        \State $\textcolor{black}{\xi(\tau,s,l)} \gets s$
        \State \textcolor{black}{$\chi(\tau,s,l) \gets 1$}
    \EndFor
\EndFor

\Statex \textbf{\underline{Recursion:}}
\For{$n \gets \tau+1$ to $N$}
  \For{$s_2 \gets 1$ to $2^{\tau}$}
    \For{$l \gets 1$ to $L$}
      \State \[
      \Psi(n,s_2,l) \;=\;
      \max\nolimits^{(l)}_{\substack{s_1\in\{1,\dots,2^\tau\}\\ k\in\{1,\dots,L\}}}
      \Big\{\Psi(n-1,s_1,k)\;\Phi_n(s_1,s_2)\Big\}
      \]\label{alg:lva:line:rec1}
      \[
      (\,s_1^\ast,\;k^\ast\,) \;=\;
      \arg\max\nolimits^{(l)}_{\substack{s_1\in\{1,\dots,2^\tau\}\\ k\in\{1,\dots,L\}}}
      \Big\{\Psi(n-1,s_1,k)\;\Phi_n(s_1,s_2)\Big\}
      \]
      \State $\xi(n,s_2,l)\gets s_1^\ast$ \label{alg:lva:line:rec2}
      \State $\chi(n,s_2,l)\gets k^\ast$ \label{alg:lva:line:rec3}
    \EndFor
  \EndFor
\EndFor

\Statex \textbf{\underline{Backtracking:}}
\For{$l \gets 1$ to $L$}
  \State $Z_{l,N} \gets \arg\max_{s} \Psi(N,s,l)$ \label{alg:lva:line:back1}
  \State $l_N \gets \textcolor{black}{l}$
  \For{$n \gets N-1$ to \textcolor{black}{$\tau$}}
    \State $Z_{l,n} \gets \xi\big(n+1,\;Z_{l,n+1},\;l_{n+1}\big)$ \label{alg:lva:line:back2}
    \State $l_n \gets \chi\big(n+1,\;Z_{l,n+1},\;l_{n+1}\big)$
  \EndFor
\EndFor
\State \Return $\mathbf{Z}$
\end{algorithmic}
\end{algorithm}

This algorithm operates in three primary steps: first initialization of the setup, then recursion using a trellis structure to compute probabilities for all possible sequences while eliminating unlikely ones in each step, and finally, backtracking to reconstruct the most probable sequences.
The algorithm inputs are the number of the most likely list $L$, the number of memory steps to consider, $\tau$, and the prior statistics $\left\{\mathbf{\pi}_i\right\}_{i=1}^N, \{\mathbf{\Phi}_i \}_{i=1}^N$.

Let $\Psi\in [0,1]^{N\times 2^\tau \times L}$ denote the probabilities of the $L$ most likely states along the trellis. Let $\xi \in \{1,...,2^{\tau}\}^{N\times 2^\tau \times L}$ denote the corresponding previous state of each state along the trellis, and let $\chi \in \{1,...,L\}^{N\times 2^\tau \times L}$ denote the corresponding rank of the current state, among the $L$ most likely options.
In the initialization stage, we fill the given initial probabilities $\left\{\mathbf{\pi}_i\right\}_{i=1}^N$  in the corresponding entries of $\Psi$, and we set each state to be its previous state (lines~\ref{alg:lva:line:ini1}-\ref{alg:lva:line:ini2}).

In the recursion stage, we iterate from the second item to the last, and for each item, compute all the $2^{\tau} \times 2^{\tau}$ possible transition probabilities between all the possible states he could be in and the states of its predecessor. For each item, state, and rank, we set in $\Psi$ the $L$-most likely probabilities of the overall sequences from the first item until the current item and state (line~\ref{alg:lva:line:rec1}), and we set in $\xi$ the corresponding previous state and in $\chi$ the corresponding rank (lines~\ref{alg:lva:line:rec2}-\ref{alg:lva:line:rec3}). Here, $\max^{(l)}$ denotes the $l$-largest value.

In the backtracking stage, we identify the $L$-most likely sequences based on the probabilities in the entries of the $\Psi$ that correspond to the last item (line~\ref{alg:lva:line:back1}), and then backtrack the steps of these sequences using the information in $\xi$ (line~\ref{alg:lva:line:back2}). The algorithm returns a list of these sequences, denoted by
$\mathbf{Z} \in \left\{1,...,2^{\tau}\right\}^{L \times N}$.
If $\tau>1$, and a further processing is applied to map the states $\{1,...,2^{\tau}\}$ to the binary states $\{ 0,1\}$, representing ``defective'' and ``non-defective''.

\subsection{\textcolor{black}{Candidate Enumeration via \texttt{getAllCombinations}}}

\textcolor{black}{In Stage~2 of \ac{msgt} (Algorithm~\ref{alg:cap}), after obtaining a list of candidate population trajectories from the \ac{lva}, we must extract candidate defective sets to be evaluated by the \ac{map} estimator. This task is performed by the procedure \texttt{getAllCombinations}, as shown in Algorithm~\ref{alg:getAllCombinations}. Given a set of indices $\mathbf{V}^{(l)} \subseteq \mathcal{N}$ corresponding to items estimated as defective in the $l$-th \ac{lva} trajectory, the goal of \texttt{getAllCombinations} is to enumerate all subsets whose cardinality lies within a prescribed range.
This enables \ac{msgt} to handle uncertainty in the true number of defectives while keeping the candidate list size computationally feasible.}

\textcolor{black}{
Under the random-$K$ model in Section~\ref{sec_problem_form}, the algorithm does not assume knowledge of the realized $K$.
Instead, Stage~2 operates with a design interval $(K_{\min},K_{\max})$, and \texttt{getAllCombinations} enumerates all subsets whose cardinality satisfies $ K_{\min} \;\le\; |c| \;\le\; K_{\max}, c \subseteq \mathbf{V}^{(l)}$. This modification ensures robustness to fluctuations in the realized number of defectives while preserving correctness on the typical event
$\{K\in[K_{\min},K_{\max}]\}$. We note that the procedure is purely combinatorial and does not rely on probabilistic assumptions and that the upper bound $\textcolor{black}{\lfloor} \gamma K_{\max}\textcolor{black}{\rfloor} $ ensures that the total number of candidate sets remains manageable, as analyzed in Section~\ref{sec_main_res}. In practice, if $|\mathbf{V}^{(l)}| < K_{\min}$, the procedure returns the empty set, and the corresponding trajectory is discarded.}

\begin{algorithm}
\small
\caption{\texttt{getAllCombinations}}\label{alg:getAllCombinations}
\hspace*{\algorithmicindent} \textbf{Input: $\mathbf{V}$, $K_{\min}$, $K_{\max}$} \\
\hspace*{\algorithmicindent} \textbf{Output: $\mathcal{C}$}
\begin{algorithmic}[1]
    \State $\mathcal{C} \gets \emptyset$
    \For{$k = K_{\min}$ to $K_{\max}$}
        \For{each subset $c \subseteq \mathbf{V}$ such that $|c| = k$}
            \State $\mathcal{C} \gets \mathcal{C} \cup \{c\}$
        \EndFor
    \EndFor
    \State \Return $\mathcal{C}$
\end{algorithmic}
\end{algorithm}

\subsection{Maximum A Posteriori (MAP)}
The \ac{map} estimator \cite{gallager1968information}, as shown in Algorithm~\ref{alg:cap_map}, returns the set $c\in \mathbf{C}$ with the highest maximum a posteriori probability among all the sets, that is, the set $c$ that obtains the maximum $P\left( Y \lvert \mathbf{X},c\right) P\left( c\right)$. In this expression, the first probability represents the likelihood of obtaining the results $Y$ from a group test using the testing matrix $\mathbf{X}$ and the given set $c$ as the defective set. If the set $c$ and the given testing matrix $\mathbf{X}$ cannot explain for $Y$, then this probability is equal to zero. The second probability corresponds to the prior probability of the overall defective set $c$, calculated using $\textcolor{black}{\mathbf{\pi}, \{\mathbf{\Phi}_i \}_{i=\tau+1}^N}$. The \ac{map} algorithm returns the estimated defective set, denoted by $\hat{\mathcal{K}}$.
\begin{algorithm}
\small
\caption{Maximum A Posteriori (MAP) \cite{gallager1968information}}\label{alg:cap_map}
\hspace*{\algorithmicindent} \textbf{Input: $\mathbf{X}, Y, \mathbf{C}, \textcolor{black}{\mathbf{\pi}, \{\mathbf{\Phi}_i \}_{i=\tau+1}^N}$} \\
\hspace*{\algorithmicindent} \textbf{Output: $\hat{\mathcal{K}}$}
\begin{algorithmic}[1]
        \State \Return $\arg\max_{c\in \mathbf{C}}{ P\left( Y \lvert \mathbf{X},c\right) P\left( c\right) }$
\end{algorithmic}
\end{algorithm}

}{} 

\section{Proof of Theorem~\ref{thm:avg_dd_bound}}\label{sec:proof:avg_dd_bound}

\textcolor{black}{Let $M_0$ denote the number of negative tests, and let $G$ denote the number of non-defective items that never appear in any negative test.}
\textcolor{black}{
We follow the $(M_0,G)$ conditioning framework for \ac{dd} given in \cite[Appendix A.C]{aldridge2014group}.
Throughout the proof, we condition on a fixed realization of the number of defectives \textcolor{black}{$k$}; the Bernoulli design parameter is $p=\ln(2)/K_{\max}$, but the argument below holds for any $p\in(0,1)$.
}
\textcolor{black}{
For a fixed defective item $i$, let $L_i$ be the number of positive tests that: (i) contain item $i$, and (ii) contain no other item from $\mathcal{P}^{(S_{1,1})}$.}

\textcolor{black}{
We condition on $(M_0,G)=(m_0,g)$ and analyze a fixed defective item $i$.
Let $\mathsf{Pos}$ denote the event that a given test is positive, and let
$\mathsf{SoloDef}(i)$ denote the event that the test contains item $i$ and no
other defective items. Under a Bernoulli$(p)$ design with \textcolor{black}{$k$} defectives, $q_0 \triangleq \Pr(\text{test negative})=(1-p)^{\textcolor{black}{k}}$ denotes the probability that a test is negative, and $q_1 \triangleq \Pr(\mathsf{SoloDef}(i))=p(1-p)^{\textcolor{black}{k}-1}$ denotes the probability that a test contains defective $i$ and no other defectives, so that
\[
a \triangleq \Pr(\mathsf{SoloDef}(i)\mid \mathsf{Pos})
=\frac{\Pr(\mathsf{SoloDef}(i))}{\Pr(\mathsf{Pos})}
=\frac{q_1}{1-q_0},
\]
is the effective probability that a test includes item $i$ and no other defective items, given that the test is positive.
Given $G=g$, the probability that none of the $g$ surviving non-defective
candidates is included in a test is $(1-p)^g$. Hence, for each positive test,
the probability that it contains item $i$ and no other item from the Stage~1
candidate set equals $a(1-p)^g$. Since there are $T-m_0$ positive tests and the Bernoulli design is independent across tests, it follows that
\[
L_i \,\big|\, (M_0=m_0,G=g)\sim \mathrm{Bin}\!\big(T-m_0,\ a(1-p)^g\big).
\]
}
\textcolor{black}{Thus, we have}
\begin{equation*}
\textcolor{black}{P(L_i>0\mid M_0=m_0,G=g)
= 1 - P(L_i=0\mid M_0=m_0,G=g).}
\end{equation*}
\textcolor{black}{For a binomial random variable, $P(L_i=0)=(1-\theta)^{n}$ with $n=T-m_0$ and $\theta=a(1-p)^g$, as detailed in Appendix~\ref{P_eq:Li_success_cond}, we obtain
\begin{equation}
P(L_i>0\mid M_0=m_0,G=g)
\;=\;
1-\bigl(1-a(1-p)^g\bigr)^{T-m_0}.
\label{eq:Li_success_cond}
\end{equation}}
\textcolor{black}{
By symmetry of the items, the expected number of items identified by the \ac{dd} step is
\begin{equation}\label{eq:sym_Li}
\mathbb{E}\big[\,|\mathcal{P}^{(\mathrm{DD})}|\,\big|\,K=\textcolor{black}{k}\big]
\;=\;
\textcolor{black}{k}\,P(L_1>0).
\end{equation}
We now apply Lemma~\ref{lem:binomial_bound}, which states that for any $u \in [0,1]$ and integer $n \ge 1$
\[
1-(1-u)^n \ge nu-\frac{n(n-1)}{2} u^2.
\]
Applying this with $u = a(1-p)^g$ and $n = T - m_0$ to~\eqref{eq:Li_success_cond} yields
\begin{align}
P(L_1>0\mid M_0=m_0,G=g)
&\ge (T-m_0)\,a(1-p)^g
    - \frac{(T-m_0)(T-m_0-1)}{2}\,a^2(1-p)^{2g}.
\label{eq:cond_lower_bound}
\end{align}
\textcolor{black}{
Taking expectation of both sides of \eqref{eq:cond_lower_bound} with respect to $(M_0,G)$, substituting into~\eqref{eq:sym_Li}, and by linearity, we obtain
\begin{align}\label{eq:E_PDD_intermediate}
\mathbb{E}\big[\,|\mathcal{P}^{(\mathrm{DD})}|\,\big|\,K=\textcolor{black}{k}\big]
&= \textcolor{black}{k}\,\mathbb{E}\big[P(L_1>0\mid M_0,G)\big] \nonumber\\
&\ge \textcolor{black}{k}\left\{
\,\mathbb{E}\big[(T-M_0)a(1-p)^G\big]
-\,\mathbb{E}\left[\frac{(T-M_0)(T-M_0-1)}{2}\,a^2(1-p)^{2G}\right]
\right\} \nonumber\\
& = \textcolor{black}{k}\left\{
a\,\mathbb{E}\big[(T-M_0)(1-p)^G\big]
-\frac{a^2}{2}\,\mathbb{E}\big[(T-M_0)(T-M_0-1)(1-p)^{2G}\big]
\right\}.
\end{align}
}
Conditioned on $M_0 = m_0$, each non-defective item appears in none of the $m_0$ negative tests with probability $(1-p)^{m_0}$. Hence, the number $G$ of such non-defectives follows
\[
G \mid M_0 = m_0 \;\sim\; \mathrm{Bin}(N-\textcolor{black}{k},\, (1-p)^{m_0}).
\]
Let $\rho = (1-p)^{m_0}$ denote this probability. Then, the probability generating function of a binomial variable implies
\begin{equation}\label{eq:m0_1}
\mathbb{E}\big[(1-p)^G \mid M_0 = m_0 \big]
= \mathbb{E}\big[(1-p)^{G}\big]
= \left(1 - p \cdot \rho \right)^{N-\textcolor{black}{k}}
= \left(1 - p(1-p)^{m_0}\right)^{N-\textcolor{black}{k}}.
\end{equation}
Similarly, using the identity $(1-p)^{2G} = \left((1-p)^2\right)^G$, we obtain
\begin{equation}\label{eq:m0_2}
\mathbb{E}\big[(1-p)^{2G} \mid M_0 = m_0\big]
= \left(1 - \big[1 - (1-p)^2\big] \cdot \rho \right)^{N-\textcolor{black}{k}}
= \left(1 - (2p - p^2)(1-p)^{m_0}\right)^{N-\textcolor{black}{k}}.
\end{equation}
Taking expectations over $M_0$ (using \eqref{eq:m0_1} and \eqref{eq:m0_2}), we obtain
\[
\mathbb{E}\big[(T - M_0)(1-p)^G\big]
= \mathbb{E}\left[(T - M_0)\left(1 - p(1 - p)^{M_0}\right)^{N-\textcolor{black}{k}}\right],
\]
and
\[
\mathbb{E}\big[(T - M_0)(T - M_0 - 1)(1-p)^{2G}\big]
= \mathbb{E}\left[(T - M_0)(T - M_0 - 1)\left(1 - (2p - p^2)(1 - p)^{M_0}\right)^{N-\textcolor{black}{k}}\right].
\]
}
\textcolor{black}{
To simplify these expressions, we now apply the following bounds. First, for all $M_0 \ge 0$, since $(1-p)^{M_0} \in [0,1]$ and $p > 0$, we have
\[
1 - p(1-p)^{M_0} \ge 1 - p.
\]
Then, since $x \mapsto x^{N-\textcolor{black}{k}}$ is monotonically increasing on $[0,1]$, it follows that
\[
\left(1 - p(1-p)^{M_0}\right)^{N-\textcolor{black}{k}} \ge (1 - p)^{N-\textcolor{black}{k}}.
\]
Therefore,
\begin{align}
\mathbb{E}\big[(T - M_0)(1-p)^G\big]
&= \mathbb{E}\left[(T - M_0)\left(1 - p(1 - p)^{M_0}\right)^{N-\textcolor{black}{k}}\right] \notag\\
&\ge \mathbb{E}[T - M_0] \cdot (1 - p)^{N-\textcolor{black}{k}}
= T(1 - q_0)(1 - p)^{N-\textcolor{black}{k}},
\label{eq:term1_bound_expanded}
\end{align}
For the second term, we simply upper-bound everything, such that,
$(T - M_0)(T - M_0 - 1) \le T^2$ and $(1 - p)^{2G} \le 1$, so
\begin{equation}
\mathbb{E}\big[(T - M_0)(T - M_0 - 1)(1 - p)^{2G}\big] \le T^2.
\label{eq:term2_bound_expanded}
\end{equation}
Substituting~\eqref{eq:term1_bound_expanded} and~\eqref{eq:term2_bound_expanded} into~\eqref{eq:E_PDD_intermediate}, we conclude the proof.
}

\subsection{Proof of \eqref{eq:Li_success_cond}}\label{P_eq:Li_success_cond}
\textcolor{black}{
Set $n \triangleq T-m_0$ and $\theta \triangleq a(1-p)^g$.
Conditioned on $(M_0,G)=(m_0,g)$, there are exactly $n$ positive tests, and in each positive test, the event ``test witnesses defective item $i$'' occurs with
probability $\theta$, independently across tests. Therefore,
\[
L_i \,\big|\, (M_0=m_0,G=g)\sim \mathrm{Bin}(n,\theta).
\]
}
\textcolor{black}{
By definition, the event $\{L_i=0\}$ means that \emph{none} of the $n$ positive
tests is a witness test for item $i$. Since these $n$ witness events are
independent and each fails with probability $1-\theta$, we obtain
\begin{align*}
&P(L_i=0\mid M_0=m_0,G=g)\\
& \overset{(a)}{=} P(\text{no witness in test 1},\ldots,\text{no witness in test $n$}\mid M_0=m_0,G=g)\\
& \overset{(b)}{=} \prod_{t=1}^{n} P(\text{no witness in test $t$}\mid M_0=m_0,G=g)\\
& \overset{(c)}{=} \prod_{t=1}^{n} (1-\theta)\\
&= (1-\theta)^n\\
&= \bigl(1-a(1-p)^g\bigr)^{T-m_0},
\end{align*}
where (a) follows since this is the definition of $L_i=0$, (b) due to the independence across tests, and (c) as each test is a witness w.p.\ $\theta$.}
\textcolor{black}{
Finally, we have
\[
P(L_i>0\mid M_0=m_0,G=g)
=1-P(L_i=0\mid M_0=m_0,G=g)
=1-\bigl(1-a(1-p)^g\bigr)^{T-m_0},
\]
which yields \eqref{eq:Li_success_cond}.
}

\subsection{Binomial Bound}

\begin{lemma}\label{lem:binomial_bound}
Let $u \in [0,1]$ and $n \in \mathbb{N}$. Then
\[
1-(1-u)^n \;\ge\; nu - \frac{n(n-1)}{2}u^2.
\]
\end{lemma}

\begin{proof}
The function $f(u) = (1 - u)^n$ is analytic on $[0,1]$, and has a Taylor expansion
around $u=0$
\begin{equation}\label{eq:BB_eq}
(1-u)^n = 1-nu + \frac{n(n-1)}{2}u^2 + R_3(u),
\end{equation}
\textcolor{black}{
where $R_3(u)$ denotes the third-order remainder in the Taylor expansion.
For $n\ge 3$, the third derivative of $f$ is
\[
f^{(3)}(u) = -n(n-1)(n-2)(1-u)^{n-3} \le 0
\quad \text{for all } u\in[0,1].
\]
By the Lagrange form of the remainder, there exists $\xi\in(0,u)$ such that
\begin{equation}\label{eq:BB_R3}
R_3(u)=\frac{f^{(3)}(\xi)}{3!}\,u^3 \le 0.
\end{equation}
For $n=1,2$, the expansion is exact and $R_3(u)\equiv 0$.
}
\textcolor{black}{
Now, we multiply the identity in \eqref{eq:BB_eq} by $-1$ and reorder terms to obtain
\[
1-(1-u)^n
=
nu-\frac{n(n-1)}{2}u^2-R_3(u).
\]
Since $R_3(u)\le 0$ for all $u\in[0,1]$ as given in \eqref{eq:BB_R3}, we have $-R_3(u)\ge 0$, and therefore
\[
1-(1-u)^n
\ge
nu-\frac{n(n-1)}{2}u^2,
\]
}
which completes the proof.
\end{proof}

\section{Probability of Error (Proof of Lemma~\ref{lemma:Ei_prob_ub})}\label{appendix:Ei_prob_ub}

Let $i\in\sbrace{K},j\in\allnkchoosei$. Let $\calItildejdiff$ denote the set of all indices $\w$ that result in a the set of \ac{fn} of size $i$ as in $j$, i.e.
\begin{equation}
    \calItildejdiff = \cbrace{\w\in\allnchoosek :\quad \fnfull{\w}{\wstar}=\locglob{i}{j}}.
\end{equation}
Recall $\wstar$ denotes the index of the set of the $K$ actual infected items. To visualize this definition, consider the example in Section~\ref{sub:example}. For $\swstar=\cbrace{1,2,3}$, let $j$ be the index such that $\locglob{2}{j}=\cbrace{2,3}$. In this case, $\calItildejdiff$ has all indices that point to $\cbrace{1,4,5},\cbrace{1,4,6},\cbrace{1,4,7},\cbrace{1,5,6},\cbrace{1,5,7},\cbrace{1,6,7}$.

Let $\zeta_{\wprime},\wprime \in\calItildejdiff$ denote the event that $\swprime$ is more likely than $\swstar$ (see Eq.~\eqref{eq:map_c}).
We use the Union Bound to bound $\pr{\eij\lvert W=\wstar}$, such that 
\begin{align}\label{eq:i_err_cond_prob}
\pr{\eij \lvert W=\wstar} &\leq \pr{\bigcup_{\wprime \in \calItildejdiff} \errevent} \leq \sum_{\wprime \in \calItildejdiff} \pr{\zeta_{\wprime}}\nonumber\\
\end{align}

For \ac{map} decoder the probability of $\zeta_{\wprime}$ is given by
\ifsingle
\begin{small}
\begin{eqnarray}
&&\pr{\errevent} = \sum _{\xswprime: \pr{\xswprime \lvert \yt } \geq \pr{\xswstar \lvert \yt}} \pr{\mathbf{X}_{S_{\wprime}}\lvert \xswstar} \nonumber\\
&&\overset{(a)}{\leq} \sum_{\xfpwprime} \pr{\xfpwprime} \frac{ \pr{\yt\lvert \xswprime}^s \pw{\swprime}^s}{\pr{\yt\lvert \xswstar}^s \pw{\swstar}^s}\nonumber\\
&&= \sum_{\xfpwprime}
\pr{\xfpwprime}\cdot\frac{
\pr{\yt, \xtpwprime \lvert \xfpwprime}^{s} \pr{\xtpwprime\lvert \xfnwprime}^{s} \pw{\wprime}^{s}
}{
\pr{\xtpwprime\lvert \xfpwprime}^{s}  \pr{\yt, \xtpwprime\lvert \xfnwprime}^{s} \pw{\wstar}^{s}
}\nonumber\\
&&\overset{(b)}{\leq} \sum_{\xfpwprime}
\pr{\xfpwprime} \frac{
\pr{\yt, \xtpwprime\lvert \xfpwprime}^{s}\pw{\wprime}^{s}
}{
  \pr{\yt, \xtpwprime\lvert \xfnwprime}^{s} \pw{\wstar}^{s}
}, \hspace{0.2cm} \forall s>0 .\nonumber\\\label{eq:zeta_prob}
\end{eqnarray}
\end{small}
\else
\begin{small}
\begin{eqnarray}
&&\pr{\errevent} = \sum _{\xswprime: \pr{\xswprime \lvert \yt } \geq \pr{\xswstar \lvert \yt}} \pr{\mathbf{X}_{S_{\wprime}}\lvert \xswstar} \nonumber\\
&&\overset{(a)}{\leq} \sum_{\xfpwprime} \pr{\xfpwprime} \frac{ \pr{\yt\lvert \xswprime}^s \pw{\swprime}^s}{\pr{\yt\lvert \xswstar}^s \pw{\swstar}^s}\nonumber\\
&&= \sum_{\xfpwprime}
\pr{\xfpwprime}\cdot\nonumber\\
&&\hspace{0.5cm}\frac{
\pr{\yt, \xtpwprime \lvert \xfpwprime}^{s} \pr{\xtpwprime\lvert \xfnwprime}^{s} \pw{\wprime}^{s}
}{
\pr{\xtpwprime\lvert \xfpwprime}^{s}  \pr{\yt, \xtpwprime\lvert \xfnwprime}^{s} \pw{\wstar}^{s}
}\nonumber\\
&&\overset{(b)}{\leq} \sum_{\xfpwprime}
\pr{\xfpwprime} \frac{
\pr{\yt, \xtpwprime\lvert \xfpwprime}^{s}\pw{\wprime}^{s}
}{
  \pr{\yt, \xtpwprime\lvert \xfnwprime}^{s} \pw{\wstar}^{s}
}, \hspace{0.2cm} \forall s>0 .\nonumber\\\label{eq:zeta_prob}
\end{eqnarray}
\end{small}
\fi
\hspace{-0.1cm}$(a)$ follows that given $\xswstar$, $P(\xtpwprime)$ is known and then only then probability of the uncommon codewords left: $\pr{\xswprime\lvert \xswstar}=\pr{\xfpwprime}$ and from the fact that $d\triangleq \frac{\pr{\yt\lvert \xswprime}}{\pr{\yt\lvert \xswstar}}\geq 1$, so we upperbound the whole expression by multiplying each term in the sum in $d^s\geq1$ and sum over more test matrices. $(b)$ is due to the fact that $\xtpwprime$ and $\xfnwprime$ are independent, because they contain two distinct sets of codewords, as are $\xtpwprime$ and $\xfpwprime$. Therefore,
\[
\textstyle P\left(\xtpwprime\lvert \xfpwprime\right) = P\left(\xtpwprime\lvert \xfnwprime\right) = P\left(\xtpwprime\right).
\]

By substituting Eq. \eqref{eq:zeta_prob} in Eq.  \eqref{eq:i_err_cond_prob}, it follows that
\ifsingle
\begin{small}
\begin{eqnarray} \label{eq:i_err_cond_prob_deriv}
\hspace{-0.3cm}&&\hspace{-0.3cm}P \paren{\eij \lvert W=\wstar, \xswstar, \yt}  \nonumber\\
\hspace{-0.3cm}&&\hspace{-0.3cm}\leq \sum_{\wprime \in \calItildejdiff}  \sum_{\xfpwprime}
P\left(\xfpwprime\right) \frac{
\pr{\yt, \xtpwprime\lvert \xfpwprime}^s  \pw{\wprime}^s
}{
  P\left( \yt, \xtpwprime\lvert \xfnwprime\right)^s \pw{\wstar}^s
} \nonumber\\
\hspace{-0.3cm}&&\hspace{-0.3cm}= \frac{\sum_{\wprime \in \calItildejdiff} \pw{\wprime}^s}{ \pw{\wstar}^s}
 \cdot \frac{
\sum_{\xfpwprime}  P\left(\xfpwprime\right) P\left( \yt, \xtpwprime\lvert \xfpwprime\right)^s
}{
  P\left( \yt, \xtpwprime\lvert \xfnwprime\right)^s
} \nonumber\\
\hspace{-0.3cm}&&\hspace{-0.3cm}\overset{(a)}{\leq} \left[\frac{\left(\sum_{\wprime \in \calItildejdiff}  \pw{\wprime}^s \right) }{\pw{\wstar}^s} \cdot \frac{
\left( \sum_{\xfpwprime}  P\left(\xfpwprime\right) P\left( \yt, \xtpwprime\lvert \xfpwprime\right)^s
\right)
}{
  P\left( \yt, \xtpwprime\lvert \xfnwprime\right)^s
}\right]^\rho \nonumber\\
\hspace{-0.3cm}&&\hspace{-0.3cm}
= \frac{\left(\sum_{\wprime \in \calItildejdiff}   \pw{\wprime}^s \right) ^\rho}{\pw{\wstar}^{s\rho}} \cdot \frac{
\left( \sum_{\xfpwprime}  P\left(\xfpwprime\right) P\left( \yt, \xtpwprime\lvert \xfpwprime\right)^s
\right)^\rho
}{
  P\left( \yt, \xtpwprime\lvert \xfnwprime\right)^{s\rho}
},\nonumber\\
\hspace{-0.3cm}&&\hspace{-0.3cm}\qquad \forall s>0, \rho \in \left[0,1\right]
\end{eqnarray}
\end{small}
\else
\begin{small}
\begin{eqnarray} \label{eq:i_err_cond_prob_deriv}
\hspace{-0.3cm}&&\hspace{-0.3cm}P \paren{\eij \lvert W=\wstar, \xswstar, \yt}  \nonumber\\
\hspace{-0.3cm}&&\hspace{-0.3cm}\leq \sum_{\wprime \in \calItildejdiff}  \sum_{\xfpwprime}
P\left(\xfpwprime\right) \frac{
\pr{\yt, \xtpwprime\lvert \xfpwprime}^s  \pw{\wprime}^s
}{
  P\left( \yt, \xtpwprime\lvert \xfnwprime\right)^s \pw{\wstar}^s
} \nonumber\\
\hspace{-0.3cm}&&\hspace{-0.3cm}= \frac{\sum_{\wprime \in \calItildejdiff} \pw{\wprime}^s}{ \pw{\wstar}^s}
 \cdot \nonumber\\
\hspace{-0.3cm}&&\hspace{-0.3cm}\hspace{2.5cm}\frac{
\sum_{\xfpwprime}  P\left(\xfpwprime\right) P\left( \yt, \xtpwprime\lvert \xfpwprime\right)^s
}{
  P\left( \yt, \xtpwprime\lvert \xfnwprime\right)^s
} \nonumber\\
\hspace{-0.3cm}&&\hspace{-0.3cm}\overset{(a)}{\leq} \Biggl[\frac{\left(\sum_{\wprime \in \calItildejdiff}  \pw{\wprime}^s \right) }{\pw{\wstar}^s} \cdot\nonumber\\
\hspace{-0.3cm}&&\hspace{-0.3cm}\hspace{2.5cm}\frac{
\left( \sum_{\xfpwprime}  P\left(\xfpwprime\right) P\left( \yt, \xtpwprime\lvert \xfpwprime\right)^s
\right)
}{
  P\left( \yt, \xtpwprime\lvert \xfnwprime\right)^s
}\Biggr]^\rho \nonumber\\
\hspace{-0.3cm}&&\hspace{-0.3cm}
= \frac{\left(\sum_{\wprime \in \calItildejdiff}   \pw{\wprime}^s \right) ^\rho}{\pw{\wstar}^{s\rho}} \cdot \nonumber\\
\hspace{-0.3cm}&&\hspace{-0.3cm}\hspace{2.5cm}\frac{
\left( \sum_{\xfpwprime}  P\left(\xfpwprime\right) P\left( \yt, \xtpwprime\lvert \xfpwprime\right)^s
\right)^\rho
}{
  P\left( \yt, \xtpwprime\lvert \xfnwprime\right)^{s\rho}
},\nonumber\\
\hspace{-0.3cm}&&\hspace{-0.3cm}\qquad \forall s>0, \rho \in \left[0,1\right]
\end{eqnarray}
\end{small}
\fi
where $(a)$ follows from the fact that the probability $P\left(\ei\lvert w^{*}, \xswstar, \yt\right)\leq 1$, hence it holds that if it is upperbounded  by $U$, then $P\left(\ei\lvert w^{*}, \xswstar, \yt\right)\leq U^\rho$ for any $\rho \in [0,1]$. 

Using the law of total probability, we have
\ifsingle
\begin{small}
\begin{eqnarray}\label{eq:i_err_prob}
\hspace{-0.3cm}&&\hspace{-0.3cm}\pr{\eij\lvert W=\wstar} \nonumber\\
\hspace{-0.3cm}&&\hspace{-0.3cm}= \sum_{\xswstar} \sum_{\yt} \pr{\xswstar, \yt} \pr{\ei \lvert W=\wstar,\xswstar, \yt} \nonumber\\
\hspace{-0.3cm}&&\hspace{-0.3cm}= \sum_{\yt} \sum_{\xtpwprime} \sum_{\xfnwprime}  \pr{\xtpwprime, \yt\lvert \xfnwprime} \pr{\xfnwprime}\cdot \pr{\ei \lvert W=\wstar,\xswstar, \yt}\nonumber \\
\hspace{-0.3cm}&&\hspace{-0.3cm} \overset{(a)}{\leq} \frac{\paren{\sum_{\wprime \in \calItildejdiff}  \pw{\wprime}^s}^\rho}{\pw{\wstar}^{s\rho}}\cdot \sum_{\yt}\sum_{\xtpwprime} \sum_{\xfnwprime}
\pr{\yt, \xtpwprime\lvert \xfnwprime} \cdot \pr{\xfnwprime}
\frac{
\paren{\sum_{\xfpwprime}  \pr{\xfpwprime} \pr{\yt, \xtpwprime\lvert \xfpwprime}^s
}^\rho
}{
  \pr{\yt, \xtpwprime\lvert \xfnwprime}^{s\rho}
}\nonumber\\
\hspace{-0.3cm}&&\hspace{-0.3cm}\overset{(b)}{\leq}
\frac{\paren{\sum_{\wprime \in \calItildejdiff}  \pw{\wprime}^{\frac{1}{\paren{1+\rho})}}}^\rho}
{\pw{\wstar}^{\frac{\rho}{\left(1+\rho\right)}}}\cdot \sum_{\yt} \sum_{\xtpwprime} \sum_{\xfnwprime}
\pr{\yt, \xtpwprime\lvert \xfnwprime}\pr{\xfnwprime} \cdot \frac{
\left( \sum_{\xfpwprime}  \pr{\xfpwprime} \pr{\yt, \xtpwprime\lvert \xfpwprime}^{\frac{1}{1+\rho}}
\right)^\rho
}{
  \pr{\yt, \xtpwprime\lvert \xfnwprime})^{\frac{\rho}{1+\rho}}
} \nonumber\\
\hspace{-0.3cm}&&\hspace{-0.3cm} =
\frac{\paren{\sum_{\wprime \in \calItildejdiff}  \pw{\wprime}^{\frac{1}{\left(1+\rho\right)}} }^\rho}{\pw{\wstar}^{\frac{\rho}{\paren{1+\rho}}}} \cdot \sum_{\yt} \sum_{\xtpwprime}
\paren{\sum_{\xfpwprime}  \pr{\xfpwprime} \pr{\yt, \xtpwprime\lvert \xfpwprime}^{\frac{1}{1+\rho}}
}^{1+\rho} \nonumber\\
\hspace{-0.3cm}&&\hspace{-0.3cm} =
\sbrace{
\frac{\sum_{\wprime \in \calItildejdiff}  \pw{\wprime}^{\frac{1}{\paren{1+\rho}}}}{\pw{\wstar}^{\frac{1}{\paren{1+\rho}}}}
}^{\rho}\cdot \sum_{\yt} \sum_{\xtpwprime}\paren{\sum_{\xfpwprime}  \pr{\xfpwprime} \pr{\yt, \xtpwprime\lvert \xfpwprime}^{\frac{1}{1+\rho}}
}^{1+\rho}\nonumber\\
\hspace{-0.3cm}&&\hspace{-0.3cm} =
\left[
\frac{\sum_{\wprime\in\calItildejdiff} \pw{\wprime}^{\frac{1}{\paren{1+\rho}}}}{\pw{\wstar}^{\frac{1}{\paren{1+\rho}}}}
\right]^{\rho}\cdot \sum_{\yt} \sum_{\xtpwprime}\paren{\sum_{\xfpwprime}  \pr{\xfpwprime} \pr{\yt, \xtpwprime\lvert \xfpwprime}^{\frac{1}{1+\rho}}
}^{1+\rho},
\end{eqnarray}
\end{small}
\else
\begin{small}
\begin{eqnarray}\label{eq:i_err_prob}
\hspace{-0.3cm}&&\hspace{-0.3cm}\pr{\eij\lvert W=\wstar} \nonumber\\
\hspace{-0.3cm}&&\hspace{-0.3cm}= \sum_{\xswstar} \sum_{\yt} \pr{\xswstar, \yt} \pr{\ei \lvert W=\wstar,\xswstar, \yt} \nonumber\\
\hspace{-0.3cm}&&\hspace{-0.3cm}= \sum_{\yt} \sum_{\xtpwprime} \sum_{\xfnwprime}  \pr{\xtpwprime, \yt\lvert \xfnwprime} \pr{\xfnwprime}\cdot\nonumber\\
\hspace{-0.3cm}&&\hspace{-0.3cm}\hspace{0.8cm} \pr{\ei \lvert W=\wstar,\xswstar, \yt}\nonumber \\
\hspace{-0.3cm}&&\hspace{-0.3cm} \overset{(a)}{\leq} \frac{\paren{\sum_{\wprime \in \calItildejdiff}  \pw{\wprime}^s}^\rho}{\pw{\wstar}^{s\rho}}\cdot\nonumber\\
\hspace{-0.3cm}&&\hspace{-0.3cm}\hspace{0.8cm}\sum_{\yt}\sum_{\xtpwprime} \sum_{\xfnwprime}
\pr{\yt, \xtpwprime\lvert \xfnwprime} \cdot \nonumber\\
\hspace{-0.3cm}&&\hspace{-0.3cm} \hspace{1.4cm}\pr{\xfnwprime}
\frac{
\paren{\sum_{\xfpwprime}  \pr{\xfpwprime} \pr{\yt, \xtpwprime\lvert \xfpwprime}^s
}^\rho
}{
  \pr{\yt, \xtpwprime\lvert \xfnwprime}^{s\rho}
}\nonumber\\
\hspace{-0.3cm}&&\hspace{-0.3cm}\overset{(b)}{\leq}
\frac{\paren{\sum_{\wprime \in \calItildejdiff}  \pw{\wprime}^{\frac{1}{\paren{1+\rho})}}}^\rho}
{\pw{\wstar}^{\frac{\rho}{\left(1+\rho\right)}}}\cdot\nonumber\\
\hspace{-0.3cm}&&\hspace{-0.3cm}  \hspace{0.8cm} \sum_{\yt} \sum_{\xtpwprime} \sum_{\xfnwprime}
\pr{\yt, \xtpwprime\lvert \xfnwprime}\pr{\xfnwprime} \cdot \nonumber\\
\hspace{-0.3cm}&&\hspace{-0.3cm} \hspace{1.6cm}\frac{
\left( \sum_{\xfpwprime}  \pr{\xfpwprime} \pr{\yt, \xtpwprime\lvert \xfpwprime}^{\frac{1}{1+\rho}}
\right)^\rho
}{
  \pr{\yt, \xtpwprime\lvert \xfnwprime})^{\frac{\rho}{1+\rho}}
} \nonumber\\
\hspace{-0.3cm}&&\hspace{-0.3cm} =
\frac{\paren{\sum_{\wprime \in \calItildejdiff}  \pw{\wprime}^{\frac{1}{\left(1+\rho\right)}} }^\rho}{\pw{\wstar}^{\frac{\rho}{\paren{1+\rho}}}} \cdot\nonumber\\
\hspace{-0.3cm}&&\hspace{-0.3cm} \hspace{0.8cm} \sum_{\yt} \sum_{\xtpwprime}
\paren{\sum_{\xfpwprime}  \pr{\xfpwprime} \pr{\yt, \xtpwprime\lvert \xfpwprime}^{\frac{1}{1+\rho}}
}^{1+\rho} \nonumber\\
\hspace{-0.3cm}&&\hspace{-0.3cm} =
\sbrace{
\frac{\sum_{\wprime \in \calItildejdiff}  \pw{\wprime}^{\frac{1}{\paren{1+\rho}}}}{\pw{\wstar}^{\frac{1}{\paren{1+\rho}}}}
}^{\rho}\cdot\nonumber\\
\hspace{-0.3cm}&&\hspace{-0.3cm}\hspace{0.8cm} \sum_{\yt} \sum_{\xtpwprime}\paren{\sum_{\xfpwprime}  \pr{\xfpwprime} \pr{\yt, \xtpwprime\lvert \xfpwprime}^{\frac{1}{1+\rho}}
}^{1+\rho}\nonumber\\
\hspace{-0.3cm}&&\hspace{-0.3cm} =
\left[
\frac{\sum_{\wprime\in\calItildejdiff} \pw{\wprime}^{\frac{1}{\paren{1+\rho}}}}{\pw{\wstar}^{\frac{1}{\paren{1+\rho}}}}
\right]^{\rho}\cdot
 \nonumber\\
\hspace{-0.3cm}&&\hspace{-0.3cm}\hspace{0.8cm} \sum_{\yt} \sum_{\xtpwprime}\paren{\sum_{\xfpwprime}  \pr{\xfpwprime} \pr{\yt, \xtpwprime\lvert \xfpwprime}^{\frac{1}{1+\rho}}
}^{1+\rho},
\end{eqnarray}
\end{small}
\fi
in $(a)$ we substitute Eq. \eqref{eq:i_err_cond_prob_deriv}, and $(b)$ follows as $\xfnwprime$ and $\xfpwprime$ are two independent matrices of $i$ codewords with the same size and distribution. Thus, $\pr{ \xfnwprime} =\pr{\xfpwprime} $ and $\pr{\yt, \xtpwprime\lvert \xfnwprime} = \pr{\yt, \xtpwprime\lvert \xfpwprime}$, and the inequality holds for any $s>0$. In particular, for $s=1/\paren{1+\rho}$.

We first provide the following lemma to then continue with the error probability analysis using Eq.~\eqref{eq:i_err_prob}.
\begin{lemma}\label{lemma:w_split}
Let $i\in\sbrace{K},j\in\allnchoosei,\alpha\geq 0$. Then
\begin{equation}\label{eq:set_partition}
\begin{split}
\sum_{\w\in\allnchoosek} \pw{\w}^{\alpha} &=\repsfact \sum_{\tpsetijcounter}\sum_{\fpsetijcounter}\pw{\w}^{\alpha}
\end{split}
\end{equation}
where
\begin{equation}
\begin{split}
&\tpseti=\cbrace{S_{1}:\abs{S_{1}}=K-i}, \nonumber \\
&\fpseti{S_1}=\cbrace{S_{2}:\abs{S_{2}}=i,\;S_{2}\cap S_{1}=\emptyset},\\
&\hspace{-1.3cm}\text{and} \\
&\w\in\allnchoosek\; s.t.\; S_{\w}=S_{1}\cup S_{2} \nonumber.
\end{split}
\end{equation}
\end{lemma}
\begin{proof}
The \ac{lhs} of Eq.~\eqref{eq:set_partition} in Lemma~\ref{lemma:w_split} has $\nchoosek$ elements. We count the number of elements of $\sum_{\tpsetijcounter}\sum_{\fpsetijcounter}$ in the \ac{rhs}. The set $\tpseti$ contains $\abs{\tpseti}=\nchooseki$ items. For a given $\tpsetijcounter$, the set $\fpsetis$ contains $\abs{\fpsetis}=\nkichoosei$ items. The overall number of combinations is given by
\begin{eqnarray}
    &&\hspace{-0.6cm}\nchooseki\cdot\nkichoosei \nonumber \\
    &=&\frac{N!}{\fact{K-i}\fact{N-\paren{K-i}}}\cdot \frac{\fact{N-\paren{K-i}}}{i!\fact{N-K}} \nonumber \\
    &=&\frac{N!}{K!\fact{N-K}} \cdot \frac{K!}{\fact{K-i}i!} \nonumber \\
    &=&\nchoosek\cdot\kchoosei \nonumber
\end{eqnarray}
By symmetry, every index $\w$ in the \ac{rhs} of \eqref{eq:set_partition} in Lemma \ref{lemma:w_split} appears the same number of times, so an equality is achieved by dividing by the \ac{rhs} by $\kchoosei$.
\end{proof}

One way to interpret the sets $\tpseti$ and $\fpsetis$ can be thought of as sets that iterate over \ac{tp}s and \ac{fp}s and the \ac{fn} for a fixed set of \ac{fn} items. Following the example in Section~\ref{sub:example}, the set number of errors is $i=2$. In this case, $\tpseti=\cbrace{\cbrace{1},\cbrace{4},\cbrace{5},\cbrace{6},\cbrace{7}}$. For a given item $\tpsetijcounter$, $\fpsetijcounter$ runs over the remaining $i$ items that do not overlap with $S_{1}$. For instance, for $S_{1}=\cbrace{1}$, we get $\fpsetis=\cbrace{\cbrace{2,3},\ldots,\cbrace{6,7}}$. Note that a specific set $S_{\w}=S_{1}\cup S_{2}$ can be chosen $\kchoosei=3$ different times. For instance, $\sw=\cbrace{1,2,3}$ can be chosen $3$ times. Each $S_{1}\in\cbrace{\cbrace{1},\cbrace{2},\cbrace{3}}$ yields an $\fpsetijcounter$ that results in $S_{1}\cup S_{2}=\cbrace{1,2,3}$. Thus, as demonstrated in the proof of Lemma \ref{lemma:w_split}, equality in \eqref{eq:set_partition} is achieved by dividing the \ac{rhs} by  $\kchoosei=3$.

Now, by substituting Eq. \eqref{eq:i_err_prob} in Eq. \eqref{eq:P(Ei)}, it follows that
\ifsingle
\begin{small}
\begin{eqnarray*}
\hspace{-0.3cm}&&\hspace{-0.3cm} \pr{\eij} = \sum_{\wstar\in \allnchoosek} \pw{\wstar} \pr{\eij|W=\wstar}\notag\\
\hspace{-0.3cm}&&\hspace{-0.3cm} \leq  \sum_{\wstar\in\allnchoosek} \pw{\wstar} \sbrace{
\frac{\sum_{\wprime\in\calItildejdiff} \pw{\wprime}^{\frac{1}{\left(1+\rho\right)}}}{\pw{\wstar}^{\frac{1}{\left(1+\rho\right)}}}}^{\rho} \sum_{\yt} \sum_{\xtpwprime} \notag\left( \sum_{\xfpwprime}  P\left(\xfpwprime\right) P\left( \yt, \xtpwprime\lvert \xfpwprime\right)^{\frac{1}{1+\rho}}
\right)^{1+\rho}
 \notag\\
\hspace{-0.3cm}&&\hspace{-0.3cm} =  \sum_{\wstar\in\allnchoosek} \pw{\wstar}^{1-\frac{\rho}{1+\rho}} {\left[
\sum_{\wprime\in\calItildejdiff}  \pw{\wprime}^{\frac{1}{\left(1+\rho\right)}}
\right]^{\rho}} \sum_{\yt} \sum_{\xtpwprime} \notag \left( \sum_{\xfpwprime}  P\left(\xfpwprime\right) P\left( \yt, \xtpwprime\lvert \xfpwprime\right)^{\frac{1}{1+\rho}}
\right)^{1+\rho} \\
\hspace{-0.3cm}&&\hspace{-0.3cm} =  \sum_{\wstar\in\allnchoosek} \pw{\wstar}^{\frac{1}{1+\rho}} {\left[\sum_{\wprime\in\calItildejdiff}  \pw{\wprime}^{\frac{1}{\left(1+\rho\right)}}
\right]^{\rho}} \sum_{\yt} \sum_{\xtpwprime} \notag \left( \sum_{\xfpwprime}  P\left(\xfpwprime\right) P\left( \yt, \xtpwprime\lvert \xfpwprime\right)^{\frac{1}{1+\rho}}
\right)^{1+\rho}\notag\\
\hspace{-0.3cm}&&\hspace{-0.3cm} \overset{(a)}{=} \repsfact \sum_{\tpsetijcounter} \left[\sum_{\fpsetijcounter} \pw{\wstar}^{\frac{1}{1+\rho}}\right]{\left[\sum_{\wprime\in\calItildejdiff}  \pw{\wprime}^{\frac{1}{\left(1+\rho\right)}}
\right]^{\rho}}\sum_{\yt} \sum_{\xtpwprime} \left( \sum_{\xfpwprime}  P\left(\xfpwprime\right) P\left( \yt, \xtpwprime\lvert \xfpwprime\right)^{\frac{1}{1+\rho}}
\right)^{1+\rho}\\
\hspace{-0.3cm}&&\hspace{-0.3cm} \overset{(b)}{\leq}  \repsfact \sum_{\tpsetijcounter}{\sbrace{\sum_{\fpsetijcounter} \pw{\wstar}^{\frac{1}{\left(1+\rho\right)}}
}^{1+\rho}} \sum_{\yt} \sum_{\xtpwprime} \left( \sum_{\xfpwprime}  \pr{\xfpwprime} \pr{\yt, \xtpwprime\lvert \xfpwprime}^{\frac{1}{1+\rho}}
\right)^{1+\rho} \\
\hspace{-0.3cm}&&\hspace{-0.3cm}\overset{(c)}{=}
\repsfact \sum_{\tpsetijcounter}{\sbrace{\sum_{\fpsetijcounter}  \pw{\wstar}^{\frac{1}{\paren{1+\rho}}}
}^{1+\rho}} \sbrace{\sum_{\yt}\sum_{\xtpwprimevec}\paren{\sum_{\xfpwprimevec}  \pr{\xfpwprimevec} \pr{\yt, \xtpwprimevec\lvert \xfpwprimevec}^{\frac{1}{1+\rho}}
}^{1+\rho}}^{T}
\nonumber\\
\hspace{-0.3cm}&&\hspace{-0.3cm}\overset{(d)}{=}
2^{-T\left(E_0(\rho) - \frac{\repsfact E_{s,j}(\rho,P_W)}{T}\right)},
\end{eqnarray*}
\end{small}
\else
\begin{small}
\begin{eqnarray*}
\hspace{-0.3cm}&&\hspace{-0.3cm} \pr{\eij} = \sum_{\wstar\in \allnchoosek} \pw{\wstar} \pr{\eij|W=\wstar}\notag\\
\hspace{-0.3cm}&&\hspace{-0.3cm} \leq  \sum_{\wstar\in\allnchoosek} \pw{\wstar} \sbrace{
\frac{\sum_{\wprime\in\calItildejdiff} \pw{\wprime}^{\frac{1}{\left(1+\rho\right)}}}{\pw{\wstar}^{\frac{1}{\left(1+\rho\right)}}}}^{\rho}\\
\hspace{-0.3cm}&&\hspace{-0.3cm}\hspace{0.8cm} \sum_{\yt} \sum_{\xtpwprime} \notag\left( \sum_{\xfpwprime}  P\left(\xfpwprime\right) P\left( \yt, \xtpwprime\lvert \xfpwprime\right)^{\frac{1}{1+\rho}}
\right)^{1+\rho}
 \notag\\
\hspace{-0.3cm}&&\hspace{-0.3cm} =  \sum_{\wstar\in\allnchoosek} \pw{\wstar}^{1-\frac{\rho}{1+\rho}} {\left[
\sum_{\wprime\in\calItildejdiff}  \pw{\wprime}^{\frac{1}{\left(1+\rho\right)}}
\right]^{\rho}}\\
\hspace{-0.3cm}&&\hspace{-0.3cm}\hspace{0.8cm} \sum_{\yt} \sum_{\xtpwprime} \notag \left( \sum_{\xfpwprime}  P\left(\xfpwprime\right) P\left( \yt, \xtpwprime\lvert \xfpwprime\right)^{\frac{1}{1+\rho}}
\right)^{1+\rho} \\
\hspace{-0.3cm}&&\hspace{-0.3cm} =  \sum_{\wstar\in\allnchoosek} \pw{\wstar}^{\frac{1}{1+\rho}} {\left[\sum_{\wprime\in\calItildejdiff}  \pw{\wprime}^{\frac{1}{\left(1+\rho\right)}}
\right]^{\rho}}\\
\hspace{-0.3cm}&&\hspace{-0.3cm}\hspace{0.8cm} \sum_{\yt} \sum_{\xtpwprime} \notag \left( \sum_{\xfpwprime}  P\left(\xfpwprime\right) P\left( \yt, \xtpwprime\lvert \xfpwprime\right)^{\frac{1}{1+\rho}}
\right)^{1+\rho}\notag\\
\hspace{-0.3cm}&&\hspace{-0.3cm} \overset{(a)}{=} \repsfact \sum_{\tpsetijcounter} \left[\sum_{\fpsetijcounter} \pw{\wstar}^{\frac{1}{1+\rho}}\right]\\
\hspace{-0.3cm}&&\hspace{-0.3cm}\hspace{0.8cm}{\left[\sum_{\wprime\in\calItildejdiff}  \pw{\wprime}^{\frac{1}{\left(1+\rho\right)}}
\right]^{\rho}}\\
\hspace{-0.3cm}&&\hspace{-0.3cm}\hspace{0.8cm}\sum_{\yt} \sum_{\xtpwprime} \left( \sum_{\xfpwprime}  P\left(\xfpwprime\right) P\left( \yt, \xtpwprime\lvert \xfpwprime\right)^{\frac{1}{1+\rho}}
\right)^{1+\rho}\\
\hspace{-0.3cm}&&\hspace{-0.3cm} \overset{(b)}{\leq}  \repsfact \sum_{\tpsetijcounter}{\sbrace{\sum_{\fpsetijcounter} \pw{\wstar}^{\frac{1}{\left(1+\rho\right)}}
}^{1+\rho}}\\
\hspace{-0.3cm}&&\hspace{-0.3cm}\hspace{0.8cm} \sum_{\yt} \sum_{\xtpwprime} \left( \sum_{\xfpwprime}  \pr{\xfpwprime} \pr{\yt, \xtpwprime\lvert \xfpwprime}^{\frac{1}{1+\rho}}
\right)^{1+\rho} \\
\hspace{-0.3cm}&&\hspace{-0.3cm}\overset{(c)}{=}
\repsfact \sum_{\tpsetijcounter}{\sbrace{\sum_{\fpsetijcounter}  \pw{\wstar}^{\frac{1}{\paren{1+\rho}}}
}^{1+\rho}}
\\
\hspace{-0.3cm}&&\hspace{-0.3cm} \hspace{0.7cm} \sbrace{\sum_{\yt}\sum_{\xtpwprimevec}\paren{\sum_{\xfpwprimevec}  \pr{\xfpwprimevec} \pr{\yt, \xtpwprimevec\lvert \xfpwprimevec}^{\frac{1}{1+\rho}}
}^{1+\rho}}^{T}
\nonumber\\
\hspace{-0.3cm}&&\hspace{-0.3cm}\overset{(d)}{=}
2^{-T\left(E_0(\rho) - \frac{\repsfact E_{s,j}(\rho,P_W)}{T}\right)},
\end{eqnarray*}
\end{small}
\fi
\ifsingle
\hspace{-0.2cm}where (a) follows from Lemma \ref{lemma:w_split} with $\alpha=\fraconerho$, (b) follows from the fact that for a fixed
\[
\tpsetijcounter, \quad\calItildejdiff\subset\cbrace{\w=\paren{S_{1},S_{2}}:\fpsetijcounter},
\]
(c) follows since $\xfpwprimevec,\xtpwprimevec$ represent a single column of $\xfpwprime,\xtpwprime$ respectively, as the elements of $\x$ are drawn in i.i.d. fashion, and (d) follows for
\else
\hspace{-0.2cm}where (a) follows from Lemma \ref{lemma:w_split} with $\alpha=\fraconerho$, (b) follows from the fact that for a fixed $\tpsetijcounter$, $\calItildejdiff\subset\cbrace{\w=\paren{S_{1},S_{2}}:\fpsetijcounter}$, (c) follows since $\xfpwprimevec,\xtpwprimevec$ represent a single column of $\xfpwprime,\xtpwprime$ respectively, as the elements of $\x$ are drawn in i.i.d. fashion, and (d) follows for
\fi
\ifsingle
\begin{small}
\begin{equation*}\label{eq:E_0}
\begin{split}
&E_0(\rho) =-\log \sum_{\yt}\sum_{\xtpwprimevec}\paren{ \sum_{\xfpwprimevec}  \pr{\xfpwprimevec} \pr{\yt, \xtpwprimevec\lvert \xfpwprimevec}^{\frac{1}{1+\rho}}
}^{1+\rho},
\end{split}
\end{equation*}
\end{small}
\else
\begin{small}
\begin{equation*}\label{eq:E_0}
\begin{split}
&E_0(\rho) = \\
&-\log \sum_{\yt}\sum_{\xtpwprimevec}\paren{ \sum_{\xfpwprimevec}  \pr{\xfpwprimevec} \pr{\yt, \xtpwprimevec\lvert \xfpwprimevec}^{\frac{1}{1+\rho}}
}^{1+\rho},
\end{split}
\end{equation*}
\end{small}
\fi
and
\begin{small}
\begin{equation*}\label{eq:E_s}
\begin{split}
\esj = \log \sum_{\tpsetijcounter}{\sbrace{\sum_{\fpsetijcounter}  \pw{\wstar}^{\frac{1}{\paren{1+\rho}}}
}^{1+\rho}}.
\end{split}
\end{equation*}
\end{small}
This completed the error probability proof.

\section{Derivative of $\esj$}\label{appendix:divE_s}
In this section, we show the derivative of $\esj$ for $\rho=0$. This result is obtained directly from \cite{slepian1973noiseless} as also elaborated in \cite[Chap. 1]{rezazadeh2019error}.
Let $\ftilderhos\triangleq \sum_{\fpsetijcounter} \pw{\wstar}^{\frac{1}{\paren{1+\rho}}}=\sum_{\fpsetijcounter} \pw{\Tset,\Fset}^{\frac{1}{\paren{1+\rho}}}$, $A=\ftilderhos^{1+\rho}=\exp{\log\paren{{\ftilderhos^{1+\rho}}}}$. Thus,
\ifsingle
\begin{small}
\begin{eqnarray}\label{eq:div_Es}
\hspace{-0.3cm}&&\hspace{-0.3cm} \derrho{A}  = A\derrho{\log\paren{A}}\nonumber\\
\hspace{-0.3cm}&&\hspace{-0.3cm}= A\derrho{\paren{1+\rho}\log\paren{\ftilderhos}}\nonumber\\
\hspace{-0.3cm}&&\hspace{-0.3cm}= A\paren{\log\paren{\ftilderhos}+(1+\rho)\frac{\partial \ftilderhos}{\partial\rho}}\nonumber\\
\hspace{-0.3cm}&&\hspace{-0.3cm} \overset{(a)}{=}  \ftilderhos^{1+\rho}\left(\log(\ftilderhos)+(1+\rho)\frac{\partial \ftilderhos}{\partial\rho}\right)\nonumber\\
\hspace{-0.3cm}&&\hspace{-0.3cm} \overset{(b)}{=} \ftilderhos^{1+\rho}\Biggr(\log(\ftilderhos)-\frac{1}{(1+\rho)} \sum_{\fpsetijcounter} \pw{\Tset,\Fset}^{\frac{1}{\paren{1+\rho}}} \log\paren{\pw{\Tset,\Fset}}\Biggr)\nonumber\\
\hspace{-0.3cm}&&\hspace{-0.3cm} \overset{(c)}{=} \paren{\sum_{\fpsetijcounter} \pw{\Tset,\Fset}^{\frac{1}{\paren{1+\rho}}}}^{1+\rho}\Biggr(\log(\ftilderhos) -\frac{1}{(1+\rho)} \sum_{\fpsetijcounter}  \pw{\Tset,\Fset}^{\frac{1}{\paren{1+\rho}}} \log\left(\pw{\Tset,\Fset}\right) \Biggr)\nonumber\\
\hspace{-0.3cm}&&\hspace{-0.3cm} \overset{(d)}{=} \left(\sum_{\fpsetijcounter} \pw{\Tset,\Fset}^{\frac{1}{\paren{1+\rho}}}\right)^{\rho}\Biggr(\sum_{\fpsetijcounter} \pw{\Tset,\Fset}^{\frac{1}{\paren{1+\rho}}}\log\frac{\sum_{\fpsetijcounter}  \pw{\Tset,\Fset}^{\frac{1}{\paren{1+\rho}}}}{\pw{\Tset,\Fset}^{\frac{1}{\paren{1+\rho}}}} \Biggr).\nonumber\\
\hspace{-0.3cm}&&\hspace{-0.3cm}
\end{eqnarray}
\end{small}
\else
\begin{small}
\begin{eqnarray}\label{eq:div_Es}
\hspace{-0.3cm}&&\hspace{-0.3cm} \derrho{A}  = A\derrho{\log\paren{A}}= A\derrho{\paren{1+\rho}\log\paren{\ftilderhos}}\nonumber\\
\hspace{-0.3cm}&&\hspace{-0.3cm}= A\paren{\log\paren{\ftilderhos}+(1+\rho)\frac{\partial \ftilderhos}{\partial\rho}}\nonumber\\
\hspace{-0.3cm}&&\hspace{-0.3cm} \overset{(a)}{=}  \ftilderhos^{1+\rho}\left(\log(\ftilderhos)+(1+\rho)\frac{\partial \ftilderhos}{\partial\rho}\right)\nonumber\\
\hspace{-0.3cm}&&\hspace{-0.3cm} \overset{(b)}{=} \ftilderhos^{1+\rho}\Biggr(\log(\ftilderhos)\nonumber\\
\hspace{-0.3cm}&&\hspace{-0.3cm} \hspace{0.6cm}-\frac{1}{(1+\rho)} \sum_{\fpsetijcounter} \pw{\Tset,\Fset}^{\frac{1}{\paren{1+\rho}}} \log\paren{\pw{\Tset,\Fset}}\Biggr)\nonumber\\
\hspace{-0.3cm}&&\hspace{-0.3cm} \overset{(c)}{=} \paren{\sum_{\fpsetijcounter} \pw{\Tset,\Fset}^{\frac{1}{\paren{1+\rho}}}}^{1+\rho}\Biggr(\log(\ftilderhos)\nonumber\\
\hspace{-0.3cm}&&\hspace{-0.3cm}\hspace{0.6cm} -\frac{1}{(1+\rho)} \sum_{\fpsetijcounter}  \pw{\Tset,\Fset}^{\frac{1}{\paren{1+\rho}}} \log\left(\pw{\Tset,\Fset}\right) \Biggr)\nonumber\\
\hspace{-0.3cm}&&\hspace{-0.3cm} \overset{(d)}{=} \left(\sum_{\fpsetijcounter} \pw{\Tset,\Fset}^{\frac{1}{\paren{1+\rho}}}\right)^{\rho}\Biggr(\sum_{\fpsetijcounter} \nonumber\\
\hspace{-0.3cm}&&\hspace{-0.3cm}\hspace{0.6cm}\pw{\Tset,\Fset}^{\frac{1}{\paren{1+\rho}}}\log\frac{\sum_{\fpsetijcounter}  \pw{\Tset,\Fset}^{\frac{1}{\paren{1+\rho}}}}{\pw{\Tset,\Fset}^{\frac{1}{\paren{1+\rho}}}} \Biggr).\nonumber\\
\hspace{-0.3cm}&&\hspace{-0.3cm}
\end{eqnarray}
\end{small}
\fi
\hspace{-0.1cm}where (a) follows replacing $A$, (b) by derivative of $\ftilderhos$, and (c) replacing $\ftilderhos$. (d) follows by applying that $\log a-\frac{1}{\paren{1+\rho}} \log b=  \log \frac{a}{b^{\frac{1}{\paren{1+\rho}}}}$.

Now, we recall that $E_{s,j}(\rho,P_W)=\log\paren{\sum_{\fpsetijcounter}\ftilderhos^{1+\rho}}$, such that $E_{s,j}(\rho,P_W)=\log\sum_{\tpsetijcounter}A$. Hence. from Eq.\eqref{eq:div_Es}, we have
\ifsingle
\begin{small}
\begin{eqnarray}\label{eq:div_EsA}
\hspace{-2.0cm} && \hspace{-0.8cm} \frac{\partial E_{s,j}(\rho,P_W)}{\partial\rho}  = \frac{\partial\log\sum_{\tpsetijcounter}A}{\partial\rho}=\frac{\partial\sum_{\tpsetijcounter}A}{\partial\rho}\nonumber\\
\hspace{-2.5cm} &=& \hspace{-0.0cm}\frac{1}{\sum_{\tpsetijcounter}A}\sum_{\tpsetijcounter} \left(\sum_{\fpsetijcounter} \pw{\Tset,\Fset}^{\frac{1}{\paren{1+\rho}}}\right)^{\rho} \nonumber\\
&&\Biggr(\sum_{\fpsetijcounter} \pw{\Tset,\Fset}^{\frac{1}{\paren{1+\rho}}}\log\frac{\sum_{\fpsetijcounter} \pw{\Tset,\Fset}^{\frac{1}{\paren{1+\rho}}}}{\pw{\Tset,\Fset}^{\frac{1}{\paren{1+\rho}}}}\Biggr) .\nonumber\\
\hspace{-2.5cm}&&
\end{eqnarray}
\end{small}
\else
\begin{small}
\begin{eqnarray}\label{eq:div_EsA}
\hspace{-2.0cm} && \hspace{-0.8cm} \frac{\partial E_{s,j}(\rho,P_W)}{\partial\rho}  = \frac{\partial\log\sum_{\tpsetijcounter}A}{\partial\rho}=\frac{\partial\sum_{\tpsetijcounter}A}{\partial\rho}\nonumber\\
\hspace{-2.5cm} &=& \hspace{-0.0cm}\frac{1}{\sum_{\tpsetijcounter}A}\sum_{\tpsetijcounter} \left(\sum_{\fpsetijcounter} \pw{\Tset,\Fset}^{\frac{1}{\paren{1+\rho}}}\right)^{\rho} \nonumber\\
&&\Biggr(\sum_{\fpsetijcounter} \pw{\Tset,\Fset}^{\frac{1}{\paren{1+\rho}}}\nonumber\\
&&\log\frac{\sum_{\fpsetijcounter} \pw{\Tset,\Fset}^{\frac{1}{\paren{1+\rho}}}}{\pw{\Tset,\Fset}^{\frac{1}{\paren{1+\rho}}}}\Biggr) .\nonumber\\
\hspace{-2.5cm}&&
\end{eqnarray}
\end{small}
\fi
Next, setting $\rho=0$ in Eq.~\eqref{eq:div_EsA}, and since
\[
\sum_{\tpsetijcounter}A|_{\rho=0}=\sum_{\tpsetijcounter}\sum_{\fpsetijcounter}  \pw{\Tset,\Fset}=1,
\]
we have
\ifsingle
\begin{small}
\begin{eqnarray}\label{eq:div_Es_0}
&&\hspace{-0.5cm}\frac{\partial E_{s,j}(\rho,P_W)}{\partial\rho} |_{\rho=0} \nonumber\\
&=& \sum_{\tpsetijcounter} \Biggr(\sum_{\fpsetijcounter} \pw{\Tset,\Fset}\log\frac{\sum_{\fpsetijcounter}  \pw{\Tset,\Fset}}{\pw{\Tset,\Fset}}\Biggr) \nonumber\\
&=& \sum_{\tpsetijcounter} \paren{\sum_{\fpsetijcounter}  \pw{\Tset,\Fset}\log\frac{\pr{\Tset}}{\pw{\Tset,\Fset}}} \nonumber\\
&=& H\paren{P_{\Fset|\Tset}}.\nonumber
\end{eqnarray}
\end{small}
\else
\begin{small}
\begin{eqnarray}\label{eq:div_Es_0}
&&\hspace{-0.5cm}\frac{\partial E_{s,j}(\rho,P_W)}{\partial\rho} |_{\rho=0} \nonumber\\
&=& \sum_{\tpsetijcounter} \Biggr(\sum_{\fpsetijcounter} \pw{\Tset,\Fset}\nonumber\\
&&\log\frac{\sum_{\fpsetijcounter}  \pw{\Tset,\Fset}}{\pw{\Tset,\Fset}}\Biggr) \nonumber\\
&=& \sum_{\tpsetijcounter} \paren{\sum_{\fpsetijcounter}  \pw{\Tset,\Fset}\log\frac{\pr{\Tset}}{\pw{\Tset,\Fset}}} \nonumber\\
&=& H\paren{P_{\Fset|\Tset}}.\nonumber
\end{eqnarray}
\end{small}
\fi
\end{document}